\documentclass[11pt]{article}

\usepackage[T1]{fontenc}
\usepackage[a4paper,margin=1in]{geometry}
\usepackage{amsthm}
\usepackage{newtxtext}
\usepackage[subscriptcorrection]{newtxmath}

\usepackage{amsmath,mathtools,bbm}
\newtheorem{theorem}{Theorem}
\newtheorem{lemma}{Lemma}
\newtheorem{corollary}{Corollary}
\newtheorem{proposition}{Proposition}
\theoremstyle{definition}
\newtheorem{assumption}{Assumption}

\theoremstyle{remark}

\usepackage{graphicx}
\usepackage{float}
\usepackage{booktabs}
\usepackage{makecell}
\usepackage{threeparttable}
\usepackage[normalem]{ulem}
\usepackage{multirow}
\usepackage{xcolor}
\usepackage{setspace}
\usepackage[plain,noend]{algorithm2e}

\usepackage[authoryear,round]{natbib}
\usepackage[hidelinks]{hyperref}

\allowdisplaybreaks
\newcommand{\E}{\mathbb{E}}

\newcommand{\argmin}{\mathop{\mathrm{arg\,min}}}

\begin{document}

\title{Handling covariate shift by model averaging}
\author{%
Yifan Zhang$^{1}$ \quad Tianfa Xie$^{1}$ \quad Xinyu Zhang$^{2}$\\[0.6em]
\small $^{1}$School of Mathematics, Statistics and Mechanics, Beijing University of Technology, Beijing 100124, China\\
\small $^{2}$Academy of Mathematics and Systems Science, Chinese Academy of Sciences, Beijing 100190, China\\[0.4em]
\small \texttt{yfzhang@188.com} \quad \texttt{xietf@bjut.edu.cn} \quad \texttt{xinyu@amss.ac.cn}}
\date{}
\maketitle

\begin{abstract}
Distributional mismatch between the data used to construct a statistical procedure and the population to which it is ultimately applied is pervasive in modern data analysis. We study covariate shift, a fundamental instance of this problem, and develop an adaptive importance-weighted model averaging method for prediction when labeled observations are available from a source distribution, whereas only unlabeled covariates are observed from the target distribution. Procedures fitted directly to the source sample generally optimize prediction risk under the source distribution and may therefore be suboptimal for target prediction. Importance weighting by the density ratio between the target and source covariate marginals provides a natural correction, but a small number of large density-ratio values can substantially inflate the variance of the resulting estimator in finite samples. We address this bias–variance trade-off by treating the degree of importance-weighting correction as a source of model uncertainty. Specifically, we construct a family of adaptive importance-weighted least-squares estimators by raising the estimated density ratio to a range of exponents, with the endpoints corresponding to ordinary least squares and standard importance-weighted least squares, and form a data-driven average over these candidates. Under model misspecification, the proposed model averaging estimator is shown to be asymptotically optimal relative to the infeasible best convex combination of the candidate estimators. Under correct specification, a diverging penalty is shown to make the selected weights concentrate near the ordinary least-squares endpoint. Simulations and a real-data application show that the proposed method achieves competitive target-prediction performance across the settings considered. 
\end{abstract}

\par\medskip

\section{Introduction}
\label{sec:Introduction}

Statistical procedures are often developed using data collected from one
population, environment, or time period and subsequently applied in another.
As a result, the distribution governing the available data may differ
systematically from the distribution under which the procedure is ultimately
used. Such source--target distributional differences are a pervasive feature
of modern data analysis, arising when statistical procedures are transported
across populations, institutions, locations, or time periods. Procedures
calibrated exclusively to the source distribution may therefore be poorly
suited to the target distribution. This general phenomenon is commonly
referred to as distribution shift or dataset shift
\citep{KouwLoog2021}.  
Among the various forms of distribution shift, covariate shift is a widely studied and practically relevant setting \citep{QuinoneroCandelaetal2009,SugiyamaKawanabe2012}. It refers to the case in which the marginal distribution of the covariates changes while the conditional distribution of the response given the covariates remains invariant. This setting arises naturally in a wide range of prediction problems. 
For example, in clinical prediction, patient populations and measurement
patterns may vary across hospitals or clinical settings, creating discrepancies
between the data used to develop a prediction rule and the data on which it is deployed \citep{Finlaysonetal2021}. 
In image recognition, 
\citet{Baeketal2024} introduced ImageNet-ES to study covariate shifts induced
by environmental and camera sensor variations in the image acquisition process,
while leaving the underlying object categories unchanged.

A common approach for handling covariate shift is to reweight the source
observations by the target-to-source covariate density ratio
\citep{Shimodaira2000}. When this ratio is known, the target prediction
risk can be expressed as a weighted risk under the source distribution,
leading naturally to importance-weighted empirical risk minimization.
Because the density ratio is typically unknown, a variety of procedures
have been developed for estimating it from source and target covariates
\citep{Huangetal2007,Sugiyamaetal2008,KanamoriHidoSugiyama2009,
BickelBrucknerScheffer2009}.
Importance-weighted methods have been studied in mean and quantile
regression, as well as in kernel methods and other nonparametric learning
procedures
\citep{Shimodaira2000,SugiyamaSuzukiKanamori2012,Fengetal2024}.
Theoretical studies have investigated the relative merits of source-only
estimation and importance weighting under correctly specified and
misspecified working models
\citep{Gogolashvilietal2023,MaPathakWainwright2023,Geetal2024}.
Related work has further clarified how source--target overlap and the tail
and moment properties of the density ratio affect the performance of
importance-weighted procedures
\citep{CortesMansourMohri2010,MaPathakWainwright2023}.
More recently, density-ratio weighting and semiparametric correction have
been applied to prediction-set construction and selective inference under
covariate shift
\citep{Tibshiranietal2019,QiuDobribanTchetgen2023,
YangKuchibhotlaTchetgen2024,JinCandes2026}.

Despite its theoretical appeal, standard importance weighting can be unstable in finite samples. 
When some density ratios are large, the weighted empirical criterion may be dominated by a small number of source observations, giving these observations disproportionate influence on the resulting estimator and leading to variance inflation.
One useful way to reduce the influence of large density-ratio values is to
replace the density ratio $\delta(\boldsymbol x)$ by
$\delta(\boldsymbol x)^\lambda$, where $\lambda\in[0,1]$ is an exponent
\citep{SugiyamaKrauledatMuller2007}. Smaller values of $\lambda$ move the density-ratio values toward one and
thereby reduce the variance cost of reweighting, whereas larger values impose
a stronger correction for the discrepancy between the source and target
covariate distributions. The exponent therefore governs a finite-sample
trade-off between the extent of covariate-shift correction and the variance
inflation induced by large density-ratio values. The most effective correction
level depends on unknown features of the problem, including the severity of
the covariate shift, the tail behavior of the density ratio, the accuracy of
density-ratio estimation, the noise level, and the degree of model
misspecification. 
Common approaches to choosing this exponent are selection-based. In particular, exponent tuning under covariate shift is commonly handled by an AIC-type criterion \citep{Shimodaira2000} or by importance-weighted cross-validation \citep{SugiyamaKrauledatMuller2007}. Although these methods are useful, they ultimately select a single correction level. As a result, the final prediction rule does not account for model uncertainty across different degrees of importance-weighting correction. Moreover, selecting a single correction level
may fail to exploit the complementary strengths of candidate estimators with
different bias--variance profiles. This limitation may reduce predictive stability and accuracy. 

As an alternative to model selection, model averaging provides a natural way
to address uncertainty in the choice of correction level. Rather than
selecting a single candidate, model averaging assigns weights to candidate
estimators and forms a combined prediction rule, thereby
retaining information from multiple candidate procedures
\citep{HjortClaeskens2003}. This feature makes model averaging particularly
attractive when prediction is the primary objective and model uncertainty is
nonnegligible.
A substantial frequentist literature has developed a variety of
model-averaging methods, including criterion-based averaging
\citep{HjortClaeskens2003}, adaptive regression
by mixing \citep{Yang2001}, Mallows model averaging
\citep{Hansen2007,WanZhangZou2010}, jackknife model averaging
\citep{HansenRacine2012}, parsimonious model averaging
\citep{Zhangetal2020b}, and $K$-fold cross-validation-based averaging
\citep{ZhangLiu2023}.
There is also a substantial literature on theoretical guarantees and post-averaging inference; see, for example, \citet{Liu2015}, \citet{ZhangLiu2019}, and \citet{Yuetal2024}. 
These developments motivate averaging over candidate correction levels as a
way to account for uncertainty in the appropriate strength of importance
weighting, rather than committing to a single selected exponent. 

This paper develops a frequentist model-averaging method for prediction under covariate shift by treating the strength of importance weighting as a source of model uncertainty. The proposed adaptive importance-weighted model averaging (AIWMA) procedure combines candidate estimators indexed by a density-ratio exponent $\lambda$, using labeled source observations and unlabeled target covariates. A data-driven criterion is constructed to select the averaging weights. Under model misspecification, the resulting estimator is asymptotically optimal relative to the infeasible best convex combination of the candidates. Under correct specification, a diverging penalty makes the selected weights concentrate near the ordinary least-squares endpoint. Monte Carlo studies and a real-data application show that AIWMA often improves on competing exponent-selection rules in the settings considered. 
Our work is also related to model-averaging-based transfer-learning
methods that use auxiliary or source information to improve prediction in a
target population and mitigate negative transfer 
\citep{HuZhang2023,Zhangetal2024,QiuZhang2025}. Those studies consider settings in which labeled target responses are available and the main uncertainty lies in whether, and how, auxiliary source samples should be incorporated. The present problem is different: no labeled responses are observed from the target distribution, covariate shift is the primary discrepancy between the source and target samples, and the main uncertainty concerns the appropriate strength of importance-weighting correction. To the best of our knowledge, this is the first frequentist model-averaging framework that treats the strength of importance-weighting correction as model uncertainty for prediction under covariate shift with unlabeled target covariates. From the perspective of
covariate-shift adaptation, the proposed method goes beyond model-selection approaches that choose a single correction level. By averaging across correction levels, AIWMA accounts for uncertainty in the degree of reweighting and can balance the bias induced by insufficient correction against the variance inflation caused by large importance weights. The theory shows that AIWMA automatically adjusts the extent of importance-weighting correction according to the specification regime, without requiring prior knowledge of whether the working model is correctly specified. 

The rest of the paper is organized as follows. Section~\ref{sec:Problem-setup} introduces the covariate-shift framework and importance-weighted correction. Section~\ref{sec:Adaptive-importance-weighted-model-averaging} develops the model-averaging formulation and the weight-selection criterion. Section~\ref{sec:Theoretical-properties} presents the theoretical results. Section~\ref{sec:Simulation-studies} reports simulation evidence, and Section~\ref{sec:Real-data-analysis} presents the real-data analysis. Section~\ref{sec:Conclusion-Remarks} concludes. Technical proofs and additional numerical results are collected in the Supplementary Material.

\section{Problem setup}
\label{sec:Problem-setup}

\subsection{Covariate shift}
\label{subsec:Covariate-shift}

Let $\mathcal X \subseteq \mathbb R^d$ be a measurable covariate space, where $d$ is fixed, and let $\mathcal Y=\mathbb R$ be the response space. Suppose that we observe labeled data $\{(\boldsymbol x_i,y_i)\}_{i=1}^{n_s}$ drawn independently from an unknown source distribution $\rho_{\mathcal X\times\mathcal Y}^{S}$, where $\boldsymbol x_i\in\mathcal X$ and $y_i\in\mathcal Y$. The objective is to construct a prediction rule that performs well under a target distribution $\rho_{\mathcal X\times\mathcal Y}^{T}$, which may differ from the source distribution. Without additional structure, the target prediction problem is generally not identifiable when the source and target distributions are allowed to differ arbitrarily. 
We impose the standard covariate-shift assumption \citep{SugiyamaKrauledatMuller2007}: the conditional distribution of the response given the covariates is invariant across the source and target distributions, whereas the covariate marginal distribution may change. Specifically, the joint distributions admit the factorization 
\begin{equation}
\label{eq:covariate-shift}
\rho_{\mathcal X\times\mathcal Y}^{S}(\boldsymbol x,y)
=
\rho_{\mathcal X}^{S}(\boldsymbol x)\rho_{\mathcal Y\mid\mathcal X}(y\mid\boldsymbol x),
\qquad
\rho_{\mathcal X\times\mathcal Y}^{T}(\boldsymbol x,y)
=
\rho_{\mathcal X}^{T}(\boldsymbol x)\rho_{\mathcal Y\mid\mathcal X}(y\mid\boldsymbol x),
\end{equation}
where $\rho_{\mathcal X}^{S}$ and $\rho_{\mathcal X}^{T}$ are the source and target covariate marginal distributions, respectively, and $\rho_{\mathcal Y\mid\mathcal X}$ is the common conditional distribution. Accordingly, we assume that an observation $(\boldsymbol x, y)$ drawn from either the source or target distribution satisfies
\begin{equation*}
y=f(\boldsymbol x)+\varepsilon,
\end{equation*}
where $f:\mathcal X\to\mathbb R$ is the common conditional mean function. The error satisfies $\E(\varepsilon\mid\boldsymbol x)=0$ and
$\E(\varepsilon^2\mid\boldsymbol x)=\sigma^2$ for some constant
$0<\sigma^2<\infty$. 

We restrict attention to the finite-dimensional linear prediction functions $f_{\boldsymbol\theta}(\boldsymbol x)
=
\boldsymbol h(\boldsymbol x)^\top\boldsymbol\theta$, 
where $\boldsymbol h(\boldsymbol x)=\{h_1(\boldsymbol x),\ldots,h_p(\boldsymbol x)\}^\top$ is a prespecified vector of basis functions and $\boldsymbol\theta=(\theta_1,\ldots,\theta_p)^\top \in \mathbb{R}^{p}$ is the coefficient vector. We take the basis vector to include an intercept, $h_1(\boldsymbol x)\equiv1$, and assume that $p$ is fixed with $p<n_s$. 
The working linear model is said to be correctly specified if there exists
$\boldsymbol\theta_0\in\mathbb R^p$ such that
$f(\boldsymbol x)=\boldsymbol h(\boldsymbol x)^\top\boldsymbol\theta_0$
for $\rho_{\mathcal X}^{S}$- and $\rho_{\mathcal X}^{T}$-almost every $\boldsymbol x$. Otherwise, the working model is misspecified.
For any $\boldsymbol{\theta}\in\mathbb R^p$, the predictive performance of
$f_{\boldsymbol{\theta}}$ is measured by the target prediction risk  
\begin{equation*}
\label{eq:target-prediction-risk}
    \mathrm{PR}^{T}(\boldsymbol{\theta})
    =
    \E_{(\boldsymbol{x},y)\sim \rho_{\mathcal X\times\mathcal Y}^{T}}
    \left[
        \left\{y-\boldsymbol{h}(\boldsymbol{x})^\top\boldsymbol{\theta} \right\}^2
    \right].
\end{equation*}
Let $\boldsymbol{\theta}^{*}$ denote the optimal coefficient that minimizes the
target prediction risk. Since
$\mathrm{PR}^{T}(\boldsymbol{\theta})
=
\E_{\boldsymbol{x}\sim\rho_{\mathcal X}^{T}}
[
\{f(\boldsymbol{x})-\boldsymbol{h}(\boldsymbol{x})^\top
\boldsymbol{\theta}\}^{2}
]
+\sigma^{2}$, $\boldsymbol{\theta}^{*}$ can equivalently be characterized as
\begin{equation}
    \boldsymbol{\theta}^{*}
    =
    \argmin_{\boldsymbol{\theta}\in\mathbb R^p}
    \E_{\boldsymbol{x}\sim\rho_{\mathcal X}^{T}}
    \left[
        \left\{f(\boldsymbol{x})
        -\boldsymbol{h}(\boldsymbol{x})^\top
        \boldsymbol{\theta}\right\}^{2}
    \right].
\end{equation}
Then
$\boldsymbol{h}(\boldsymbol{x})^\top\boldsymbol{\theta}^{*}$ is the best linear approximation
to $f(\boldsymbol{x})$ under the target covariate marginal distribution.
Let $b(\boldsymbol{x}) = f(\boldsymbol{x})- \boldsymbol{h}(\boldsymbol{x})^\top\boldsymbol{\theta}^{*}$. 
We refer to $b(\boldsymbol{x})$ as the approximation error of the working linear model under the target distribution. Under correct
specification, $\boldsymbol{\theta}^{*}=\boldsymbol{\theta}_0$, and hence
$b(\boldsymbol{x})=0$ for $\rho_{\mathcal X}^{S}$- and
$\rho_{\mathcal X}^{T}$-almost every $\boldsymbol{x}$.

\subsection{Adaptive importance-weighted correction}
\label{subsec:Adaptive-importance-weighted-correction}

Consider the empirical risk based on the labeled source observations: 
\begin{equation}
\label{eq:source-empirical-risk}
Q_{n_s}^{(0)}(\boldsymbol\theta)
=
\frac{1}{n_s}
\sum_{i=1}^{n_s}
\left\{
y_i-\boldsymbol h(\boldsymbol x_i)^\top\boldsymbol\theta
\right\}^2 .
\end{equation}
For any fixed $\boldsymbol\theta$, this criterion is an unbiased estimator of the source prediction risk, 
$\mathrm{PR}^{S}(\boldsymbol\theta)
=
\E_{(\boldsymbol x,y)\sim \rho_{\mathcal X\times\mathcal Y}^{S}}
[
\{y-\boldsymbol h(\boldsymbol x)^\top\boldsymbol\theta\}^2
]$. 
In the absence of covariate shift, that is, when
$\rho_{\mathcal X\times\mathcal Y}^{S}=\rho_{\mathcal X\times\mathcal Y}^{T}$,
$Q_{n_s}^{(0)}(\boldsymbol\theta)$ is also unbiased for the target prediction risk,
$\mathrm{PR}^{T}(\boldsymbol\theta)$. Minimizing $Q_{n_s}^{(0)}(\boldsymbol{\theta})$ yields the ordinary
least-squares (OLS) estimator
\begin{equation}
\label{eq:theta-ols}
\widehat{\boldsymbol\theta}^{\mathrm{OLS}}
=
\left(\boldsymbol H^\top\boldsymbol H\right)^{-1}\boldsymbol H^\top\boldsymbol y,
\end{equation}
where $\boldsymbol y=(y_1,\ldots,y_{n_s})^\top$ and
$\boldsymbol H=\{\boldsymbol h(\boldsymbol x_1),\ldots,\boldsymbol h(\boldsymbol x_{n_s})\}^\top
\in\mathbb R^{n_s\times p}$.

Under covariate shift,
$\rho_{\mathcal X\times\mathcal Y}^{S}
\ne
\rho_{\mathcal X\times\mathcal Y}^{T}$.
Consequently, the unweighted criterion
$Q_{n_s}^{(0)}(\boldsymbol\theta)$ is generally not unbiased for
$\mathrm{PR}^{T}(\boldsymbol\theta)$, and the resulting OLS estimator need not minimize the target prediction
risk. To correct this discrepancy, a standard strategy is importance weighting \citep{Shimodaira2000}. Suppose that the target covariate marginal distribution $\rho_{\mathcal X}^{T}$ is absolutely continuous with respect to the source covariate marginal distribution $\rho_{\mathcal X}^{S}$. Define the target-to-source covariate density ratio, also referred to as the
importance weight, as the Radon--Nikodym derivative of $\rho_{\mathcal X}^{T}$ with respect to $\rho_{\mathcal X}^{S}$: 
\begin{equation} 
\label{eq:Density-ratio} 
    \delta(\boldsymbol x) = \frac{\mathrm{d}\rho_{\mathcal X}^{T}}{\mathrm{d}\rho_{\mathcal X}^{S}}(\boldsymbol x). 
\end{equation}
If $\rho_{\mathcal X}^{S}$ and $\rho_{\mathcal X}^{T}$ have probability
density functions with respect to a common reference measure, then
$\delta(\boldsymbol x)$ reduces to the ratio of the target covariate density to
the source covariate density. 
With this density ratio and under the covariate-shift assumption
\eqref{eq:covariate-shift}, any integrable function
$g:\mathcal X\times\mathcal Y\to\mathbb R$ satisfies the
importance-weighting identity
\begin{equation}
\label{eq:Importance-weighting-identity}
\E_{(\boldsymbol x,y)\sim\rho_{\mathcal X\times\mathcal Y}^{T}}
\left\{g(\boldsymbol x,y)\right\}
=
\E_{(\boldsymbol x,y)\sim\rho_{\mathcal X\times\mathcal Y}^{S}}
\left\{\delta(\boldsymbol x)g(\boldsymbol x,y)\right\}.
\end{equation}
Applying
\eqref{eq:Importance-weighting-identity} with
$g(\boldsymbol x,y)=\{y-\boldsymbol h(\boldsymbol x)^\top\boldsymbol\theta\}^{2}$ motivates the standard importance-weighted empirical criterion
\begin{equation}
\label{eq:Importance-weighted-criterion}
Q_{n_s}^{(1)}(\boldsymbol\theta)
=
\frac{1}{n_s}
\sum_{i=1}^{n_s}
\delta(\boldsymbol x_i)
\left\{
y_i-\boldsymbol h(\boldsymbol x_i)^\top\boldsymbol\theta
\right\}^2 .
\end{equation}
Minimizing $Q_{n_s}^{(1)}(\boldsymbol\theta)$ yields the importance-weighted least-squares (IWLS) estimator 
\begin{equation*}
\label{eq:theta-hat-IWLS}
\widehat{\boldsymbol\theta}^{\mathrm{IWLS}}
=
(\boldsymbol H^\top\boldsymbol\Delta\boldsymbol H)^{-1}
\boldsymbol H^\top\boldsymbol\Delta\boldsymbol y,
\end{equation*}
where $
\boldsymbol\Delta
=
\operatorname{diag}\{\delta(\boldsymbol x_1),\ldots,\delta(\boldsymbol x_{n_s})\}$.

When the density ratio is known, $Q_{n_s}^{(1)}(\boldsymbol\theta)$ is an
unbiased estimator of $\mathrm{PR}^{T}(\boldsymbol\theta)$ for every fixed
$\boldsymbol\theta$. In finite samples, however, a small number of source
observations with large density-ratio values may dominate the weighted
criterion. Their disproportionate influence can substantially inflate the
variance of the resulting IWLS estimator and thereby impair its predictive
performance under the target distribution. 
Following \citet{SugiyamaKrauledatMuller2007}, we consider an adaptive importance-weighting scheme that replaces
$\delta(\boldsymbol x)$ by $\delta(\boldsymbol x)^\lambda$ with
$\lambda\in[0,1]$, thereby shrinking extreme density-ratio values toward one. For a given $\lambda\in[0,1]$, define 
\begin{equation}
\label{eq:Q-n-lambda-theta}
Q_{n_s}^{(\lambda)}(\boldsymbol\theta)
=
\frac{1}{n_s}
\sum_{i=1}^{n_s}
\delta(\boldsymbol x_i)^\lambda
\left\{
y_i-\boldsymbol h(\boldsymbol x_i)^\top\boldsymbol\theta
\right\}^2 ,
\end{equation}
with the convention that $\delta(\boldsymbol x_i)^0=1$. 
Minimizing \eqref{eq:Q-n-lambda-theta} yields the adaptive importance-weighted least-squares (AIWLS) estimator 
\begin{equation}
\label{eq:theta-hat-AIWLS}
\widehat{\boldsymbol\theta}^{\mathrm{AIWLS}}(\lambda)
=
(\boldsymbol H^\top\boldsymbol\Delta^\lambda\boldsymbol H)^{-1}
\boldsymbol H^\top\boldsymbol\Delta^\lambda\boldsymbol y,
\end{equation}
where $\boldsymbol\Delta^\lambda
=
\operatorname{diag}\{\delta(\boldsymbol x_1)^\lambda,\ldots,\delta(\boldsymbol x_{n_s})^\lambda\}$. The endpoints recover two standard estimators: $\lambda=0$ gives OLS, whereas
$\lambda=1$ gives IWLS. Thus, the exponent $\lambda$ governs the strength of the covariate-shift correction. Smaller values of $\lambda$ move the estimator toward the unweighted least-squares fit and are typically more stable, whereas larger values of $\lambda$ apply stronger importance-weighting correction but can incur greater variance when some source observations receive large density-ratio values. 
The Supplementary Material provides an illustrative example showing that the
best-performing correction exponent need not coincide with either endpoint
$\lambda=0$ or $\lambda=1$, but may lie in the interior of $[0,1]$. 
The next subsection formalizes this trade-off.

\subsection{Asymptotic distribution of AIWLS}
\label{subsec:Asymptotic-distribution-of-AIWLS}

We derive the asymptotic distribution of 
$\widehat{\boldsymbol\theta}^{\mathrm{AIWLS}}(\lambda)$ for a fixed exponent
$\lambda\in[0,1]$. 
This result clarifies how the choice of $\lambda$ reflects
the trade-off between finite-sample stability and correction toward the target
covariate distribution. To isolate the effect of the correction exponent, we
first consider the oracle case in which the density-ratio values at the source
covariates are known. For a fixed $\lambda\in[0,1]$, define
\begin{equation}
\boldsymbol\theta^{*}(\lambda)
=
\argmin_{\boldsymbol\theta\in\mathbb R^p}
\E_{\boldsymbol x\sim\rho_{\mathcal X}^{S}}
\left[
\delta(\boldsymbol x)^\lambda
\left\{
f(\boldsymbol x)
-
\boldsymbol h(\boldsymbol x)^\top\boldsymbol\theta
\right\}^2
\right].
\label{eq:theta-star-lambda}
\end{equation}
Let
$b_\lambda(\boldsymbol x)
=
f(\boldsymbol x)
-
\boldsymbol h(\boldsymbol x)^\top
\boldsymbol\theta^{*}(\lambda)$ denote the corresponding approximation error, 
and define $\boldsymbol\Sigma_b(\lambda)
=
\E_{\boldsymbol x\sim\rho_{\mathcal X}^{S}}
[
\delta(\boldsymbol x)^{2\lambda}
b_\lambda(\boldsymbol x)^2
\boldsymbol h(\boldsymbol x)\boldsymbol h(\boldsymbol x)^\top
]$. We further define $ \boldsymbol\Sigma(\lambda)
=
\E_{\boldsymbol x\sim\rho_{\mathcal X}^{S}}
[
\delta(\boldsymbol x)^\lambda
\boldsymbol h(\boldsymbol x)\boldsymbol h(\boldsymbol x)^\top
]$ and $\boldsymbol\Sigma_\varepsilon(\lambda)
=
\sigma^2
\E_{\boldsymbol x\sim\rho_{\mathcal X}^{S}}
[
\delta(\boldsymbol x)^{2\lambda}
\boldsymbol h(\boldsymbol x)\boldsymbol h(\boldsymbol x)^\top
]$. 
For a vector $\boldsymbol v$, let $\|\boldsymbol v\|_2$ denote its
Euclidean norm, and, for a symmetric matrix $\boldsymbol A$, let
$\lambda_{\min}(\boldsymbol A)$ denote its smallest eigenvalue.
Throughout, $\xrightarrow{p}$ and $\xrightarrow{d}$ denote convergence
in probability and convergence in distribution, respectively.

\begin{theorem}
\label{thm:AIWLS-asymptotic-normality}
Suppose that the density-ratio values
$\{\delta(\boldsymbol x_i)\}_{i=1}^{n_s}$ are known. Fix any
$\lambda\in[0,1]$. Assume that there exists a constant $c_\lambda>0$
such that
$\lambda_{\min}\{\boldsymbol\Sigma(\lambda)\}\ge c_\lambda$.
In addition, suppose that the moment conditions
$\E_{\boldsymbol x\sim\rho_{\mathcal X}^{S}}
[
\{1+\delta(\boldsymbol x)^{2\lambda}\}
\|\boldsymbol h(\boldsymbol x)\|_2^2
]
<\infty$
and
$\E_{\boldsymbol x\sim\rho_{\mathcal X}^{S}}
[
\delta(\boldsymbol x)^{2\lambda}
b_\lambda(\boldsymbol x)^2
\|\boldsymbol h(\boldsymbol x)\|_2^2
]
<\infty$
hold. Then, as $n_s\to\infty$, 
\begin{equation}
\label{eq:AIWLS-asymptotic-normality}
\sqrt{n_s}
\left\{
\widehat{\boldsymbol\theta}^{\mathrm{AIWLS}}(\lambda)
-
\boldsymbol\theta^{*}(\lambda)
\right\}
\xrightarrow{d}
N\{\boldsymbol 0,\boldsymbol\Omega(\lambda)\},
\end{equation}
where $\boldsymbol\Omega(\lambda)
=
\boldsymbol\Sigma(\lambda)^{-1}
\{
\boldsymbol\Sigma_b(\lambda)
+
\boldsymbol\Sigma_\varepsilon(\lambda)
\}
\boldsymbol\Sigma(\lambda)^{-1}$.
\end{theorem}

At $\lambda=1$, the importance-weighting identity in
\eqref{eq:Importance-weighting-identity}, together with the definition of
$\boldsymbol\theta^*(\lambda)$, yields
$\boldsymbol\theta^*(1)=\boldsymbol\theta^*$. Hence, the IWLS estimator is consistent for $\boldsymbol{\theta}^{*}$, the coefficient vector that minimizes the target prediction risk. 
If the working linear model is correctly specified, the objective in
\eqref{eq:theta-star-lambda} is minimized at the common coefficient
$\boldsymbol\theta_0$ for every $\lambda\in[0,1]$. Consequently,
$\boldsymbol\theta^*(\lambda)=\boldsymbol\theta_0=\boldsymbol\theta^*$ for all $\lambda\in[0,1]$. 
Under model misspecification, however, choosing $\lambda<1$ may move
$\boldsymbol\theta^*(\lambda)$ away from $\boldsymbol\theta^*$.  
Larger values of $\lambda$ provide a stronger correction
toward the target covariate distribution, with exact correction at
$\lambda=1$. However, the matrices
$\boldsymbol\Sigma_b(\lambda)$ and
$\boldsymbol\Sigma_\varepsilon(\lambda)$ in the asymptotic variance
$\boldsymbol\Omega(\lambda)$ involve second-moment terms weighted by
$\delta(\boldsymbol x)^{2\lambda}$. 
Increasing $\lambda$ makes the estimator more sensitive to observations with large density ratios and may increase its variance. 
The exponent $\lambda$
accordingly governs a trade-off between bias correction under
misspecification and variance control.
Since the
value of $\lambda$ that minimizes the target prediction risk  depends on unknown features
of the problem, including the severity of covariate shift, the magnitude of the
density-ratio values, the noise level, and the degree of model misspecification,
selecting a single exponent may lead to unstable finite-sample performance and
fail to exploit the complementary strengths of different correction levels. 
The next section addresses this uncertainty by averaging over a collection of AIWLS estimators.

\section{Adaptive importance-weighted model averaging}
\label{sec:Adaptive-importance-weighted-model-averaging}


In practice, the density ratio is typically unknown and must be estimated.
Suppose that, in addition to the labeled source sample
$\{(\boldsymbol x_i,y_i)\}_{i=1}^{n_s}$, we observe an unlabeled target
covariate sample $\{\boldsymbol x_j'\}_{j=1}^{n_t}$ drawn independently from
$\rho_{\mathcal X}^{T}$ and independently of the source sample.
Let $n=\min\{n_s,n_t\}$ denote the smaller of the two sample sizes.
Throughout the asymptotic analysis, all limits are taken as $n\to\infty$. 
Based on the two covariate samples, we estimate the density ratio at each observed source covariate using a standard density-ratio estimation method, such as the Kullback--Leibler importance estimation procedure (KLIEP)
\citep{Sugiyamaetal2008}, and denote the resulting nonnegative estimates by $\widehat\delta(\boldsymbol x_i)$, $i=1,\ldots,n_s$.

Let
$\Lambda_n=\{\lambda_1,\lambda_2,\ldots,\lambda_{M_n}\}\subset[0,1]$
be a prespecified grid of candidate exponents, with $M_n$ allowed to diverge so that the grid becomes increasingly dense as the sample size grows. 
Without loss of generality, we assume that the grid contains
both endpoints and index its elements so that
$0=\lambda_1<\lambda_2<\cdots<\lambda_{M_n}=1$. Hence, the candidate family includes ordinary least squares and standard
importance-weighted least squares as its two endpoint candidates. For each $m = 1, \ldots, M_{n}$, define the $m$th AIWLS candidate estimator by 
\begin{equation}
\label{eq:theta-hat-m}
\widehat{\boldsymbol\theta}_m
=
\left(
\boldsymbol H^\top
\widehat{\boldsymbol\Delta}_m
\boldsymbol H
\right)^{-1}
\boldsymbol H^\top
\widehat{\boldsymbol\Delta}_m
\boldsymbol y
\equiv
\boldsymbol L_m\boldsymbol y,
\end{equation}
where
$\widehat{\boldsymbol\Delta}_m
=
\operatorname{diag}\{
\widehat\delta(\boldsymbol x_1)^{\lambda_m},
\ldots,
\widehat\delta(\boldsymbol x_{n_s})^{\lambda_m}
\}$. 
The candidate estimators therefore differ in the strength
of the importance-weighting correction determined by $\lambda_m$. 
Let $\boldsymbol w=(w_1,\ldots,w_{M_{n}})^\top$ be an averaging weight vector in $\mathcal{W}_{n}
=
\{
\boldsymbol w\in[0,1]^{M_{n}}:
\sum_{m=1}^{M_{n}} w_m=1
\}$. 
Define the adaptive importance-weighted model averaging (AIWMA) estimator as
\begin{equation}
\label{eq:theta-hat-w}
\widehat{\boldsymbol\theta}(\boldsymbol w)
=
\sum_{m=1}^{M_{n}}
w_m\widehat{\boldsymbol\theta}_m
=
\boldsymbol L(\boldsymbol w)\boldsymbol y,
\end{equation}
where $\boldsymbol L(\boldsymbol w)
=
\sum_{m=1}^{M_{n}}
w_m\boldsymbol L_m$. 


We next derive a criterion for selecting the averaging weights. For a
given weight vector $\boldsymbol w$, the target prediction risk of
$\widehat{\boldsymbol\theta}(\boldsymbol w)$, evaluated at a 
target observation
$(\boldsymbol z,u)\sim\rho_{\mathcal X\times\mathcal Y}^{T}$, is
\begin{equation*}
\mathrm{PR}^{T}\{\widehat{\boldsymbol\theta}(\boldsymbol w)\}
=
\E_{(\boldsymbol z,u)\sim\rho_{\mathcal X\times\mathcal Y}^{T}}
\left[
\left\{
u-\boldsymbol h(\boldsymbol z)^\top
\widehat{\boldsymbol\theta}(\boldsymbol w)
\right\}^{2}
\right].
\end{equation*}
By the definition of $\boldsymbol\theta^*$, the approximation error
$b(\boldsymbol z)
=f(\boldsymbol z)-\boldsymbol h(\boldsymbol z)^\top\boldsymbol\theta^*$
is orthogonal to the working basis under the target covariate
distribution. It follows that
\begin{align}
\label{eq:target-risk-expansion-w}
\mathrm{PR}^{T}\{\widehat{\boldsymbol\theta}(\boldsymbol w)\}
={}&
\left\{
\widehat{\boldsymbol\theta}(\boldsymbol w)-\boldsymbol\theta^*
\right\}^\top
\boldsymbol\Sigma^{T}
\left\{
\widehat{\boldsymbol\theta}(\boldsymbol w)-\boldsymbol\theta^*
\right\}
+
\E_{\boldsymbol z\sim\rho_{\mathcal X}^{T}}
\{b(\boldsymbol z)^2\}
+
\sigma^2
\nonumber\\
={}&
\widehat{\boldsymbol\theta}(\boldsymbol w)^\top
\boldsymbol\Sigma^{T}
\widehat{\boldsymbol\theta}(\boldsymbol w)
-
2\widehat{\boldsymbol\theta}(\boldsymbol w)^\top
\boldsymbol\Sigma^{T}\boldsymbol\theta^*
+
r_0,
\end{align}
where
$\boldsymbol\Sigma^{T}
=
\E_{\boldsymbol z\sim\rho_{\mathcal X}^{T}}
\{\boldsymbol h(\boldsymbol z)\boldsymbol h(\boldsymbol z)^\top\}$
and
$r_0
=
{\boldsymbol\theta^*}^\top
\boldsymbol\Sigma^{T}\boldsymbol\theta^*
+
\E_{\boldsymbol z\sim\rho_{\mathcal X}^{T}}\{b(\boldsymbol z)^2\}
+
\sigma^2$
does not depend on $\boldsymbol w$. Therefore, constructing a criterion for the averaging weights requires
estimating the two $\boldsymbol w$-dependent terms
$\widehat{\boldsymbol\theta}(\boldsymbol w)^\top
\boldsymbol\Sigma^{T}
\widehat{\boldsymbol\theta}(\boldsymbol w)$ and
$-2\widehat{\boldsymbol\theta}(\boldsymbol w)^\top
\boldsymbol\Sigma^{T}\boldsymbol\theta^*$. 
Let
$\mathcal F
=
\sigma(
\boldsymbol x_1,\ldots,\boldsymbol x_{n_s},
\boldsymbol x_1',\ldots,\boldsymbol x_{n_t}'
)$
denote the $\sigma$-field generated by the observed source covariates
and the unlabeled target covariates. Suppose temporarily that there
exists a linear estimator
$\widehat{\boldsymbol\theta}_{u}
=\boldsymbol L_{u}\boldsymbol y$
that is conditionally unbiased for $\boldsymbol\theta^*$, namely,
$\E(\widehat{\boldsymbol\theta}_{u}\mid\mathcal F)
=\boldsymbol\theta^*$, where the conditional expectation is taken with
respect to the source errors given $\mathcal F$. Then
\begin{equation}
\label{eq:cross-term-unbiased-identity}
\E
\left[
\widehat{\boldsymbol\theta}(\boldsymbol w)^\top
\boldsymbol\Sigma^{T}
\widehat{\boldsymbol\theta}_{u}
-
\sigma^2
\operatorname{tr}
\left\{
\boldsymbol\Sigma^{T}
\boldsymbol L(\boldsymbol w)
\boldsymbol L_{u}^\top
\right\}
\mid \mathcal F
\right]
=
\E
\left[
\widehat{\boldsymbol\theta}(\boldsymbol w)^\top
\boldsymbol\Sigma^{T}
\boldsymbol\theta^*
\mid \mathcal F
\right].
\end{equation}
Motivated by this identity, we define the ideal weight-selection criterion
\begin{equation}
\label{eq:C-star-w}
C^*(\boldsymbol w)
=
\left\{
\widehat{\boldsymbol\theta}(\boldsymbol w)
-
\widehat{\boldsymbol\theta}_{u}
\right\}^\top
\boldsymbol\Sigma^{T}
\left\{
\widehat{\boldsymbol\theta}(\boldsymbol w)
-
\widehat{\boldsymbol\theta}_{u}
\right\}
+
2\sigma^2
\operatorname{tr}
\left\{
\boldsymbol\Sigma^{T}
\boldsymbol L(\boldsymbol w)
\boldsymbol L_{u}^\top
\right\}.
\end{equation}
Let
$r_u
=
\E(
\widehat{\boldsymbol\theta}_{u}^\top
\boldsymbol\Sigma^{T}
\widehat{\boldsymbol\theta}_{u}
\mid\mathcal F)
-r_0$, which is independent of $\boldsymbol w$. Combining
\eqref{eq:target-risk-expansion-w} with
\eqref{eq:cross-term-unbiased-identity} yields
\begin{equation*}
\E\left[
C^*(\boldsymbol w)-r_u
\mid\mathcal F
\right]
=
\E\left[
\mathrm{PR}^{T}
\{\widehat{\boldsymbol\theta}(\boldsymbol w)\}
\mid\mathcal F
\right].
\end{equation*}  

However, a conditionally unbiased estimator of $\boldsymbol\theta^*$ is generally unavailable in practice. We therefore replace
$\widehat{\boldsymbol\theta}_{u}$ in $C^{*}(\boldsymbol{w})$ with the IWLS candidate $\widehat{\boldsymbol\theta}_{M_n}$ corresponding to $\lambda_{M_n}=1$. This substitution is motivated by the fact that
$\widehat{\boldsymbol\theta}_{M_n}$ is asymptotically conditionally
unbiased for $\boldsymbol\theta^*$ under regularity conditions. We further replace the unknown quantities $\boldsymbol\Sigma^{T}$ and $\sigma^2$ with their estimators
$\widehat{\boldsymbol\Sigma}^{T}
=n_t^{-1}\sum_{j=1}^{n_t}
\boldsymbol h(\boldsymbol x_j')
\boldsymbol h(\boldsymbol x_j')^\top$
and $\widehat\sigma^2
=\|\boldsymbol y-\boldsymbol H
\widehat{\boldsymbol\theta}_1\|_2^2/(n_s-p)$, respectively, where $\widehat{\boldsymbol\theta}_1$ is the OLS candidate corresponding to $\lambda_1=0$. 
The factor $2$ in $C^*(\boldsymbol w)$ is dictated by the exact
conditional unbiasedness identity. For the feasible criterion, we replace
this canonical factor by a positive deterministic sequence $\phi_n$, which
controls the multiplicative calibration of the variance-correction term.
This formulation retains the unbiasedness-motivated choice $\phi_n=2$ and
also permits diverging choices used to establish the subsequent
weight-concentration result.
The resulting weight-selection
criterion is
\begin{equation}
\label{eq:C-n-w}
C(\boldsymbol w)
=
\left\{
\widehat{\boldsymbol\theta}(\boldsymbol w)
-
\widehat{\boldsymbol\theta}_{M_n}
\right\}^{\top}
\widehat{\boldsymbol\Sigma}^{T}
\left\{
\widehat{\boldsymbol\theta}(\boldsymbol w)
-
\widehat{\boldsymbol\theta}_{M_n}
\right\}
+
\phi_n\widehat\sigma^2
\operatorname{tr}
\left\{
\widehat{\boldsymbol\Sigma}^{T}
\boldsymbol L(\boldsymbol w)
\boldsymbol L_{M_n}^{\top}
\right\}.
\end{equation}

The criterion $C(\boldsymbol w)$ can be written in the quadratic form
\begin{equation*}
\label{eq:C-n-quadratic-form}
C(\boldsymbol w)
=
\boldsymbol w^\top\boldsymbol\Psi\boldsymbol w
+
\phi_n\widehat\sigma^2
\boldsymbol w^\top\boldsymbol\tau,
\end{equation*}
where $\boldsymbol\Psi$ is the $M_{n} \times M_{n}$ matrix with entries
$\Psi_{ij}
=
(\widehat{\boldsymbol\theta}_i-\widehat{\boldsymbol\theta}_{M_{n}})^\top
\widehat{\boldsymbol\Sigma}^{T}
(\widehat{\boldsymbol\theta}_j-\widehat{\boldsymbol\theta}_{M_{n}})$, and
$\boldsymbol\tau=(\tau_1,\ldots,\tau_{M_{n}})^\top$ has entries
$\tau_m
=
\operatorname{tr}
\{
\widehat{\boldsymbol\Sigma}^{T}
\boldsymbol L_m
\boldsymbol L_{M_{n}}^\top
\}$.
Hence, selecting the averaging weights reduces to a standard quadratic
programming problem over the simplex, 
$\widehat{\boldsymbol w}
\in
\argmin_{\boldsymbol w\in\mathcal{W}_{n}}
C(\boldsymbol w)$. 
The resulting
AIWMA prediction is  $\widehat f_{\mathrm{AIWMA}}(\boldsymbol x)
=
\boldsymbol h(\boldsymbol x)^\top
\widehat{\boldsymbol\theta}(\widehat{\boldsymbol w})
=
\sum_{m=1}^{M_{n}}
\widehat w_m
\boldsymbol h(\boldsymbol x)^\top
\widehat{\boldsymbol\theta}_m$. 

\section{Theoretical properties of AIWMA}
\label{sec:Theoretical-properties}

In this section, we study the theoretical properties of AIWMA under both
model misspecification and correct specification. Recall that
$\mathcal F
=
\sigma(
\boldsymbol x_1,\ldots,\boldsymbol x_{n_s},
\boldsymbol x_1',\ldots,\boldsymbol x_{n_t}'
)$
is the $\sigma$-field generated by the observed source and target covariates.
Let $(\boldsymbol z,u)$ denote an independent observation drawn from
$\rho_{\mathcal X\times\mathcal Y}^{T}$.  For $\boldsymbol w\in\mathcal{W}_{n}$, define the
target prediction error (TPE) by
\begin{equation}
\label{eq:TPE-w}
\mathrm{TPE}\left(\boldsymbol{w}\right)
=
\E\left[
\left\{
u
-
\boldsymbol h(\boldsymbol z)^\top
\widehat{\boldsymbol\theta}\left(\boldsymbol{w}\right)
\right\}^{2}
\mid
\mathcal F
\right],
\end{equation}
where the conditional expectation is taken over the source errors and this independent target observation, given $\mathcal F$.
Using $u=f(\boldsymbol z)+\varepsilon$, together with
$\E(\varepsilon\mid\boldsymbol z)=0$ and
$\E(\varepsilon^2\mid\boldsymbol z)=\sigma^2$, we obtain $\mathrm{TPE}(\boldsymbol{w})
=
\E[
\{
f(\boldsymbol z)
-
\boldsymbol h(\boldsymbol z)^\top
\widehat{\boldsymbol\theta}(\boldsymbol{w})
\}^{2}
\mid
\mathcal F
]
+
\sigma^{2}$. 
Since $\sigma^2$ does not depend on $\boldsymbol w$, minimizing
$\mathrm{TPE}(\boldsymbol w)$ is equivalent to minimizing the target excess
prediction error 
(TEPE): 
\begin{equation}
\label{eq:TEPE-w}
\mathrm{TEPE}\left(\boldsymbol{w}\right)
=
\E\left[
\left\{
f(\boldsymbol z)
-
\boldsymbol h(\boldsymbol z)^\top
\widehat{\boldsymbol\theta}\left(\boldsymbol{w}\right)
\right\}^{2}
\mid
\mathcal F
\right].
\end{equation}
To formulate the results under model misspecification, write $\boldsymbol\theta_m^*
\equiv
\boldsymbol\theta^*(\lambda_m)$,
$m=1,\ldots,M_n$, where
$\boldsymbol\theta^*(\lambda_m)$ is defined in
\eqref{eq:theta-star-lambda}. For $\boldsymbol w\in\mathcal{W}_{n}$, let
$\boldsymbol\theta^*(\boldsymbol w)=\sum_{m=1}^{M_{n}}w_m\boldsymbol\theta_m^*$ and
define
$\mathrm{TEPE}^{*}(\boldsymbol w)
=
\E_{\boldsymbol z\sim\rho_{\mathcal X}^{T}}
[\{
f(\boldsymbol z)-
\boldsymbol h(\boldsymbol z)^\top
\boldsymbol\theta^*(\boldsymbol w)
\}^{2}
]$.
Let
$\xi_n=\inf_{\boldsymbol w\in\mathcal{W}_{n}}\mathrm{TEPE}^{*}(\boldsymbol w)$
denote the minimum target excess prediction error in the class of model
averaging estimators associated with $\boldsymbol{\theta}_{m}^{*}$. 
Under correct specification,
$\boldsymbol\theta_m^*=\boldsymbol\theta_0$ for every
$m=1,\ldots,M_n$. Hence,
$\boldsymbol\theta^*(\boldsymbol w)=\boldsymbol\theta_0$ and
$\mathrm{TEPE}^*(\boldsymbol w)=0$ for every
$\boldsymbol w\in\mathcal W_n$, so that $\xi_n=0$. 
Finally, define the uniform estimation error of the exponentiated density ratios over the candidate set and the observed source covariates by
$\kappa_n
=
\max_{1\le m\le M_n}
\max_{1\le i\le n_s}
|
\widehat\delta(\boldsymbol x_i)^{\lambda_m}
-
\delta(\boldsymbol x_i)^{\lambda_m}
|$. 

\subsection{Asymptotic unbiasedness of the AIWMA criterion}
\label{subsec:Asymptotic-unbiasedness-of-the-AIWMA-criterion}

We first establish the asymptotic unbiasedness of the proposed weight-selection criterion. We impose the following regularity conditions.

\begin{assumption}
\label{assump:1}
There exists a constant $c_1>0$ such that
$\inf_{0\le\lambda\le1}\lambda_{\min}\{\boldsymbol\Sigma(\lambda)\}\ge c_1$.
\end{assumption}

\begin{assumption}
\label{assump:2}
There exists a constant $0<c_2<\infty$ such that
$\E(\varepsilon^4\mid\boldsymbol x)\le c_2$ almost surely.
\end{assumption}

\begin{assumption}
\label{assump:3}
The approximation errors and basis functions satisfy
\begin{equation*}
\sup_{0\le\lambda\le1}
\E_{\boldsymbol x\sim\rho_{\mathcal X}^{S}}     \left[
b_\lambda(\boldsymbol x)^2
\left\|\boldsymbol h(\boldsymbol x) \right\|_2^2
\right]<\infty,
\qquad
\E_{\boldsymbol x\sim\rho_{\mathcal X}^{S}}
\left[
\left\|\boldsymbol h(\boldsymbol x) \right\|_2^4
\right]<\infty .
\end{equation*}
\end{assumption}

\begin{assumption}
\label{assump:4}
There exists a constant $0<\overline\delta<\infty$ such that
$0\le\delta(\boldsymbol x)\le\overline\delta$ for
$\rho_{\mathcal X}^{S}$-almost every $\boldsymbol x$.
\end{assumption}

\begin{assumption}
\label{assump:5}
$M_{n}^2/(n\xi_n^2)=o(1)$, and $\kappa_{n}=o_p(\xi_n)$.
\end{assumption}

Assumption~\ref{assump:1} requires $\boldsymbol\Sigma(\lambda)$ to be uniformly
nonsingular over $\lambda\in[0,1]$. Assumptions~\ref{assump:2}--\ref{assump:3}
impose moment conditions on the errors, basis functions, and approximation
errors.
Assumption~\ref{assump:4} requires the density ratio to be bounded above. This
condition is stronger than the absolute continuity condition needed to define
$\delta(\boldsymbol x)$, but it is commonly used in theoretical studies of
importance weighting and covariate-shift correction; see, for example,
\citet{CortesMansourMohri2010} and  \citet{Reddi2015}. 
Assumption~\ref{assump:5} imposes rate restrictions on the sample size $n$, the 
number of candidate estimators $M_{n}$, the density-ratio estimation error
$\kappa_n$, and the minimum target excess prediction error $\xi_n$. The
condition $M_{n}^2/(n\xi_n^2)=o(1)$ implicitly requires $\xi_n\ne0$, and hence
corresponds to the misspecified setting in which the candidate averaging class
has nonzero minimum target excess prediction error. It also allows $\xi_n$ to
converge to zero, provided that it does so at a slower rate than
$M_{n}/\sqrt n$. The requirement $\kappa_n=o_p(\xi_n)$ requires the density-ratio estimation
error to be negligible relative to the minimum target excess prediction error.

We next introduce auxiliary quantities used to state the conditional expansion
of the criterion. For every $\boldsymbol w\in\mathcal{W}_{n}$, let
\begin{align}
\label{eq:en-w}
e_n(\boldsymbol w)
={}&
\E\left[
\left\{
\widehat{\boldsymbol\theta}(\boldsymbol w)
-
\widehat{\boldsymbol\theta}_{M_{n}}
\right\}^\top
\left(
\widehat{\boldsymbol\Sigma}^{T}
-
\boldsymbol\Sigma^{T}
\right)
\left\{
\widehat{\boldsymbol\theta}(\boldsymbol w)
-
\widehat{\boldsymbol\theta}_{M_{n}}
\right\}
\mid \mathcal F
\right]
\nonumber\\
&\quad
+
2\E(\widehat\sigma^2\mid\mathcal F)
\operatorname{tr}
\left[
\left(
\widehat{\boldsymbol\Sigma}^{T}
-
\boldsymbol\Sigma^{T}
\right)
\boldsymbol L(\boldsymbol w)
\boldsymbol L_{M_{n}}^\top
\right],
\\
\label{eq:dn-w}
d_n(\boldsymbol w)
={}&
-2
\boldsymbol b^\top
\boldsymbol L(\boldsymbol w)^\top
\boldsymbol\Sigma^{T}
\boldsymbol L_{M_{n}}
\boldsymbol b
+
2
\frac{\|\boldsymbol G\boldsymbol b\|_2^2}
{\operatorname{tr}(\boldsymbol G)}
\operatorname{tr}
\left\{
\boldsymbol\Sigma^{T}
\boldsymbol L(\boldsymbol w)
\boldsymbol L_{M_{n}}^\top
\right\},
\end{align}
where
$\boldsymbol b=\{b(\boldsymbol x_1),\ldots,b(\boldsymbol x_{n_s})\}^\top$
and
$\boldsymbol G
=
\boldsymbol I_{n_s}
-
\boldsymbol H(\boldsymbol H^\top\boldsymbol H)^{-1}\boldsymbol H^\top$,
with $\boldsymbol I_{n_s}$ denoting the $n_s\times n_s$ identity matrix. 
The term $e_n(\boldsymbol w)$ reflects the use of the empirical target second-moment matrix $\widehat{\boldsymbol\Sigma}^{T}$ in place of
$\boldsymbol\Sigma^{T}$, whereas $d_n(\boldsymbol w)$ is induced by model
misspecification and residual-based variance estimation. Under correct specification and when $\widehat{\boldsymbol\Sigma}^{T}=\boldsymbol\Sigma^{T}$, both terms vanish. 
Finally, define the term that does not depend on $\boldsymbol w$ by 
\begin{equation}
\label{eq:rn}
r_{n}
=
\E\left[
\left(
\widehat{\boldsymbol\theta}_{M_{n}}
-
\boldsymbol\theta^*
\right)^\top
\boldsymbol\Sigma^{T}
\left(
\widehat{\boldsymbol\theta}_{M_{n}}
-
\boldsymbol\theta^*
\right)
\mid \mathcal F
\right]
-
\E_{\boldsymbol x\sim\rho_{\mathcal X}^{T}}
\{b(\boldsymbol x)^2\}.
\end{equation}

\begin{theorem}
\label{thm:criterion-unbiasedness}
Suppose that $\phi_n=2$. Then, for every
$\boldsymbol w\in\mathcal{W}_n$,
\begin{equation}
\label{eq:cond-expansion-main}
\E\left[
C(\boldsymbol w)-r_n
\mid
\mathcal F
\right]
=
\mathrm{TEPE}(\boldsymbol w)
+
d_n(\boldsymbol w)
+
e_n(\boldsymbol w).
\end{equation}
Furthermore, if Assumptions~\ref{assump:1}--\ref{assump:5} hold, then
\begin{equation}
\label{eq:cond-asymp-unbiased-equivalent}
\E\left[
C(\boldsymbol w)-r_n
\mid
\mathcal F
\right]
=
\mathrm{TEPE}(\boldsymbol w)\{1+o_p(1)\},
\end{equation}
uniformly over $\boldsymbol w\in\mathcal{W}_n$.
\end{theorem}

\begin{corollary}
\label{cor:criterion-exact-unbiasedness}
Suppose that $\phi_n=2$, that the working linear model is correctly
specified, and that
$\widehat{\boldsymbol\Sigma}^{T}=\boldsymbol\Sigma^{T}$. Then, for every
$\boldsymbol w\in\mathcal{W}_n$,
\begin{equation}
\label{eq:cond-exact-unbiased-correct}
\E\left[
C(\boldsymbol w)-r_n
\mid
\mathcal F
\right]
=
\mathrm{TEPE}(\boldsymbol w).
\end{equation}
\end{corollary} 

Theorem~\ref{thm:criterion-unbiasedness} justifies the use of
$C(\boldsymbol w)$ for selecting the averaging weights. When $\phi_n=2$, after
subtracting the $\boldsymbol w$-free term $r_n$, the criterion is uniformly
asymptotically unbiased for $\mathrm{TEPE}(\boldsymbol w)$.
Moreover, as shown in
Corollary~\ref{cor:criterion-exact-unbiasedness}, in the idealized
setting where the working model is correctly specified and
$\widehat{\boldsymbol\Sigma}^{T}=\boldsymbol\Sigma^{T}$, both remainder
terms $d_n(\boldsymbol w)$ and $e_n(\boldsymbol w)$ vanish. Consequently, after subtracting $r_n$, the criterion is exactly conditionally unbiased
for $\mathrm{TEPE}(\boldsymbol w)$.

\subsection{Asymptotic optimality under model misspecification}
\label{subsec:Asymptotic-optimality-under-model-misspecification}

When the working linear model is misspecified, prediction accuracy under the
target distribution is the central objective. To establish the asymptotic optimality of the estimated weights, we replace Assumption~\ref{assump:5} with the following strengthened
rate condition.

\begin{assumption}
\label{assump:6}
$M_{n}^2/(n\xi_n^2)=o(1)$, $\phi_n/(n\xi_n)=o(1)$, and $\kappa_{n}=o_p(\xi_n)$.
\end{assumption}

Assumption~\ref{assump:6} strengthens Assumption~\ref{assump:5} by imposing an
additional rate restriction on the penalty level $\phi_n$ in the
weight-selection criterion. Since $M_{n}^2/(n\xi_n^2)=o(1)$ is still required, this
assumption implicitly requires $\xi_n\ne0$, corresponding to the misspecified
setting in which the candidate averaging class has nonzero minimum target
excess prediction error. The additional condition $\phi_n/(n\xi_n)=o(1)$
allows the penalty to be fixed or diverging, but requires it to grow slowly
enough relative to $n\xi_n$.

\begin{theorem}
\label{thm:asymptotic-optimality}
If Assumptions~\ref{assump:1}--\ref{assump:4} and
\ref{assump:6} hold, then $\widehat{\boldsymbol w}$ is asymptotically optimal
in the sense that
\begin{equation}
\label{eq:asymptotic-optimality-main}
\frac{
\mathrm{TEPE}(\widehat{\boldsymbol w})
}{
\inf_{\boldsymbol w\in\mathcal{W}_{n}}
\mathrm{TEPE}\left(\boldsymbol w\right)
}
\xrightarrow{p} 1.
\end{equation}
\end{theorem}

Theorem~\ref{thm:asymptotic-optimality} shows that the proposed model averaging
procedure is asymptotically optimal in the sense that its target excess
prediction error is asymptotically equivalent to that of the infeasible best
convex combination in the same candidate class. Since the best convex
combination attains TEPE no larger than that of the best single candidate, this
result highlights the potential advantage of averaging over candidate
correction levels rather than selecting a single exponent.

\subsection{Weight concentration under correct specification}
\label{subsec:weight-concentration-under-correct-specification}

We next examine the behavior of the AIWMA weights under correct specification.
We first establish that the OLS endpoint minimizes
$\mathrm{TEPE}(\boldsymbol w)$ over the entire class of AIWMA estimators. For each $m=1,\ldots,M_{n}$, let $\boldsymbol w^{(m)}\in\mathcal{W}_{n}$ denote the weight vector that places unit mass on the $m$th candidate and zero mass on all other candidates. In particular, $\boldsymbol w^{(1)}$
corresponds to the OLS endpoint. 

\begin{proposition}
\label{prop:correct-spec-OLS-optimal}
Suppose that the working linear model is correctly specified. Then, for every
$\boldsymbol w\in\mathcal{W}_{n}$,
\begin{equation}
\label{eq:prop-ols-optimal-convex}
\mathrm{TEPE}(\boldsymbol w)
-
\mathrm{TEPE}\left(\boldsymbol w^{(1)}\right)
=
\sigma^2
\operatorname{tr}
\left[
\boldsymbol\Sigma^{T}
\left\{
\boldsymbol L(\boldsymbol w)
-
\boldsymbol L_1
\right\}
\left\{
\boldsymbol L(\boldsymbol w)
-
\boldsymbol L_1
\right\}^{\top}
\right]
\ge 0 .
\end{equation}
\end{proposition}

Proposition~\ref{prop:correct-spec-OLS-optimal} shows that, under correct specification, the OLS endpoint
minimizes $\mathrm{TEPE}(\boldsymbol w)$ over $\boldsymbol w\in\mathcal{W}_{n}$.
We next investigate whether the criterion $C(\boldsymbol w)$ asymptotically
assigns its weight to candidate exponents near this endpoint.

Let $a_n\to0$ be a deterministic positive sequence, and define the
index set of candidate exponents near the OLS endpoint by $\mathcal N_n( a_n)
=
\{
m\in\{1,\ldots,M_{n}\}:\lambda_m\le a_n
\}$. For $\lambda\in[0,1]$, recall that
$\boldsymbol\Sigma(\lambda)
=
\E_{\boldsymbol x\sim\rho_{\mathcal X}^{S}}
[
\delta(\boldsymbol x)^\lambda
\boldsymbol h(\boldsymbol x)\boldsymbol h(\boldsymbol x)^\top
]$,
and define
$\boldsymbol\Gamma(\lambda)
=
\E_{\boldsymbol x\sim\rho_{\mathcal X}^{S}}
[
\delta(\boldsymbol x)^{1+\lambda}
\boldsymbol h(\boldsymbol x)\boldsymbol h(\boldsymbol x)^\top
]$,
$\nu(\lambda)
=
\operatorname{tr}
\{\boldsymbol\Sigma(\lambda)^{-1}
\boldsymbol\Gamma(\lambda)\}$,
and
$\eta_n
=
\min_{m\notin\mathcal N_n( a_n)}
\{\nu(\lambda_m)-\nu(0)\}$.
We impose the following additional conditions.

\begin{assumption}
\label{assump:7}
The basis functions satisfy
$\E_{\boldsymbol x\sim\rho_{\mathcal X}^{S}}\{\|\boldsymbol h(\boldsymbol x)\|_2^4\}<\infty$.
\end{assumption}

\begin{assumption}
\label{assump:8}
$\nu(\lambda_m)\ge\nu(0)$ for $m=1,\ldots,M_{n}$, and $\eta_n>0$.
\end{assumption}

\begin{assumption}
\label{assump:9}
$M_{n}=o(n^{1/2}\eta_n)$, $\kappa_n=o_p(\eta_n)$, and
$\phi_n\eta_n\to\infty$.
\end{assumption}

Assumption~\ref{assump:7} is the counterpart of
Assumption~\ref{assump:3} for the correctly specified case. When the working
linear model is correctly specified, the approximation error vanishes for every
correction exponent, that is, $b_\lambda(\boldsymbol x)=0$ for all
$\lambda\in[0,1]$. 
Assumption~\ref{assump:8} is an identification condition for the function
$\nu(\lambda)$. This condition formalizes
the variance advantage of the OLS endpoint under correct specification.
Sufficient conditions and illustrative examples under which
Assumption~\ref{assump:8} holds are provided in the Supplementary Material.
Assumption~\ref{assump:9} imposes joint rate restrictions involving the
sample size $n$, the number of candidate estimators $M_n$, the
density-ratio estimation error $\kappa_n$, the penalty level $\phi_n$,
and the separation level $\eta_n$. 


\begin{theorem}
\label{thm:weight-concentration}
If the working linear model is correctly specified and
Assumptions~\ref{assump:1}--\ref{assump:2}, \ref{assump:4}, and
\ref{assump:7}--\ref{assump:9} hold, then, as $n\to\infty$,
\begin{equation}
\label{eq:weight-concentration-main}
\sum_{m\notin\mathcal N_n( a_n)}
\widehat w_m
\xrightarrow{p}
0.
\end{equation}
Moreover, if there exists a constant $c>0$ such that $\lambda_2\ge c$ for all
sufficiently large $n$, then
\begin{equation}
\label{eq:weight-concentration-ols}
\widehat w_1
\xrightarrow{p}
1.
\end{equation}
\end{theorem}

Theorem~\ref{thm:weight-concentration} shows that, under correct
specification, the AIWMA weights concentrate on candidate exponents lying
in a shrinking neighborhood of the OLS endpoint $\lambda=0$.
Assumptions~\ref{assump:1}, \ref{assump:4}, and \ref{assump:7} imply
that $\nu(\lambda)$ is uniformly bounded over $\lambda\in[0,1]$, and
hence $\eta_n=O(1)$. Therefore, the condition
$\phi_n\eta_n\to\infty$ in Assumption~\ref{assump:9} necessarily
requires $\phi_n\to\infty$, thereby ruling out fixed choices such as
$\phi_n=2$. Thus, when the working linear model is correctly specified,
the diverging penalty prevents unnecessary importance-weighting
correction. In particular, if the nonzero candidate exponents are uniformly bounded away from zero, the result implies weight consistency for the OLS endpoint.
\section{Simulation studies}
\label{sec:Simulation-studies}

This section examines the finite-sample performance of the proposed AIWMA
procedure under covariate shift. 
We consider two candidate
exponent grids. The fixed grid is $\Lambda=\{0,0.1,\dots,0.9,1\}$, 
which yields $M=11$ candidate single-exponent AIWLS estimators. We also consider the diverging grid 
$\Lambda_n=\{\lambda_m=(m-1)/(M_n-1):m=1,\ldots,M_n\}$, where $M_n=\lceil 5n_s^{1/3}\rceil$ and $\lceil c\rceil$ denotes the smallest integer greater than or equal to c, to examine the robustness of our conclusions to the use of a denser candidate set.
Unless otherwise stated, density ratios are estimated by KLIEP using Gaussian kernels, with the bandwidth selected by five-fold cross-validation \citep{Sugiyamaetal2008}. The Supplementary Material provides implementation details and results for alternative density-ratio estimators. Because the true density ratio is known in the simulation designs, we also report oracle results to assess sensitivity to density-ratio estimation and the choice of estimator.

\subsection{Prediction performance under model misspecification}
\label{subsec:sim-prediction}

We first examine the predictive performance of AIWMA under the target
distribution when the working model is misspecified. 
The competing methods include the proposed AIWMA procedures with
$\phi_n=2$ and $\phi_n=\log n$ in the weight-selection criterion, denoted by
AIWMA-2 and AIWMA-$\log n$, respectively;  ordinary least squares
(OLS); standard importance-weighted least
squares (IWLS); the AIC-type
exponent-selection rule of \citet{Shimodaira2000} (AIW-AIC); its BIC-type
counterpart, obtained by replacing the AIC penalty coefficient with $\log n$
(AIW-BIC); the importance-weighted cross-validation rule of
\citet{SugiyamaKrauledatMuller2007} (AIW-CV); and two smooth
information-criterion averaging methods, denoted by AIW-SAIC and AIW-SBIC.
Specifically, these two averaging methods assign the normalized weight
$\exp(-\mathrm{IC}_m/2)\big/
\sum_{j=1}^{M_n}\exp(-\mathrm{IC}_j/2)$
to the $m$th candidate AIWLS estimator, where $\mathrm{IC}_m$ is the
corresponding AIC- or BIC-type criterion value for the candidate exponent
$\lambda_m$. 
We also include a clipped IWLS variant designed to mitigate variance inflation caused by extreme density-ratio values \citep{Gogolashvilietal2023}. In our implementation, the estimated density ratios are clipped at their empirical 90th percentile; the resulting method is denoted by IWLS-Clip.

For $\boldsymbol x=(x_1,\ldots,x_6)^\top$, we generate the response as  
\begin{equation*}
y
=
f(\boldsymbol x)+\varepsilon
=
1+x_1+0.8x_2+0.6x_3+0.4x_6
+
0.15(x_1^3+x_2^3)
+
0.1x_4x_5+\varepsilon,
\end{equation*}
where $\varepsilon\sim N(0,\sigma^2)$. 
The working model uses only the linear basis
$\boldsymbol h(\boldsymbol x)=(1,x_1,\ldots,x_6)^\top$ and is therefore
misspecified, as it omits the cubic and interaction terms in
$f(\boldsymbol x)$. 
The source covariate distribution is
$\rho_{\mathcal X}^{S}=N(\boldsymbol 0,1.5^2\boldsymbol I_6)$, whereas the
target covariate distribution is
$\rho_{\mathcal X}^{T}=N(\zeta\boldsymbol a,\boldsymbol I_6)$, with
$\boldsymbol a=(1,1,0,0,0,0)^\top$. The parameter $\zeta$ controls the
separation between the source and target covariate means, with larger values
corresponding to more severe covariate shift. We consider
$\zeta\in\{1.2,1.5\}$. 
For each value of $\zeta$, we calibrate $\sigma^2$ so that the
target-distribution coefficient of determination,
$R^2=\operatorname{Var}_{\boldsymbol{x}\sim\rho_{\mathcal X}^{T}}
\{f(\boldsymbol{x})\}/
[\operatorname{Var}_{\boldsymbol{x}\sim\rho_{\mathcal X}^{T}}
\{f(\boldsymbol{x})\}+\sigma^2]$,
takes each value in $\{0.1,0.2,\ldots,0.9\}$.
In each replication we generate an i.i.d.\ labeled
source sample
$\{(\boldsymbol x_i,y_i)\}_{i=1}^{n_s}
\sim\rho_{\mathcal X\times\mathcal Y}^{S}$,
with $n_s\in\{50,100,150,200\}$, together with an i.i.d.\ unlabeled target
covariate sample
$\{\boldsymbol x_i'\}_{i=1}^{n_t}\sim\rho_{\mathcal X}^{T}$
of size $n_t=1000$. All methods are fitted using the labeled source sample,
with the unlabeled target covariates additionally used, when required, to
estimate the density ratio and construct the weight-selection criterion. For
evaluation, we generate an independent target covariate sample
$\{\boldsymbol z_j\}_{j=1}^{n_{\mathrm{eval}}}$
from $\rho_{\mathcal X}^{T}$, with $n_{\mathrm{eval}}=1000$. Prediction
performance is assessed by the empirical target excess prediction error 
\begin{equation*}
\label{eq:sim-epr}
\mathrm{TEPE}(\widehat f)
=
\frac{1}{n_{\mathrm{eval}}}
\sum_{j=1}^{n_{\mathrm{eval}}}
\left\{
\widehat f(\boldsymbol z_j)-f(\boldsymbol z_j)
\right\}^{2},
\end{equation*}
where $f$ is the true regression function and $\widehat f(\boldsymbol z_j)$ is the prediction at the $j$th target evaluation point produced by the method being evaluated. All reported results are averaged over $B=1000$ Monte Carlo
replications. We normalize the empirical TEPE of each method by that of
AIWMA-$\log n$, so that a value below one indicates better prediction
performance than AIWMA-$\log n$, whereas a value above one indicates worse
performance. 

\begin{figure}[htbp]
\centering
\includegraphics[
width=0.90\linewidth,
height=0.90\textheight,
keepaspectratio
]{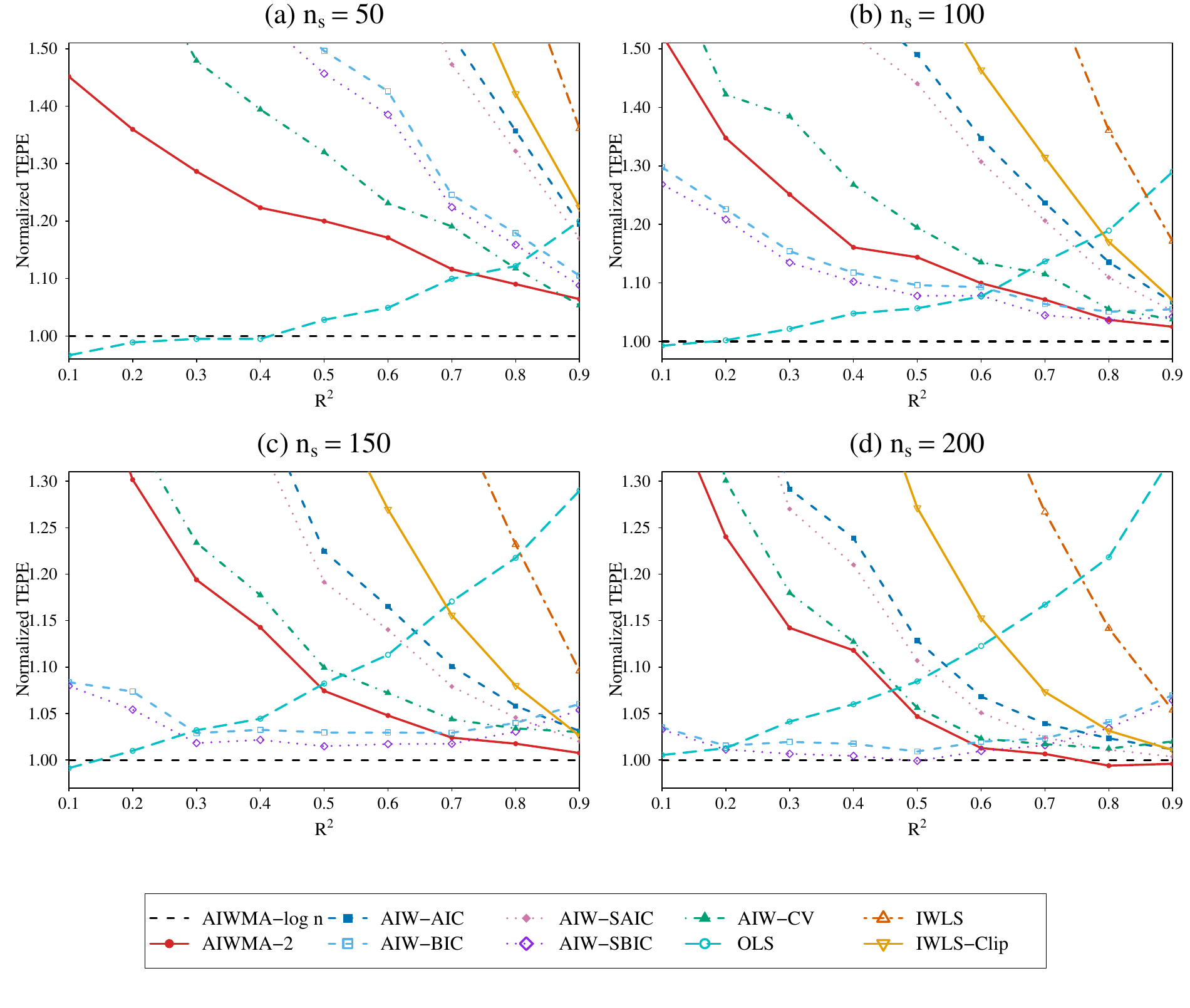}
\caption{Normalized target excess prediction error (TEPE) relative to
AIWMA-$\log n$ in the misspecification design, using KLIEP
density-ratio estimation and a fixed candidate grid with $\zeta=1.2$.}
\label{fig:multi-kliep-fixed-zeta12}

\par\smallskip
{\small
\noindent\textbf{Alt text:}
Four panels compare ten prediction methods over target $R^2$ values from
0.1 to 0.9 for source sample sizes $n_s=50$, 100, 150, and 200.
The vertical axis is the target excess prediction error normalized by that
of AIWMA-$\log n$; hence, the value for AIWMA-$\log n$ equals one by
construction, and smaller values indicate better relative performance.
AIWMA-$\log n$ remains competitive across the considered settings.
OLS is particularly competitive when the target $R^2$ and source sample
size are small, whereas methods applying stronger importance-weighting
correction become more competitive as the target $R^2$ increases.
\par}
\end{figure}

\begin{figure}[htbp]
\centering
\includegraphics[
width=0.90\linewidth,
height=0.90\textheight,
keepaspectratio
]{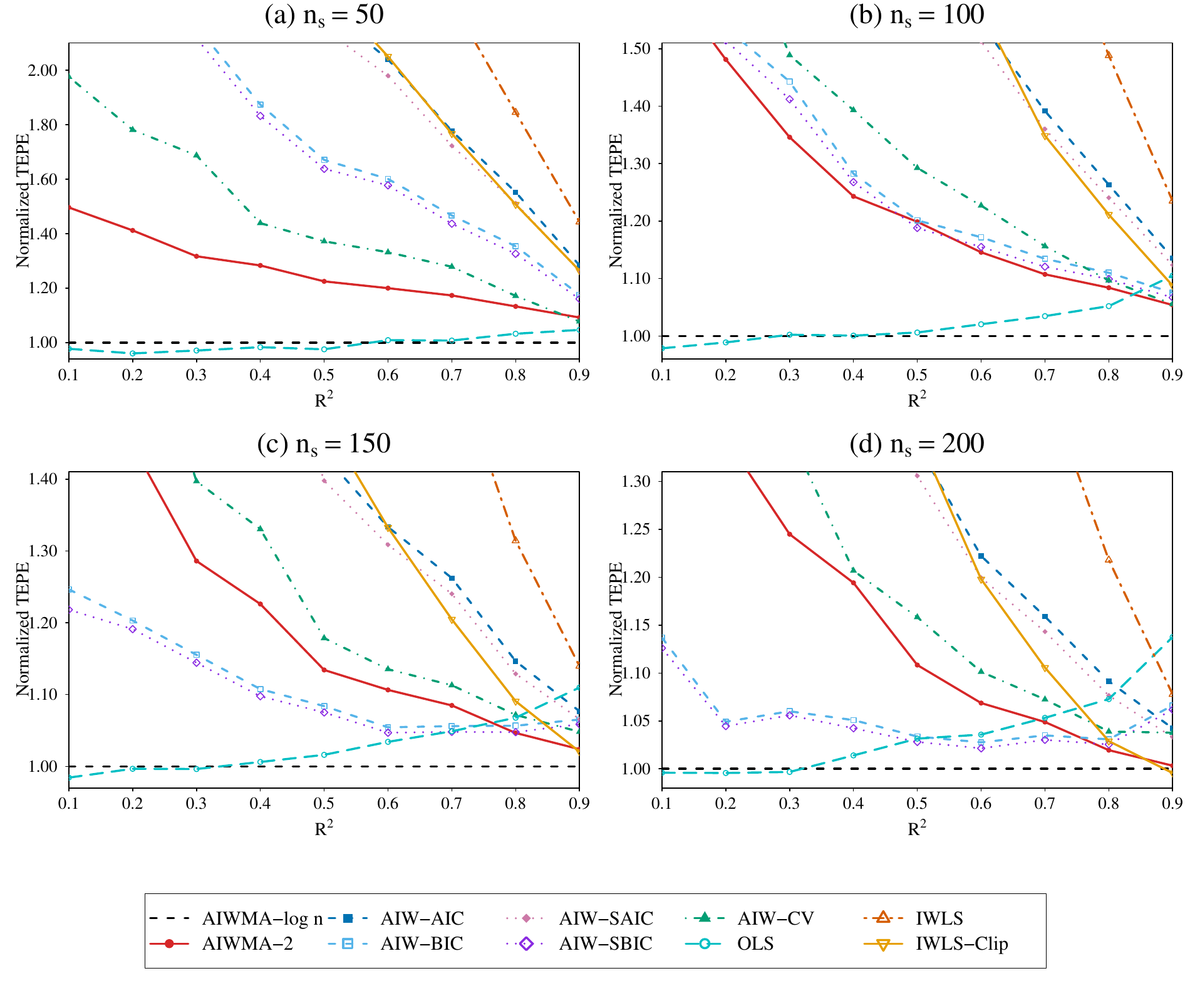}
\caption{Normalized target excess prediction error (TEPE) relative to
AIWMA-$\log n$ in the misspecification design, using KLIEP
density-ratio estimation and a fixed candidate grid with $\zeta=1.5$.}
\label{fig:multi-kliep-fixed-zeta15}

\par\smallskip
{\small
\noindent\textbf{Alt text:}
Four panels compare ten prediction methods over target $R^2$ values from
0.1 to 0.9 for source sample sizes $n_s=50$, 100, 150, and 200.
The vertical axis is the target excess prediction error normalized by that
of AIWMA-$\log n$; hence, the value for AIWMA-$\log n$ equals one by
construction, and smaller values indicate better relative performance.
AIWMA-$\log n$ remains competitive across the considered settings.
OLS is particularly competitive when the target $R^2$ and source sample
size are small, whereas methods applying stronger importance-weighting
correction become more competitive as the target $R^2$ increases. 
\par}
\end{figure}

Figures~\ref{fig:multi-kliep-fixed-zeta12}--\ref{fig:multi-kliep-fixed-zeta15}
present the results for the misspecification design using
KLIEP density-ratio estimation and the fixed candidate grid.
Across the considered source sample sizes, target $R^2$ values, and shift
strengths, AIWMA-$\log n$ remains competitive relative to the alternative
methods. OLS is particularly competitive when $R^2$ and $n_s$ are small.
In this low signal-to-noise regime, the variance reduction achieved by
avoiding density-ratio weighting can outweigh the bias reduction obtained by
adjusting for the target covariate distribution.
As $R^2$ increases, however, the response noise decreases, making the
approximation bias of the unweighted working model more consequential.
Consequently, the relative performance of OLS deteriorates because it does
not adjust for the target covariate distribution.
Candidates with larger exponents, by contrast, become more competitive
because stronger correction toward the target distribution becomes more
valuable relative to its associated variance cost.
By averaging across candidate exponents, AIWMA-$\log n$ adapts to this
transition without committing to a single correction level and remains
competitive throughout the considered range of target $R^2$ values.
The corresponding diverging-grid results, reported in the Supplementary
Material, are qualitatively similar, indicating that the main conclusions
are not driven by a particular discretization of the exponent interval.
The Supplementary Material also reports results based on alternative
density-ratio estimators, the oracle density ratio, and a one-dimensional
illustrative example, all of which lead to broadly similar conclusions. 

\subsection{Weight concentration under correct specification}
\label{subsec:sim-weight-consistency}

We next examine whether the AIWMA criterion assigns increasing weight to the
OLS endpoint when the working model is correctly specified. The data-generating
process is
\begin{equation*}
y
=
f(\boldsymbol x)+\varepsilon
=
1+1.5x_1-x_2+\varepsilon,
\qquad
\varepsilon\sim N(0,\sigma^2).
\end{equation*}
The working basis is correctly specified as
$\boldsymbol h(\boldsymbol x)=(1,x_1,x_2)^\top$. 
The source covariates are generated from
$\rho_{\mathcal X}^{S}
=
N(\boldsymbol 0,1.5^2\boldsymbol I_2)$,
whereas the target covariates are generated from
$\rho_{\mathcal X}^{T}
=
N(\boldsymbol\mu,\boldsymbol I_2)$. We set
$\boldsymbol\mu=\zeta(1,1)^\top$ with
$\zeta\in\{0.5,1.0,1.5\}$, so that larger values of $\zeta$ correspond to
larger separation between the source and target covariate means.
To vary the signal-to-noise level, we consider target-distribution
coefficients of determination $R^2\in\{0.3,0.5,0.7\}$ and, for each
specified value of $R^2$, choose the error variance accordingly, as in the
misspecification experiments.  
The labeled source sample sizes are
$n_s\in\{25,100,200,400,800,1600\}$, and the unlabeled target covariate sample
size is fixed at $n_t=2000$.
Motivated by Theorem~\ref{thm:weight-concentration}, we focus on the AIWMA
criterion with $\phi_n=\log n$. We report results using
KLIEP density-ratio estimation under both the fixed and diverging candidate
exponent grids. All results are averaged over $B=1000$
Monte Carlo replications. The primary diagnostic is the average weight assigned
to the OLS endpoint, denoted by $w_{\lambda=0}$.

\begin{figure}[htbp]
\centering
\includegraphics[
width=0.90\linewidth,
height=0.90\textheight,
keepaspectratio
]{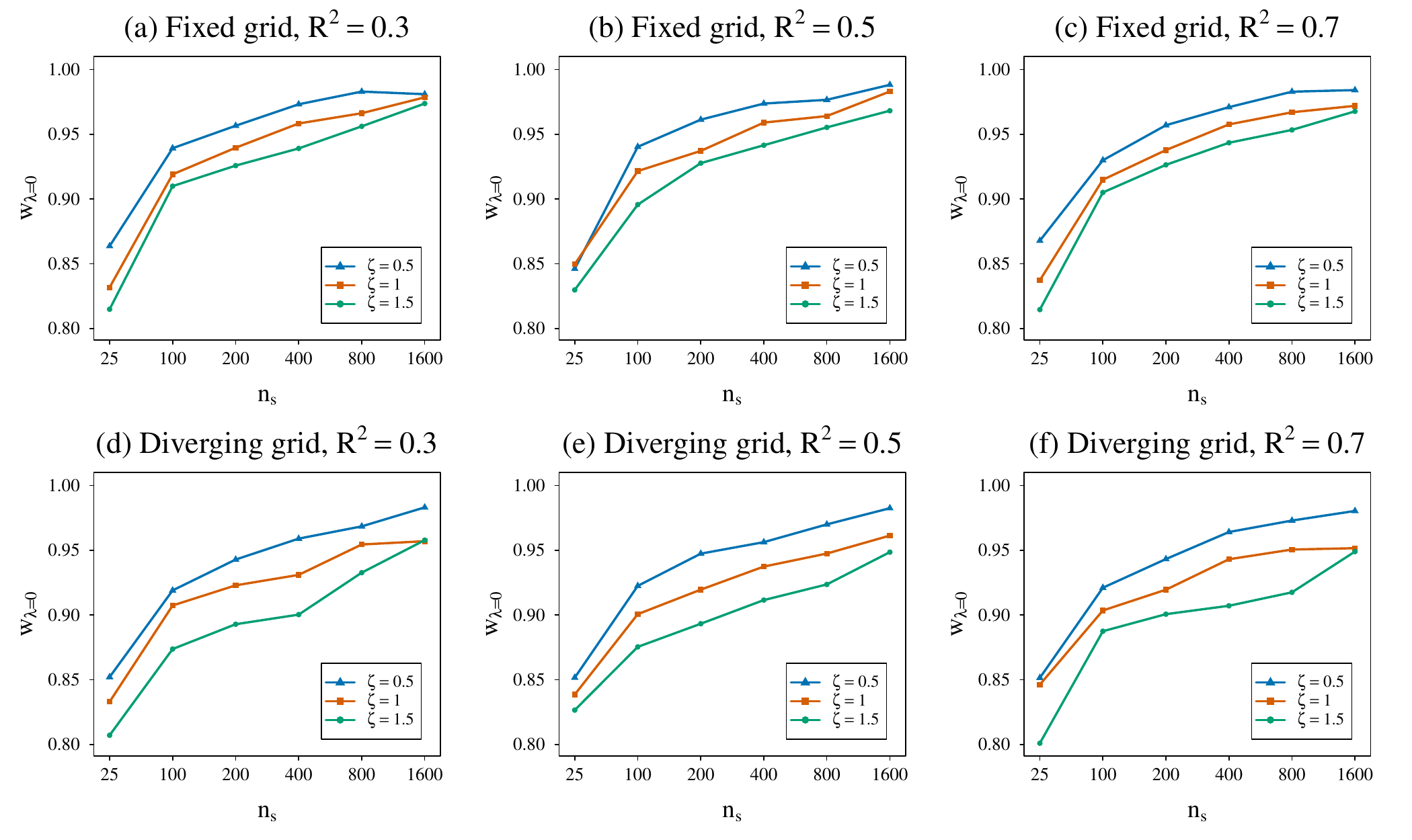}
\caption{Average estimated weight assigned to the OLS endpoint under correct
specification using KLIEP density-ratio estimation. The upper and lower rows
correspond to the fixed and diverging candidate exponent grids, respectively.}
\label{fig:weight-concentration-kliep}

\par\smallskip
\begin{minipage}{0.90\linewidth}
\small
\raggedright
\noindent\textbf{Alt text:}
Six panels show the average estimated weight assigned to the OLS endpoint as
the source sample size increases from 25 to 1600. The three columns correspond
to target $R^2$ values of 0.3, 0.5, and 0.7, while the upper and lower rows
correspond to the fixed and diverging candidate exponent grids. Each panel
contains curves for shift parameters $\zeta=0.5$, 1.0, and 1.5. The average
OLS weight generally increases toward one with the source sample size.
Concentration is stronger under milder shifts, while the diverging grid assigns
slightly less weight exactly to the OLS endpoint because some weight is
distributed among exponents close to zero.
\end{minipage}
\end{figure}

Figure~\ref{fig:weight-concentration-kliep} shows a pattern of weight concentration near the OLS endpoint. Under the fixed candidate exponent grid,
the average OLS weight generally increases with $n_s$ across all three shift
configurations and moves toward one as the source sample size grows. The concentration is stronger under the milder shift and weaker under the stronger shift, although the gap narrows as $n_s$ increases. One possible explanation is that stronger shifts make the IWLS endpoint less stable, which may reduce the accuracy of the weight-selection criterion; larger source sample sizes are then needed to compensate for this loss. 
Across the three signal-to-noise settings considered, $w_{\lambda=0}$
consistently increases with $n_s$ and moves toward one. 
The diverging-grid results show the same qualitative pattern. Relative to the fixed-grid case, the average weight assigned exactly to
$\lambda=0$ is slightly smaller for the same shift level, target $R^2$ level,
and source sample size, because the denser exponent grid allows weight to be
distributed among candidates close to the OLS endpoint. Nevertheless, $w_{\lambda=0}$ still increases with $n_s$ and remains close to one for large sample sizes. Additional results using
alternative density-ratio estimators and the oracle density ratio are reported in the Supplementary
Material and lead to the same qualitative conclusions. Taken together, the simulations support the two roles of AIWMA. Under
misspecification, averaging over correction levels stabilizes target prediction
relative to selecting a single exponent or using standard  importance weighting.
Under correct specification, the diverging-penalty version increasingly favors
the OLS endpoint, reducing unnecessary variance inflation from importance
weighting.

\section{Real-data analysis}
\label{sec:Real-data-analysis}

We further evaluate AIWMA using U.S. Natality public-use files from the National Vital Statistics System, maintained by the National Center for Health Statistics \citep{AbrevayaDahl2008,GrunebaumChervenak2026}. These annual files contain individual-level information derived from U.S. birth certificates, including infant birth outcomes, maternal demographic characteristics, health behaviors, prenatal-care utilization, pregnancy-related risk factors, and infant characteristics. 
We use infant birth weight, measured in kilograms, as the response. The working covariates include maternal age, body mass index, number of prenatal visits, numbers of cigarettes smoked before pregnancy and during the first trimester, gestational weight gain, the obstetric estimate of gestational age, maternal education, race or ethnicity, maternal nativity, marital status, parity, timing of prenatal-care initiation, participation in the Special Supplemental Nutrition Program for Women, Infants, and Children, pre-pregnancy diabetes, gestational diabetes, pre-pregnancy hypertension, gestational hypertension, previous preterm birth, and infant sex. We adopt a linear working model with an intercept, $\boldsymbol h(\boldsymbol x)=(1,\boldsymbol x^\top)^\top$, where $\boldsymbol x$ denotes the vector of standardized covariates listed above. To construct a temporal source--target split, we treat the 2016 records as the source population and the 2024 records as the target population. Further details on data preprocessing and covariate construction are provided in the Supplementary Material.

In each replication, we draw $n_s\in \{50,100,150,200\}$ labeled source observations, $n_t=1000$ unlabeled target observations, and an independent labeled target evaluation sample of size $n_{\mathrm{eval}}=1000$. Results are averaged over $B=1000$ replications. The density-ratio estimation methods, candidate exponent grids, and competing procedures are the same as those used in Section~\ref{sec:Simulation-studies}. Because the true regression function is unavailable, predictive accuracy is evaluated using the empirical target prediction error defined by 
\begin{equation*}
\label{eq:natality-tpe}
\mathrm{TPE}^{(\ell)}(g)
=
\frac{1}{n_{\mathrm{eval}}}
\sum_{j=1}^{n_{\mathrm{eval}}}
\left\{
y_j^{(\ell)}
-
\widehat f_g^{(\ell)}
\bigl(\boldsymbol x_j^{(\ell)}\bigr)
\right\}^{2},
\end{equation*}
where
$\{(\boldsymbol x_j^{(\ell)},y_j^{(\ell)})\}_{j=1}^{n_{\mathrm{eval}}}$
is the labeled target evaluation sample in replication $\ell$, and
$\widehat f_g^{(\ell)}$ is the prediction rule fitted by method $g$. 
We report the average TPE gap relative to the best-performing method in each
replication:
\begin{equation*}
\label{eq:natality-delta-tpe}
\Delta\mathrm{TPE}(g)
=
\frac{1}{B}
\sum_{\ell=1}^{B}
\left[
\mathrm{TPE}^{(\ell)}(g)
-
\min_{g'\in\mathcal G}
\mathrm{TPE}^{(\ell)}(g')
\right],
\end{equation*} 
where $\mathcal G$ denotes the collection of competing methods. Smaller values
indicate better target prediction performance, with values close to zero
indicating that a method performs nearly as well as the best procedure on
average across replications. 
Table~\ref{tab:delta-tpe-kliep-natality} shows that AIWMA-$\log n$
attains the smallest average $\Delta\mathrm{TPE}$ across all reported source sample sizes under both candidate grids. Consistent with the simulation results, AIWMA-$\log n$
therefore exhibits favorable target prediction performance under covariate
shift.

\begin{table}[htbp]
\caption{Values of $10^2\times\Delta\mathrm{TPE}$ based on KLIEP
density-ratio estimation under the fixed and diverging candidate grids.
Within each row, boldface and underlining indicate the smallest and
second-smallest values, respectively.}
\label{tab:delta-tpe-kliep-natality}
\centering
\footnotesize
\setlength{\tabcolsep}{2.8pt}
\renewcommand{\arraystretch}{1.08}

\begin{tabular*}{\textwidth}{
  @{\extracolsep{\fill}}
  l
  *{10}{r}
  @{}
}
\toprule
&
\multicolumn{10}{c}{Method}
\\
\cmidrule(lr){2-11}
$n_s$
& AIWMA-$\log n$
& AIWMA-2
& AIW-AIC
& AIW-BIC
& AIW-SAIC
& AIW-SBIC
& AIW-CV
& OLS
& IWLS
& IWLS-Clip
\\
\midrule

\multicolumn{11}{@{}l}{\textit{Panel A: Fixed grid}}
\\[2pt]

50
& \textbf{1.0310}
& \underline{1.0672}
& 1.2711
& 1.2594
& 1.1487
& 1.1610
& 1.0996
& 1.1170
& 1.4232
& 1.4205
\\

100
& \textbf{0.2170}
& 0.2366
& 0.2770
& 0.2630
& 0.2518
& 0.2516
& \underline{0.2341}
& 0.2608
& 0.3494
& 0.3474
\\

150
& \textbf{0.0774}
& \underline{0.0820}
& 0.1195
& 0.1069
& 0.1109
& 0.1054
& 0.1052
& 0.1134
& 0.1829
& 0.1821
\\

200
& \textbf{0.0506}
& \underline{0.0533}
& 0.0719
& 0.0678
& 0.0637
& 0.0630
& 0.0704
& 0.0876
& 0.0986
& 0.0968
\\

\addlinespace[4pt]
\multicolumn{11}{@{}l}{\textit{Panel B: Diverging grid}}
\\[2pt]

50
& \textbf{1.0328}
& \underline{1.0660}
& 1.2723
& 1.2419
& 1.1429
& 1.1552
& 1.1041
& 1.1176
& 1.4238
& 1.4212
\\

100
& \textbf{0.2170}
& 0.2365
& 0.2784
& 0.2646
& 0.2509
& 0.2508
& \underline{0.2345}
& 0.2612
& 0.3499
& 0.3478
\\

150
& \textbf{0.0774}
& \underline{0.0820}
& 0.1198
& 0.1064
& 0.1107
& 0.1055
& 0.1049
& 0.1134
& 0.1830
& 0.1822
\\

200
& \textbf{0.0504}
& \underline{0.0535}
& 0.0721
& 0.0674
& 0.0639
& 0.0631
& 0.0705
& 0.0879
& 0.0989
& 0.0960
\\

\bottomrule
\end{tabular*}
\end{table}

\section{Concluding remarks}
\label{sec:Conclusion-Remarks}

This paper develops an adaptive importance-weighted model-averaging method for
prediction under covariate shift. By averaging candidate AIWLS estimators
associated with different degrees of
importance-weighting correction, the proposed method provides
a data-driven compromise between the stability of unweighted estimation and
the bias correction afforded by importance weighting. Theoretical results
establish asymptotic optimality under model misspecification and concentration
of the averaging weights near the OLS endpoint under correct specification. Simulation studies and a real-data analysis provide
empirical support for these theoretical results and illustrate the competitive 
finite-sample performance of the proposed method.

The present results also lay the groundwork for several directions for future research. The current analysis is restricted to squared-error loss and fixed-dimensional
linear prediction models. Extending the proposed framework to more general loss
functions, regularized estimators, and flexible prediction methods would require
new risk expansions and corresponding criteria for determining the averaging
weights. We also focus on a single-source, single-target setting in which no
labeled target observations are available. Extensions to partially labeled
target samples and multiple source populations with heterogeneous relevance to
the target distribution warrant further investigation. Finally, the candidate
family considered here varies only the exponent applied to the density-ratio
correction. A broader model-averaging framework could additionally account for
uncertainty in basis complexity, regularization parameters, and density-ratio estimation. 




\section*{Supplementary Material}
\label{sec:Supplementary-Material}

The Supplementary Material, included after the References in this arXiv version, contains an illustrative example of the bias--variance trade-off induced by the correction exponent, further discussion of the regularity conditions, proofs of the theoretical results, implementation details for the numerical studies, additional simulation results, and data-processing details for the real-data analysis.




\bibliographystyle{plainnat}
\bibliography{reference}

\clearpage

\setcounter{section}{0}
\setcounter{subsection}{0}
\setcounter{equation}{0}
\setcounter{figure}{0}
\setcounter{table}{0}

\renewcommand{\thesection}{S\arabic{section}}
\renewcommand{\thesubsection}{\thesection.\arabic{subsection}}

\numberwithin{equation}{section}
\renewcommand{\theequation}{\thesection.\arabic{equation}}

\renewcommand{\thefigure}{S\arabic{figure}}
\renewcommand{\thetable}{S\arabic{table}}

\makeatletter
\@ifundefined{c@theorem}{}{%
  \numberwithin{theorem}{section}%
  \renewcommand{\thetheorem}{\thesection.\arabic{theorem}}%
}
\@ifundefined{c@lemma}{}{%
  \numberwithin{lemma}{section}%
  \renewcommand{\thelemma}{\thesection.\arabic{lemma}}%
}
\@ifundefined{c@proposition}{}{%
  \numberwithin{proposition}{section}%
  \renewcommand{\theproposition}{\thesection.\arabic{proposition}}%
}
\@ifundefined{c@corollary}{}{%
  \numberwithin{corollary}{section}%
  \renewcommand{\thecorollary}{\thesection.\arabic{corollary}}%
}
\@ifundefined{c@assumption}{}{%
  \numberwithin{assumption}{section}%
  \renewcommand{\theassumption}{\thesection.\arabic{assumption}}%
}
\@ifundefined{c@definition}{}{%
  \numberwithin{definition}{section}%
  \renewcommand{\thedefinition}{\thesection.\arabic{definition}}%
}
\@ifundefined{c@example}{}{%
  \numberwithin{example}{section}%
  \renewcommand{\theexample}{\thesection.\arabic{example}}%
}
\@ifundefined{c@remark}{}{%
  \numberwithin{remark}{section}%
  \renewcommand{\theremark}{\thesection.\arabic{remark}}%
}
\makeatother

\begin{center}
{\Large\bfseries Supplementary Material for\\[3pt]
``Handling covariate shift by model averaging''\par}
\vspace{8pt}
{\normalsize Yifan Zhang, Tianfa Xie and Xinyu Zhang\par}
\end{center}
\vspace{12pt}

\noindent

This Supplementary Material includes an illustrative example of the bias--variance trade-off induced by the correction exponent, further discussion of the regularity conditions, proofs of the theoretical results, implementation details for the numerical studies, additional simulation results, and data-processing details for the real-data analysis. 

For notational simplicity, unless otherwise stated, we write $M=M_n$ and $\mathcal W=\mathcal W_n$ throughout the Supplementary Material. For a matrix $\boldsymbol A$, let $\Vert \boldsymbol A\Vert_F$, $\Vert \boldsymbol{A} \Vert$ denote its Frobenius norm and spectral norm, respectively. For a vector $\boldsymbol v$, let $\|\boldsymbol v\|_2$ denote its
Euclidean norm, and, for a symmetric matrix $\boldsymbol A$, let
$\lambda_{\min}(\boldsymbol A)$ denote its smallest eigenvalue.
Throughout, $\xrightarrow{p}$ and $\xrightarrow{d}$ denote convergence
in probability and convergence in distribution, respectively. 

\section{An illustrative example of the bias--variance trade-off induced by the correction exponent}
\label{app:tradeoff-example}

This section presents an illustrative example showing how the correction
exponent $\lambda$ induces a bias--variance trade-off. Consider the
one-dimensional regression model
\begin{equation*}
y=f(x)+\varepsilon=\operatorname{sinc}(x)+\varepsilon,
\qquad
\varepsilon\sim N(0,0.1^2),
\end{equation*}
where $\operatorname{sinc}(x)=\sin(x)/x$ for $x\ne0$ and
$\operatorname{sinc}(0)=1$.
The source covariates are generated from
$\rho_{\mathcal X}^{S}=N(1,0.5^2)$, whereas the target covariates are generated
from $\rho_{\mathcal X}^{T}=N(2,0.25^2)$. The working model uses the linear
basis $\boldsymbol h(x)=(1,x)^\top$ and is therefore misspecified. For a labeled
source sample $\{(x_i,y_i)\}_{i=1}^{n_s}$ with $n_s=50$, we compute the AIWLS
estimator for each $\lambda\in[0,1]$ using the known target-to-source covariate
density ratio $\delta(x)$. The endpoints $\lambda=0$ and $\lambda=1$
correspond to OLS and standard importance-weighted least squares, respectively.
To assess target prediction accuracy, we generate an independent target
evaluation sample $\{z_j\}_{j=1}^{n_{\mathrm{eval}}}$ with
$n_{\mathrm{eval}}=1000$ and $z_j\sim\rho_{\mathcal X}^{T}$. For each
$\lambda$, we compute the empirical target excess prediction error (TEPE)
\begin{equation*}
\mathrm{TEPE}(\lambda)
=
\frac{1}{n_{\mathrm{eval}}}
\sum_{j=1}^{n_{\mathrm{eval}}}
\left\{
\boldsymbol h(z_j)^\top
\widehat{\boldsymbol\theta}_{\mathrm{AIWLS}}(\lambda)
-
f(z_j)
\right\}^2 .
\end{equation*}

The resulting TEPE curve is shown in
Figure~\ref{fig:tradeoff-example}. In this example, the minimum is attained at
an intermediate correction level rather than at either endpoint. Thus, neither
unweighted fitting nor standard importance weighting gives the smallest
empirical target excess prediction error in this finite sample. This example
illustrates that, under covariate shift, the practical question is not only
whether to correct for the source--target covariate discrepancy, but also how
strongly the correction should be applied.

\begin{figure}[H]
\centering
\includegraphics[width=0.9\linewidth]{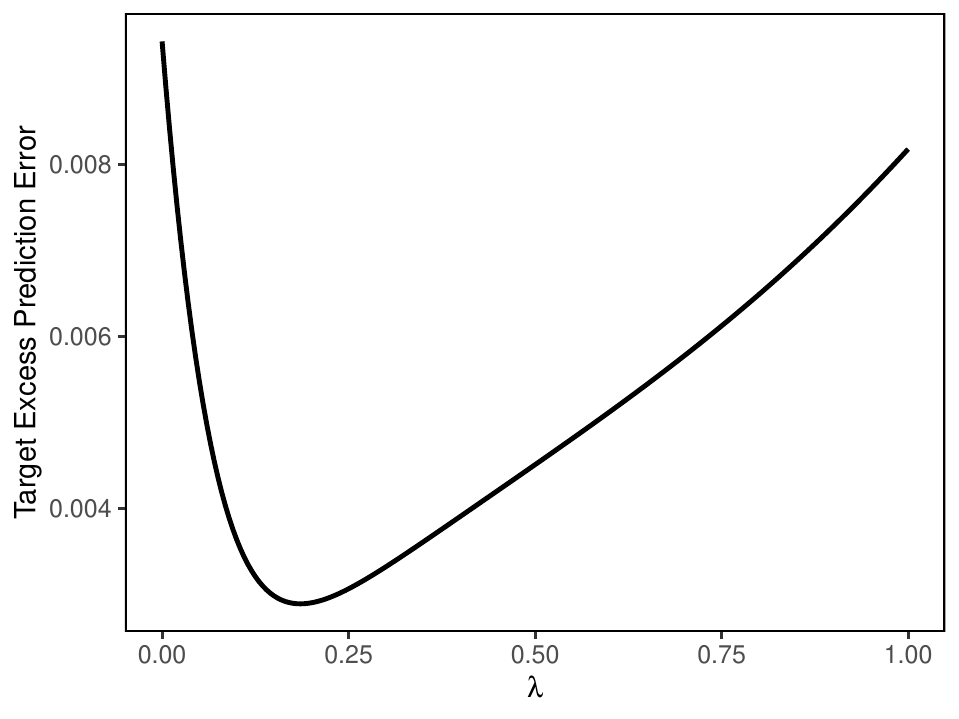}
\caption{TEPE versus the correction exponent $\lambda$. The
endpoints $\lambda=0$ and $\lambda=1$ correspond to OLS and standard
importance-weighted least squares, respectively.}
\label{fig:tradeoff-example}
\end{figure}

\section{Further discussion of regularity conditions}
\label{sec:supp-assumptions}

Assumption 8 is an identification condition for the OLS endpoint over the
candidate exponent set. It requires $\nu(\lambda_m)$ to be no smaller at any
candidate exponent than at $\lambda_1=0$, and requires a positive gap for
candidate exponents outside a shrinking neighborhood of zero. A sufficient
condition for Assumption 8 is that $\nu(\lambda)$ is strictly increasing on
$[0,1]$. In this case, for any candidate exponent set with $\lambda_1=0$,
\begin{equation*}
\nu(\lambda_m)\ge \nu(0),\qquad m=1,\ldots,M,
\end{equation*}
and
\begin{equation*}
\eta_n
=
\min_{m\notin\mathcal N_n( a_{n})}
\{\nu(\lambda_m)-\nu(0)\}
>0
\end{equation*}
whenever $\mathcal N_n( a_{n})\ne\{1,\ldots,M\}$. The following
Gaussian example gives a sufficient condition under which this monotonicity
holds. Suppose that
\begin{equation*}
\rho_{\mathcal X}^{S}
=
N(\boldsymbol 0,\sigma_S^2\boldsymbol I_d),
\qquad
\rho_{\mathcal X}^{T}
=
N(\boldsymbol\mu,\sigma_T^2\boldsymbol I_d),
\end{equation*}
where $0<\sigma_T^2<\sigma_S^2$, and let
$\boldsymbol h(\boldsymbol x)=(1,\boldsymbol x^\top)^\top$. Define
\begin{equation*}
r=\frac{\sigma_S^2}{\sigma_T^2}>1,
\qquad
\Delta_\mu=\frac{\|\boldsymbol\mu\|_2^2}{\sigma_T^2},
\qquad
D_\lambda=1+\lambda(r-1),
\qquad
E_\lambda=r+\lambda(r-1).
\end{equation*}
A direct calculation yields
\begin{equation*}
\begin{aligned}
\nu(\lambda)
&=
\left(
\frac{rD_\lambda}{E_\lambda}
\right)^{d/2}
\exp\left[
\frac{\Delta_\mu}{2}
\left\{
\frac{\lambda(1+\lambda)}{E_\lambda}
+
\frac{\lambda(1-\lambda)}{D_\lambda}
\right\}
\right] \\
&\quad \times
\left[
1
+
d\frac{D_\lambda}{E_\lambda}
+
\frac{r\Delta_\mu}{D_\lambda E_\lambda^2}
\right].
\end{aligned}
\end{equation*}
Therefore, a sufficient condition for Assumption 8 is
\begin{equation*}
\inf_{\lambda\in[0,1]}
\frac{\partial}{\partial\lambda}\log\nu(\lambda)>0.
\end{equation*}
Under this condition, $\nu(\lambda)$ is strictly increasing on $[0,1]$, and
hence Assumption 8 holds for any candidate exponent set. When
$\Delta_\mu=0$, the above expression reduces to
\begin{equation*}
\nu(\lambda)
=
\left(
\frac{rD_\lambda}{E_\lambda}
\right)^{d/2}
\left[
1+d\frac{D_\lambda}{E_\lambda}
\right].
\end{equation*}
Since $D_\lambda/E_\lambda$ is strictly increasing in $\lambda$ whenever
$r>1$, $\nu(\lambda)$ is strictly increasing on $[0,1]$. Moreover,
$\partial\log\nu(\lambda)/\partial\lambda$ is continuous in
$(\Delta_\mu,\lambda)$. Therefore, for each fixed $(r,d)$, there exists a
constant $\Delta_0=\Delta_0(r,d)>0$ such that
\begin{equation*}
\inf_{\lambda\in[0,1]}
\frac{\partial}{\partial\lambda}\log\nu(\lambda)>0
\qquad
\text{whenever }0\le \Delta_\mu<\Delta_0 .
\end{equation*}
Consequently, Assumption 8 holds whenever
\begin{equation*}
\frac{\|\boldsymbol\mu\|_2^2}{\sigma_T^2}<\Delta_0(r,d).
\end{equation*}
This example covers a class of covariate shifts in which the source covariate
distribution is more dispersed than the target covariate distribution and the
standardized squared mean difference
$\|\boldsymbol\mu\|_2^2/\sigma_T^2$ is sufficiently small.

\section{Proofs of the theoretical results}
\label{sec:supp-proofs}

\subsection{Proof of Theorem~1}

Fix an arbitrary $\lambda\in[0,1]$. By (9),
\begin{equation*}
\widehat{\boldsymbol\theta}^{\mathrm{AIWLS}}(\lambda)
=
(\boldsymbol H^\top\boldsymbol\Delta^\lambda\boldsymbol H)^{-1}
\boldsymbol H^\top\boldsymbol\Delta^\lambda\boldsymbol y .
\end{equation*}
Using
\begin{equation*}
y_i
=
f(\boldsymbol x_i)+\varepsilon_i
=
\boldsymbol h(\boldsymbol x_i)^\top\boldsymbol\theta^*(\lambda)
+
b_\lambda(\boldsymbol x_i)
+
\varepsilon_i ,
\end{equation*}
we obtain
\begin{equation}
\label{eq:proof-thm1-basic-expansion}
\sqrt{n_s}
\left\{
\widehat{\boldsymbol\theta}^{\mathrm{AIWLS}}(\lambda)
-
\boldsymbol\theta^*(\lambda)
\right\}
=
\boldsymbol S_{n_s}(\lambda)^{-1}
\boldsymbol U_{n_s}(\lambda),
\end{equation}
where
\begin{equation*}
\boldsymbol S_{n_s}(\lambda)
=
\frac{1}{n_s}
\boldsymbol H^\top\boldsymbol\Delta^\lambda\boldsymbol H
=
\frac{1}{n_s}
\sum_{i=1}^{n_s}
\delta(\boldsymbol x_i)^\lambda
\boldsymbol h(\boldsymbol x_i)\boldsymbol h(\boldsymbol x_i)^\top
\end{equation*}
and
\begin{equation*}
\boldsymbol U_{n_s}(\lambda)
=
\frac{1}{\sqrt{n_s}}
\sum_{i=1}^{n_s}
\delta(\boldsymbol x_i)^\lambda
\boldsymbol h(\boldsymbol x_i)
\{b_\lambda(\boldsymbol x_i)+\varepsilon_i\}.
\end{equation*}

We first show that $\boldsymbol S_{n_s}(\lambda)$ converges to
$\boldsymbol\Sigma(\lambda)$. For each $j,k=1,\ldots,p$,
\begin{equation*}
\left|
\delta(\boldsymbol x)^\lambda h_j(\boldsymbol x)h_k(\boldsymbol x)
\right|
\le
\delta(\boldsymbol x)^\lambda
\|\boldsymbol h(\boldsymbol x)\|_2^2
\le
\{1+\delta(\boldsymbol x)^{2\lambda}\}
\|\boldsymbol h(\boldsymbol x)\|_2^2 .
\end{equation*}
The right-hand side is integrable by the moment conditions in Theorem~1.
Since $p$ is fixed, the weak law of large numbers applied entrywise gives
\begin{equation}
\label{eq:proof-thm1-gram-convergence}
\boldsymbol S_{n_s}(\lambda)
\xrightarrow{p}
\boldsymbol\Sigma(\lambda).
\end{equation}
By the nonsingularity condition in Theorem~1,
$\boldsymbol\Sigma(\lambda)$ is nonsingular. Therefore,
$\boldsymbol S_{n_s}(\lambda)$ is nonsingular with probability tending to one,
and the continuous mapping theorem gives
\begin{equation}
\label{eq:proof-thm1-inverse-convergence}
\boldsymbol S_{n_s}(\lambda)^{-1}
\xrightarrow{p}
\boldsymbol\Sigma(\lambda)^{-1}.
\end{equation}

It remains to derive the limiting distribution of
$\boldsymbol U_{n_s}(\lambda)$. Define
\begin{equation*}
\boldsymbol\xi_i(\lambda)
=
\delta(\boldsymbol x_i)^\lambda
\boldsymbol h(\boldsymbol x_i)
\{b_\lambda(\boldsymbol x_i)+\varepsilon_i\},
\qquad i=1,\ldots,n_s .
\end{equation*}
The random vectors $\boldsymbol\xi_i(\lambda)$ are independent and
identically distributed. Since $\boldsymbol\theta^*(\lambda)$ is defined by
(10), differentiating the objective in (10) at its minimizer gives
\begin{equation}
\label{eq:proof-thm1-normal-equation}
\E_{\boldsymbol x\sim\rho_{\mathcal X}^{S}}
\left[
\delta(\boldsymbol x)^\lambda
\boldsymbol h(\boldsymbol x)b_\lambda(\boldsymbol x)
\right]
=
\boldsymbol0 .
\end{equation}
The expectation in \eqref{eq:proof-thm1-normal-equation} is well defined
because, by the Cauchy–Schwarz inequality and the moment conditions in Theorem~1,
\begin{equation*}
\E_{\boldsymbol x\sim\rho_{\mathcal X}^{S}}
\left[
\delta(\boldsymbol x)^\lambda
\|\boldsymbol h(\boldsymbol x)\|_2
|b_\lambda(\boldsymbol x)|
\right]
\le
\left\{
\E_{\boldsymbol x\sim\rho_{\mathcal X}^{S}}
\left[
\delta(\boldsymbol x)^{2\lambda}
b_\lambda(\boldsymbol x)^2
\|\boldsymbol h(\boldsymbol x)\|_2^2
\right]
\right\}^{1/2}
<\infty .
\end{equation*}
Together with $\E(\varepsilon_i\mid\boldsymbol x_i)=0$,
\eqref{eq:proof-thm1-normal-equation} implies
\begin{equation*}
\E\{\boldsymbol\xi_i(\lambda)\}=\boldsymbol0 .
\end{equation*}

Using $\E(\varepsilon_i\mid\boldsymbol x_i)=0$ and
$\E(\varepsilon_i^2\mid\boldsymbol x_i)=\sigma^2$, we have
\begin{align*}
\operatorname{Var}\{\boldsymbol\xi_i(\lambda)\}
&=
\E
\left[
\delta(\boldsymbol x_i)^{2\lambda}
\boldsymbol h(\boldsymbol x_i)\boldsymbol h(\boldsymbol x_i)^\top
\{b_\lambda(\boldsymbol x_i)+\varepsilon_i\}^2
\right] \\
&=
\E_{\boldsymbol x\sim\rho_{\mathcal X}^{S}}
\left[
\delta(\boldsymbol x)^{2\lambda}
b_\lambda(\boldsymbol x)^2
\boldsymbol h(\boldsymbol x)\boldsymbol h(\boldsymbol x)^\top
\right]
+
\sigma^2
\E_{\boldsymbol x\sim\rho_{\mathcal X}^{S}}
\left[
\delta(\boldsymbol x)^{2\lambda}
\boldsymbol h(\boldsymbol x)\boldsymbol h(\boldsymbol x)^\top
\right] \\
&=
\boldsymbol\Sigma_b(\lambda)
+
\boldsymbol\Sigma_\varepsilon(\lambda).
\end{align*}
The entries of this variance matrix are finite by the moment conditions in
Theorem~1. Hence the multivariate central limit theorem gives
\begin{equation}
\label{eq:proof-thm1-score-clt}
\boldsymbol U_{n_s}(\lambda)
=
\frac{1}{\sqrt{n_s}}
\sum_{i=1}^{n_s}
\boldsymbol\xi_i(\lambda)
\xrightarrow{d}
N\{
\boldsymbol0,
\boldsymbol\Sigma_b(\lambda)+\boldsymbol\Sigma_\varepsilon(\lambda)
\}.
\end{equation}

Combining \eqref{eq:proof-thm1-basic-expansion},
\eqref{eq:proof-thm1-inverse-convergence}, and
\eqref{eq:proof-thm1-score-clt}, Slutsky's theorem yields
\begin{equation*}
\sqrt{n_s}
\left\{
\widehat{\boldsymbol\theta}^{\mathrm{AIWLS}}(\lambda)
-
\boldsymbol\theta^*(\lambda)
\right\}
\xrightarrow{d}
N\{\boldsymbol0,\boldsymbol\Omega(\lambda)\},
\end{equation*}
where
\begin{equation*}
\boldsymbol\Omega(\lambda)
=
\boldsymbol\Sigma(\lambda)^{-1}
\left\{
\boldsymbol\Sigma_b(\lambda)
+
\boldsymbol\Sigma_\varepsilon(\lambda)
\right\}
\boldsymbol\Sigma(\lambda)^{-1}.
\end{equation*}
This proves Theorem~1.

\qed

\subsection{Auxiliary lemmas}

\begin{lemma}
\label{lem:proof-thm2-moment-implications}
Suppose that Assumptions~1, 3, and 4 hold. Then
\begin{equation}
\label{eq:proof-thm2-bh-moment}
\E_{\boldsymbol x\sim\rho_{\mathcal X}^{S}}
\left[
b(\boldsymbol x)^2
\|\boldsymbol h(\boldsymbol x)\|_2^2
\right]
<\infty,
\end{equation}
\begin{equation}
\label{eq:proof-thm2-b-moment}
\E_{\boldsymbol x\sim\rho_{\mathcal X}^{S}}
\{b(\boldsymbol x)^2\}
<\infty,
\end{equation}
\begin{equation}
\label{eq:proof-thm2-fh-moment}
\E_{\boldsymbol x\sim\rho_{\mathcal X}^{S}}
\left[
f(\boldsymbol x)^2
\|\boldsymbol h(\boldsymbol x)\|_2^2
\right]
<\infty,
\end{equation}
\begin{equation}
\label{eq:proof-thm2-target-h-moment}
\E_{\boldsymbol x\sim\rho_{\mathcal X}^{T}}
\left[
\|\boldsymbol h(\boldsymbol x)\|_2^4
\right]
<\infty,
\end{equation}
\begin{equation}
\label{eq:proof-thm2-target-fh-moment}
\E_{\boldsymbol x\sim\rho_{\mathcal X}^{T}}
\left[
f(\boldsymbol x)^2
\|\boldsymbol h(\boldsymbol x)\|_2^2
\right]
<\infty,
\end{equation}
and
\begin{equation}
\label{eq:proof-thm2-theta-star-unif-bound}
\sup_{0\le\lambda\le1}
\|\boldsymbol\theta^*(\lambda)\|_2
<\infty .
\end{equation}
\end{lemma}

\begin{proof}
By the importance-weighting identity, the objective in (10) with
$\lambda=1$ coincides with the target prediction risk objective defining
$\boldsymbol\theta^*$. Hence
$\boldsymbol\theta^*(1)=\boldsymbol\theta^*$ and
$b_1(\boldsymbol x)=b(\boldsymbol x)$. Taking $\lambda=1$ in Assumption~3
gives \eqref{eq:proof-thm2-bh-moment}. Since the basis includes an intercept,
$\|\boldsymbol h(\boldsymbol x)\|_2^2\ge1$, and
\eqref{eq:proof-thm2-b-moment} follows from
\eqref{eq:proof-thm2-bh-moment}.

Using
$f(\boldsymbol x)=\boldsymbol h(\boldsymbol x)^\top\boldsymbol\theta^*
+b(\boldsymbol x)$, we have
\begin{align*}
f(\boldsymbol x)^2\|\boldsymbol h(\boldsymbol x)\|_2^2
&\le
2\|\boldsymbol\theta^*\|_2^2
\|\boldsymbol h(\boldsymbol x)\|_2^4
+
2b(\boldsymbol x)^2
\|\boldsymbol h(\boldsymbol x)\|_2^2 .
\end{align*}
The first term on the right-hand side is integrable by Assumption~3, and the
second is integrable by \eqref{eq:proof-thm2-bh-moment}. This proves
\eqref{eq:proof-thm2-fh-moment}.

By the importance-weighting identity and Assumption~4,
\begin{align*}
\E_{\boldsymbol x\sim\rho_{\mathcal X}^{T}}
\left[
\|\boldsymbol h(\boldsymbol x)\|_2^4
\right]
&=
\E_{\boldsymbol x\sim\rho_{\mathcal X}^{S}}
\left[
\delta(\boldsymbol x)
\|\boldsymbol h(\boldsymbol x)\|_2^4
\right] \\
&\le
\overline\delta
\E_{\boldsymbol x\sim\rho_{\mathcal X}^{S}}
\left[
\|\boldsymbol h(\boldsymbol x)\|_2^4
\right]
<\infty .
\end{align*}
This proves \eqref{eq:proof-thm2-target-h-moment}. The same argument, together
with \eqref{eq:proof-thm2-fh-moment}, gives
\eqref{eq:proof-thm2-target-fh-moment}.

It remains to prove \eqref{eq:proof-thm2-theta-star-unif-bound}. By the
first-order condition for the population objective in (10),
\begin{equation*}
\boldsymbol\Sigma(\lambda)\boldsymbol\theta^*(\lambda)
=
\E_{\boldsymbol x\sim\rho_{\mathcal X}^{S}}
\left[
\delta(\boldsymbol x)^\lambda
\boldsymbol h(\boldsymbol x)f(\boldsymbol x)
\right].
\end{equation*}
Using $f(\boldsymbol x)=\boldsymbol h(\boldsymbol x)^\top
\boldsymbol\theta^*+b(\boldsymbol x)$, we obtain
\begin{equation*}
\boldsymbol\theta^*(\lambda)-\boldsymbol\theta^*
=
\boldsymbol\Sigma(\lambda)^{-1}
\E_{\boldsymbol x\sim\rho_{\mathcal X}^{S}}
\left[
\delta(\boldsymbol x)^\lambda
\boldsymbol h(\boldsymbol x)b(\boldsymbol x)
\right].
\end{equation*}
Assumption~1 gives
$\sup_{0\le\lambda\le1}\|\boldsymbol\Sigma(\lambda)^{-1}\|\le c_1^{-1}$.
Moreover, by Assumption~4, the Cauchy–Schwarz inequality,
\eqref{eq:proof-thm2-b-moment}, and Assumption~3,
\begin{align*}
&
\sup_{0\le\lambda\le1}
\left\|
\E_{\boldsymbol x\sim\rho_{\mathcal X}^{S}}
\left[
\delta(\boldsymbol x)^\lambda
\boldsymbol h(\boldsymbol x)b(\boldsymbol x)
\right]
\right\|_2
\\
&\quad
\le
\max\{1,\overline\delta\}
\left\{
\E_{\boldsymbol x\sim\rho_{\mathcal X}^{S}}
\|\boldsymbol h(\boldsymbol x)\|_2^2
\right\}^{1/2}
\left\{
\E_{\boldsymbol x\sim\rho_{\mathcal X}^{S}}
b(\boldsymbol x)^2
\right\}^{1/2}
<\infty .
\end{align*}
Therefore,
$\sup_{0\le\lambda\le1}
\|\boldsymbol\theta^*(\lambda)-\boldsymbol\theta^*\|_2<\infty$.
Since $\boldsymbol\theta^*$ is fixed and finite-dimensional,
\eqref{eq:proof-thm2-theta-star-unif-bound} follows.
\end{proof}

\begin{lemma}
\label{lem:proof-thm2-uniform-bounds}
Let
\begin{equation*}
\widehat{\boldsymbol S}_m
=
\frac{1}{n_s}
\boldsymbol H^\top
\widehat{\boldsymbol\Delta}_m
\boldsymbol H .
\end{equation*}
Suppose
that Assumptions~1 and 3--5 hold. Then
\begin{equation}
\label{eq:proof-thm2-unif-S-rate}
\max_{1\le m\le M}
\left\|
\widehat{\boldsymbol S}_m
-
\boldsymbol\Sigma(\lambda_m)
\right\|
=
O_p\left(\frac{M}{\sqrt{n_s}}+\kappa_n\right).
\end{equation}
Moreover,
\begin{equation}
\label{eq:proof-thm2-unif-Sinv-rate}
\max_{1\le m\le M}
\left\|
\widehat{\boldsymbol S}_m^{-1}
\right\|
=
O_p(1),
\end{equation}
\begin{equation}
\label{eq:proof-thm2-L-fro-rate}
\max_{1\le m\le M}
\|\boldsymbol L_m\|_F
=
O_p(n_s^{-1/2}),
\qquad
\sup_{\boldsymbol w\in\mathcal W}
\|\boldsymbol L(\boldsymbol w)\|_F
=
O_p(n_s^{-1/2}),
\end{equation}
\begin{equation}
\label{eq:proof-thm2-trace-rate}
\sup_{\boldsymbol w\in\mathcal W}
\left|
\operatorname{tr}
\left\{
\boldsymbol\Sigma^{T}
\boldsymbol L(\boldsymbol w)
\boldsymbol L_{M}^\top
\right\}
\right|
=
O_p(n_s^{-1}),
\end{equation}
\begin{equation}
\label{eq:proof-thm2-Lmb-generic}
\max_{1\le m\le M}
\|\boldsymbol L_m\boldsymbol b\|_2
=
O_p(1),
\end{equation}
and
\begin{equation}
\label{eq:proof-thm2-LMb-full}
\|\boldsymbol L_{M}\boldsymbol b\|_2
=
O_p(n_s^{-1/2}+\kappa_n).
\end{equation}
Finally,
\begin{equation}
\label{eq:proof-thm2-At-rate}
\|\widehat{\boldsymbol\Sigma}^{T}-\boldsymbol\Sigma^{T}\|
=
O_p(n_t^{-1/2}).
\end{equation}
\end{lemma}

\begin{proof}
For each $m$, define
\begin{equation*}
\boldsymbol S_m
=
\frac{1}{n_s}
\boldsymbol H^\top\boldsymbol\Delta_m\boldsymbol H,
\qquad
\boldsymbol\Delta_m
=
\operatorname{diag}
\{
\delta(\boldsymbol x_1)^{\lambda_m},\ldots,
\delta(\boldsymbol x_{n_s})^{\lambda_m}
\}.
\end{equation*}
Then
\begin{equation*}
\max_{1\le m\le M}
\|
\widehat{\boldsymbol S}_m-\boldsymbol\Sigma(\lambda_m)
\|
\le
I_{1n}+I_{2n},
\end{equation*}
where
\begin{equation*}
I_{1n}
=
\max_{1\le m\le M}
\|\boldsymbol S_m-\boldsymbol\Sigma(\lambda_m)\|,
\qquad
I_{2n}
=
\max_{1\le m\le M}
\|\widehat{\boldsymbol S}_m-\boldsymbol S_m\|.
\end{equation*}

We first bound $I_{1n}$. Fix $j,k\in\{1,\ldots,p\}$ and define
\begin{equation*}
Z_{i,m,jk}
=
\delta(\boldsymbol x_i)^{\lambda_m}
h_j(\boldsymbol x_i)h_k(\boldsymbol x_i)
-
\E_{\boldsymbol x\sim\rho_{\mathcal X}^{S}}
\left[
\delta(\boldsymbol x)^{\lambda_m}
h_j(\boldsymbol x)h_k(\boldsymbol x)
\right].
\end{equation*}
Since $\lambda_m\in[0,1]$ and Assumption~4 gives
$\delta(\boldsymbol x)^{\lambda_m}\le\max\{1,\overline\delta\}$,
Assumption~3 implies
$\sup_{1\le m\le M}\E(Z_{i,m,jk}^2)\le C$ for a finite constant $C$.
For every $K>0$, Chebyshev's inequality and the union bound give
\begin{align*}
&
\Pr\left(
\max_{1\le m\le M}
\left|
\frac{1}{n_s}
\sum_{i=1}^{n_s}
Z_{i,m,jk}
\right|
>
K\frac{M}{\sqrt{n_s}}
\right)
\\
&\quad
\le
\sum_{m=1}^{M}
\frac{
\operatorname{Var}(n_s^{-1}\sum_{i=1}^{n_s}Z_{i,m,jk})
}{
K^2M^2/n_s
}
\le
\frac{C}{K^2M}.
\end{align*}
Since $p$ is fixed, this entrywise bound implies
\begin{equation*}
I_{1n}
=
O_p\left(\frac{M}{\sqrt{n_s}}\right).
\end{equation*}

By the definition of $\kappa_n$ over the candidate exponent set,
\begin{align*}
I_{2n}
&\le
\kappa_n
\frac{1}{n_s}
\sum_{i=1}^{n_s}
\|\boldsymbol h(\boldsymbol x_i)
\boldsymbol h(\boldsymbol x_i)^\top\|
\\
&=
\kappa_n
\frac{1}{n_s}
\sum_{i=1}^{n_s}
\|\boldsymbol h(\boldsymbol x_i)\|_2^2
=
O_p(\kappa_n),
\end{align*}
where the last equality follows from Assumption~3. Combining the bounds for
$I_{1n}$ and $I_{2n}$ proves \eqref{eq:proof-thm2-unif-S-rate}.

We next prove the uniform boundedness of the inverse matrices. Let
$\boldsymbol w^{(M)}$ denote the unit vector selecting the candidate
$\lambda_M=1$. Since $\lambda_M=1$ belongs to the candidate set,
\begin{equation*}
\xi_n
\le
\mathrm{TEPE}^*(\boldsymbol w^{(M)})
=
\E_{\boldsymbol x\sim\rho_{\mathcal X}^{T}}
\{b(\boldsymbol x)^2\}.
\end{equation*}
By the importance-weighting identity, Assumption~4, and
\eqref{eq:proof-thm2-b-moment},
\begin{equation*}
\E_{\boldsymbol x\sim\rho_{\mathcal X}^{T}}
\{b(\boldsymbol x)^2\}
=
\E_{\boldsymbol x\sim\rho_{\mathcal X}^{S}}
\{\delta(\boldsymbol x)b(\boldsymbol x)^2\}
\le
\overline\delta
\E_{\boldsymbol x\sim\rho_{\mathcal X}^{S}}
\{b(\boldsymbol x)^2\}
<\infty .
\end{equation*}
Thus $\xi_n=O(1)$. Assumption~5 gives $M/\sqrt n=o(1)$ and
$\kappa_n=o_p(1)$. Since $n=\min(n_s,n_t)$, it follows that
$M/\sqrt{n_s}+\kappa_n=o_p(1)$. Therefore, by
\eqref{eq:proof-thm2-unif-S-rate}, Assumption~1, and Weyl's inequality, with
probability tending to one,
\begin{equation*}
\min_{1\le m\le M}
\lambda_{\min}(\widehat{\boldsymbol S}_m)
\ge
\frac{c_1}{2}.
\end{equation*}
This proves \eqref{eq:proof-thm2-unif-Sinv-rate} and also shows that the
sample Gram matrices are nonsingular with probability tending to one.

Since
\begin{equation*}
\boldsymbol L_m
=
\widehat{\boldsymbol S}_m^{-1}
\left(
\frac{1}{n_s}
\boldsymbol H^\top
\widehat{\boldsymbol\Delta}_m
\right),
\end{equation*}
we next control the second factor. Because $\lambda_M=1$ belongs to the
candidate set,
\begin{equation*}
\max_{1\le i\le n_s}
|\widehat\delta(\boldsymbol x_i)-\delta(\boldsymbol x_i)|
\le
\kappa_n .
\end{equation*}
Together with Assumption~4 and $\kappa_n=o_p(1)$, this gives
$\max_{1\le i\le n_s}\widehat\delta(\boldsymbol x_i)=O_p(1)$. Since
$\widehat\delta(\boldsymbol x_i)\ge0$,
\begin{equation*}
\max_{1\le i\le n_s}
\max_{1\le m\le M}
\widehat\delta(\boldsymbol x_i)^{\lambda_m}
=
O_p(1).
\end{equation*}
Consequently,
\begin{align*}
\max_{1\le m\le M}
\left\|
\frac{1}{n_s}
\boldsymbol H^\top
\widehat{\boldsymbol\Delta}_m
\right\|_F^2
&=
\max_{1\le m\le M}
\frac{1}{n_s^2}
\sum_{i=1}^{n_s}
\widehat\delta(\boldsymbol x_i)^{2\lambda_m}
\|\boldsymbol h(\boldsymbol x_i)\|_2^2
\\
&=
O_p(n_s^{-1}).
\end{align*}
Combining this bound with \eqref{eq:proof-thm2-unif-Sinv-rate} proves the
first part of \eqref{eq:proof-thm2-L-fro-rate}. The second part follows from
the convexity of the Frobenius norm and the fact that $\mathcal W$ is the
simplex.

By \eqref{eq:proof-thm2-target-h-moment}, $\|\boldsymbol\Sigma^T\|<\infty$.
Hence the trace bound \eqref{eq:proof-thm2-trace-rate} follows from
\begin{align*}
\left|
\operatorname{tr}
\left\{
\boldsymbol\Sigma^{T}
\boldsymbol L(\boldsymbol w)
\boldsymbol L_{M}^\top
\right\}
\right|
&\le
\|\boldsymbol\Sigma^{T}\|
\|\boldsymbol L(\boldsymbol w)\|_F
\|\boldsymbol L_{M}\|_F .
\end{align*}

For \eqref{eq:proof-thm2-Lmb-generic}, observe that
\begin{equation*}
\boldsymbol L_m\boldsymbol b
=
\widehat{\boldsymbol S}_m^{-1}
\frac{1}{n_s}
\sum_{i=1}^{n_s}
\widehat\delta(\boldsymbol x_i)^{\lambda_m}
\boldsymbol h(\boldsymbol x_i)b(\boldsymbol x_i).
\end{equation*}
The boundedness of the estimated powers, the Cauchy–Schwarz inequality,
\eqref{eq:proof-thm2-b-moment}, and Assumption~3 imply
\begin{equation*}
\max_{1\le m\le M}
\left\|
\frac{1}{n_s}
\sum_{i=1}^{n_s}
\widehat\delta(\boldsymbol x_i)^{\lambda_m}
\boldsymbol h(\boldsymbol x_i)b(\boldsymbol x_i)
\right\|_2
=
O_p(1).
\end{equation*}
Together with \eqref{eq:proof-thm2-unif-Sinv-rate}, this proves
\eqref{eq:proof-thm2-Lmb-generic}.

For the standard importance-weighted endpoint, $\lambda_M=1$. Hence
\begin{align*}
\boldsymbol L_{M}\boldsymbol b
={}&
\widehat{\boldsymbol S}_{M}^{-1}
\frac{1}{n_s}
\sum_{i=1}^{n_s}
\delta(\boldsymbol x_i)
\boldsymbol h(\boldsymbol x_i)b(\boldsymbol x_i)
\\
&+
\widehat{\boldsymbol S}_{M}^{-1}
\frac{1}{n_s}
\sum_{i=1}^{n_s}
\{\widehat\delta(\boldsymbol x_i)-\delta(\boldsymbol x_i)\}
\boldsymbol h(\boldsymbol x_i)b(\boldsymbol x_i).
\end{align*}
The average in the first term has mean zero because
\begin{equation*}
\E_{\boldsymbol x\sim\rho_{\mathcal X}^{S}}
\{\delta(\boldsymbol x)\boldsymbol h(\boldsymbol x)b(\boldsymbol x)\}
=
\E_{\boldsymbol x\sim\rho_{\mathcal X}^{T}}
\{\boldsymbol h(\boldsymbol x)b(\boldsymbol x)\}
=
\boldsymbol0,
\end{equation*}
where the last equality is the projection orthogonality under the target
covariate distribution. Its second moment is finite by
\eqref{eq:proof-thm2-bh-moment} and Assumption~4. Therefore, Chebyshev's
inequality gives
\begin{equation*}
\left\|
\frac{1}{n_s}
\sum_{i=1}^{n_s}
\delta(\boldsymbol x_i)
\boldsymbol h(\boldsymbol x_i)b(\boldsymbol x_i)
\right\|_2
=
O_p(n_s^{-1/2}).
\end{equation*}
The second average is bounded by
\begin{equation*}
\kappa_n
\frac{1}{n_s}
\sum_{i=1}^{n_s}
\|\boldsymbol h(\boldsymbol x_i)\|_2|b(\boldsymbol x_i)|
=
O_p(\kappa_n).
\end{equation*}
Combining these bounds with \eqref{eq:proof-thm2-unif-Sinv-rate} proves
\eqref{eq:proof-thm2-LMb-full}.

Finally, \eqref{eq:proof-thm2-At-rate} follows by applying Chebyshev's
inequality entrywise to
\begin{equation*}
\widehat{\boldsymbol\Sigma}^{T}
=
\frac{1}{n_t}
\sum_{j=1}^{n_t}
\boldsymbol h(\boldsymbol x_j')
\boldsymbol h(\boldsymbol x_j')^\top,
\end{equation*}
using \eqref{eq:proof-thm2-target-h-moment} and the fact that $p$ is fixed.
\end{proof}

\subsection{Proof of Theorem~2}

Throughout this proof, conditional expectations are taken given
$\mathcal F$ unless explicitly stated otherwise. Conditional on
$\mathcal F$, the matrices $\boldsymbol L_m$,
$\boldsymbol L(\boldsymbol w)$, $\boldsymbol L_M$, $\boldsymbol G$,
$\widehat{\boldsymbol\Sigma}^{T}$, and
$\widehat{\boldsymbol\Sigma}^{T}-\boldsymbol\Sigma^{T}$ are nonrandom. Let
$\boldsymbol\varepsilon=(\varepsilon_1,\ldots,\varepsilon_{n_s})^\top$. Since
$\boldsymbol y=\boldsymbol H\boldsymbol\theta^*+\boldsymbol b+
\boldsymbol\varepsilon$ and $\boldsymbol L_m\boldsymbol H=\boldsymbol I_p$,
we have
\begin{equation}
\label{eq:proof-thm2-decomposition}
\widehat{\boldsymbol\theta}(\boldsymbol w)-\boldsymbol\theta^*
=
\boldsymbol L(\boldsymbol w)\boldsymbol b
+
\boldsymbol L(\boldsymbol w)\boldsymbol\varepsilon,
\qquad
\widehat{\boldsymbol\theta}_{M}-\boldsymbol\theta^*
=
\boldsymbol L_{M}\boldsymbol b
+
\boldsymbol L_{M}\boldsymbol\varepsilon .
\end{equation}

By the projection orthogonality under the target covariate distribution,
$\E_{\boldsymbol z\sim\rho_{\mathcal X}^{T}}
\{\boldsymbol h(\boldsymbol z)b(\boldsymbol z)\}=\boldsymbol0$. Using
$f(\boldsymbol z)=\boldsymbol h(\boldsymbol z)^\top\boldsymbol\theta^*
+b(\boldsymbol z)$, the definition of $\mathrm{TEPE}(\boldsymbol w)$ gives
\begin{equation}
\label{eq:proof-thm2-tepe-expansion}
\mathrm{TEPE}(\boldsymbol w)
=
\E_{\boldsymbol z\sim\rho_{\mathcal X}^{T}}\{b(\boldsymbol z)^2\}
+
\E\left[
\left\{
\widehat{\boldsymbol\theta}(\boldsymbol w)-\boldsymbol\theta^*
\right\}^\top
\boldsymbol\Sigma^{T}
\left\{
\widehat{\boldsymbol\theta}(\boldsymbol w)-\boldsymbol\theta^*
\right\}
\mid\mathcal F
\right].
\end{equation}

Since $\phi_n=2$, the definition of $e_n(\boldsymbol w)$ gives
\begin{align*}
\E\{C(\boldsymbol w)\mid\mathcal F\}
={}&
\E\left[
\left\{
\widehat{\boldsymbol\theta}(\boldsymbol w)
-
\widehat{\boldsymbol\theta}_{M}
\right\}^\top
\boldsymbol\Sigma^{T}
\left\{
\widehat{\boldsymbol\theta}(\boldsymbol w)
-
\widehat{\boldsymbol\theta}_{M}
\right\}
\mid\mathcal F
\right]
\\
&+
2\E(\widehat\sigma^2\mid\mathcal F)
\operatorname{tr}
\left\{
\boldsymbol\Sigma^{T}
\boldsymbol L(\boldsymbol w)
\boldsymbol L_{M}^\top
\right\}
+
e_n(\boldsymbol w).
\end{align*}
Moreover,
\begin{equation*}
\widehat{\boldsymbol\theta}(\boldsymbol w)-\widehat{\boldsymbol\theta}_{M}
=
\{\widehat{\boldsymbol\theta}(\boldsymbol w)-\boldsymbol\theta^*\}
-
\{\widehat{\boldsymbol\theta}_{M}-\boldsymbol\theta^*\}.
\end{equation*}
Therefore, using \eqref{eq:proof-thm2-tepe-expansion} and the definition of
$r_n$,
\begin{align}
\label{eq:proof-thm2-first-term-expansion}
&
\E\left[
\left\{
\widehat{\boldsymbol\theta}(\boldsymbol w)
-
\widehat{\boldsymbol\theta}_{M}
\right\}^\top
\boldsymbol\Sigma^{T}
\left\{
\widehat{\boldsymbol\theta}(\boldsymbol w)
-
\widehat{\boldsymbol\theta}_{M}
\right\}
\mid\mathcal F
\right]
\nonumber\\
&\quad
=
\mathrm{TEPE}(\boldsymbol w)
+
r_n
-
2
\E\left[
\left\{
\widehat{\boldsymbol\theta}(\boldsymbol w)
-
\boldsymbol\theta^*
\right\}^\top
\boldsymbol\Sigma^{T}
\left(
\widehat{\boldsymbol\theta}_{M}
-
\boldsymbol\theta^*
\right)
\mid\mathcal F
\right].
\end{align}

By \eqref{eq:proof-thm2-decomposition}, the last conditional expectation
equals
\begin{equation}
\label{eq:proof-thm2-cross-term}
\boldsymbol b^\top
\boldsymbol L(\boldsymbol w)^\top
\boldsymbol\Sigma^{T}
\boldsymbol L_{M}\boldsymbol b
+
\sigma^2
\operatorname{tr}
\left\{
\boldsymbol\Sigma^{T}
\boldsymbol L(\boldsymbol w)
\boldsymbol L_{M}^\top
\right\},
\end{equation}
because $\E(\boldsymbol\varepsilon\mid\mathcal F)=\boldsymbol0$ and
$\E(\boldsymbol\varepsilon\boldsymbol\varepsilon^\top\mid\mathcal F)
=\sigma^2\boldsymbol I_{n_s}$.

Since $\lambda_1=0$, $\widehat{\boldsymbol\theta}_1$ is the OLS estimator.
Hence
$\widehat\sigma^2=\|\boldsymbol G\boldsymbol y\|_2^2/
\operatorname{tr}(\boldsymbol G)$, where
$\operatorname{tr}(\boldsymbol G)=n_s-p$. Moreover,
$\boldsymbol G\boldsymbol y=\boldsymbol G\boldsymbol b+
\boldsymbol G\boldsymbol\varepsilon$. The symmetry and idempotency of
$\boldsymbol G$ imply
\begin{equation}
\label{eq:proof-thm2-sigmahat-exp}
\E(\widehat\sigma^2\mid\mathcal F)
=
\sigma^2
+
\frac{\|\boldsymbol G\boldsymbol b\|_2^2}
{\operatorname{tr}(\boldsymbol G)}.
\end{equation}
Substituting \eqref{eq:proof-thm2-cross-term} and
\eqref{eq:proof-thm2-sigmahat-exp} into
\eqref{eq:proof-thm2-first-term-expansion} gives
\begin{equation}
\label{eq:proof-thm2-exact-expansion}
\E\{C(\boldsymbol w)\mid\mathcal F\}
=
\mathrm{TEPE}(\boldsymbol w)
+
r_n
+
d_n(\boldsymbol w)
+
e_n(\boldsymbol w).
\end{equation}
Subtracting $r_n$ proves (23).

It remains to prove the uniform asymptotic equivalence in
(24). From the
definition of $d_n(\boldsymbol w)$ and
Lemma~\ref{lem:proof-thm2-uniform-bounds},
\begin{equation}
\label{eq:proof-thm2-dn-rate}
\sup_{\boldsymbol w\in\mathcal W}
|d_n(\boldsymbol w)|
=
O_p(n_s^{-1/2}+\kappa_n+n_s^{-1}).
\end{equation}
Indeed,
\begin{equation*}
\sup_{\boldsymbol w\in\mathcal W}
\|\boldsymbol L(\boldsymbol w)\boldsymbol b\|_2
\le
\max_{1\le m\le M}
\|\boldsymbol L_m\boldsymbol b\|_2
=
O_p(1),
\end{equation*}
$\|\boldsymbol L_M\boldsymbol b\|_2=O_p(n_s^{-1/2}+\kappa_n)$, and
\begin{equation*}
\frac{\|\boldsymbol G\boldsymbol b\|_2^2}{\operatorname{tr}(\boldsymbol G)}
\le
\frac{\|\boldsymbol b\|_2^2}{n_s-p}
=
O_p(1),
\end{equation*}
where the last bound follows from \eqref{eq:proof-thm2-b-moment}.

We next bound $e_n(\boldsymbol w)$. First,
\begin{equation}
\label{eq:proof-thm2-theta-diff-second}
\sup_{\boldsymbol w\in\mathcal W}
\E\left[
\left\|
\widehat{\boldsymbol\theta}(\boldsymbol w)
-
\widehat{\boldsymbol\theta}_{M}
\right\|_2^2
\mid\mathcal F
\right]
=
O_p(1).
\end{equation}
To see this, write
$\bar{\boldsymbol\theta}_m=\boldsymbol L_m\boldsymbol f$, where
$\boldsymbol f=\{f(\boldsymbol x_1),\ldots,f(\boldsymbol x_{n_s})\}^\top$,
and let
$\bar{\boldsymbol\theta}(\boldsymbol w)=
\sum_{m=1}^{M}w_m\bar{\boldsymbol\theta}_m$. Then
\begin{equation*}
\widehat{\boldsymbol\theta}(\boldsymbol w)
=
\bar{\boldsymbol\theta}(\boldsymbol w)
+
\boldsymbol L(\boldsymbol w)\boldsymbol\varepsilon,
\qquad
\widehat{\boldsymbol\theta}_M
=
\bar{\boldsymbol\theta}_M
+
\boldsymbol L_M\boldsymbol\varepsilon .
\end{equation*}
By \eqref{eq:proof-thm2-fh-moment} and
Lemma~\ref{lem:proof-thm2-uniform-bounds},
$\max_{1\le m\le M}\|\bar{\boldsymbol\theta}_m\|_2=O_p(1)$. Also,
\begin{equation*}
\sup_{\boldsymbol w\in\mathcal W}
\E\left[
\left\|
\{\boldsymbol L(\boldsymbol w)-\boldsymbol L_M\}
\boldsymbol\varepsilon
\right\|_2^2
\mid\mathcal F
\right]
=
\sigma^2
\sup_{\boldsymbol w\in\mathcal W}
\|\boldsymbol L(\boldsymbol w)-\boldsymbol L_M\|_F^2
=
O_p(n_s^{-1}).
\end{equation*}
This proves \eqref{eq:proof-thm2-theta-diff-second}.

From \eqref{eq:proof-thm2-sigmahat-exp} and
\eqref{eq:proof-thm2-b-moment},
$\E(\widehat\sigma^2\mid\mathcal F)=O_p(1)$. Therefore, by the definition of
$e_n(\boldsymbol w)$,
\begin{align*}
\sup_{\boldsymbol w\in\mathcal W}|e_n(\boldsymbol w)|
&\le
\|\widehat{\boldsymbol\Sigma}^{T}-\boldsymbol\Sigma^{T}\|
\sup_{\boldsymbol w\in\mathcal W}
\E\left[
\left\|
\widehat{\boldsymbol\theta}(\boldsymbol w)
-
\widehat{\boldsymbol\theta}_{M}
\right\|_2^2
\mid\mathcal F
\right]
\\
&\quad
+
2\E(\widehat\sigma^2\mid\mathcal F)
\|\widehat{\boldsymbol\Sigma}^{T}-\boldsymbol\Sigma^{T}\|
\sup_{\boldsymbol w\in\mathcal W}
\|\boldsymbol L(\boldsymbol w)\|_F
\|\boldsymbol L_M\|_F .
\end{align*}
Using \eqref{eq:proof-thm2-At-rate},
\eqref{eq:proof-thm2-theta-diff-second}, and
\eqref{eq:proof-thm2-L-fro-rate}, we obtain
\begin{equation}
\label{eq:proof-thm2-en-rate}
\sup_{\boldsymbol w\in\mathcal W}
|e_n(\boldsymbol w)|
=
O_p(n_t^{-1/2}).
\end{equation}

Combining \eqref{eq:proof-thm2-dn-rate} and
\eqref{eq:proof-thm2-en-rate}, and using Assumption~5, which implies
$1/(\sqrt{n_s}\xi_n)=o(1)$, $1/(\sqrt{n_t}\xi_n)=o(1)$,
$1/(n_s\xi_n)=o(1)$, and $\kappa_n/\xi_n=o_p(1)$, gives
\begin{equation}
\label{eq:proof-thm2-rem-relative}
\sup_{\boldsymbol w\in\mathcal W}
|d_n(\boldsymbol w)+e_n(\boldsymbol w)|
=
o_p(\xi_n).
\end{equation}

We finally relate $\mathrm{TEPE}(\boldsymbol w)$ to
$\mathrm{TEPE}^*(\boldsymbol w)$. By the first-order condition for the
population objective in (10),
\begin{equation*}
\boldsymbol\Sigma(\lambda)\boldsymbol\theta^*(\lambda)
=
\E_{\boldsymbol x\sim\rho_{\mathcal X}^{S}}
\left[
\delta(\boldsymbol x)^\lambda
\boldsymbol h(\boldsymbol x)f(\boldsymbol x)
\right].
\end{equation*}
Using the same Chebyshev and union-bound argument as in
Lemma~\ref{lem:proof-thm2-uniform-bounds}, together with
\eqref{eq:proof-thm2-fh-moment} and the definition of $\kappa_n$, gives
\begin{equation}
\label{eq:proof-thm2-mu-rate}
\max_{1\le m\le M}
\left\|
\frac{1}{n_s}
\boldsymbol H^\top
\widehat{\boldsymbol\Delta}_m
\boldsymbol f
-
\E_{\boldsymbol x\sim\rho_{\mathcal X}^{S}}
\left[
\delta(\boldsymbol x)^{\lambda_m}
\boldsymbol h(\boldsymbol x)f(\boldsymbol x)
\right]
\right\|_2
=
O_p\left(\frac{M}{\sqrt{n_s}}+\kappa_n\right).
\end{equation}
Since
\begin{equation*}
\bar{\boldsymbol\theta}_m
=
\widehat{\boldsymbol S}_m^{-1}
\frac{1}{n_s}
\boldsymbol H^\top\widehat{\boldsymbol\Delta}_m\boldsymbol f,
\end{equation*}
\eqref{eq:proof-thm2-unif-S-rate},
\eqref{eq:proof-thm2-unif-Sinv-rate},
\eqref{eq:proof-thm2-mu-rate}, and
\eqref{eq:proof-thm2-theta-star-unif-bound} imply
\begin{equation}
\label{eq:proof-thm2-bartheta-rate}
\max_{1\le m\le M}
\|\bar{\boldsymbol\theta}_m-\boldsymbol\theta^*(\lambda_m)\|_2
=
O_p\left(\frac{M}{\sqrt{n_s}}+\kappa_n\right).
\end{equation}
Therefore, because $\mathcal W$ is the simplex,
\begin{equation}
\label{eq:proof-thm2-barthetaw-rate}
\sup_{\boldsymbol w\in\mathcal W}
\left\|
\sum_{m=1}^{M}w_m\bar{\boldsymbol\theta}_m
-
\boldsymbol\theta^*(\boldsymbol w)
\right\|_2
=
O_p\left(\frac{M}{\sqrt{n_s}}+\kappa_n\right).
\end{equation}

By the definition of $\mathrm{TEPE}(\boldsymbol w)$ and the decomposition
\begin{equation*}
\widehat{\boldsymbol\theta}(\boldsymbol w)
=
\sum_{m=1}^{M}w_m\bar{\boldsymbol\theta}_m
+
\boldsymbol L(\boldsymbol w)\boldsymbol\varepsilon,
\end{equation*}
we have
\begin{align}
\label{eq:proof-thm2-tepe-bar-decomp}
\mathrm{TEPE}(\boldsymbol w)
={}&
\E_{\boldsymbol z\sim\rho_{\mathcal X}^{T}}
\left[
\left\{
f(\boldsymbol z)
-
\boldsymbol h(\boldsymbol z)^\top
\sum_{m=1}^{M}w_m\bar{\boldsymbol\theta}_m
\right\}^{2}
\right]
\nonumber\\
&+
\sigma^2
\operatorname{tr}
\left\{
\boldsymbol\Sigma^{T}
\boldsymbol L(\boldsymbol w)
\boldsymbol L(\boldsymbol w)^\top
\right\}.
\end{align}
The trace term in \eqref{eq:proof-thm2-tepe-bar-decomp} is uniformly
$O_p(n_s^{-1})$. The target quadratic loss is locally Lipschitz on bounded
sets as a function of the coefficient vector, because its Hessian is
$2\boldsymbol\Sigma^T$. Moreover,
$\sum_{m=1}^{M}w_m\bar{\boldsymbol\theta}_m$ and
$\boldsymbol\theta^*(\boldsymbol w)$ are uniformly bounded with probability
tending to one. Hence \eqref{eq:proof-thm2-barthetaw-rate} gives
\begin{equation}
\label{eq:proof-thm2-tepe-approx}
\sup_{\boldsymbol w\in\mathcal W}
\left|
\mathrm{TEPE}(\boldsymbol w)
-
\mathrm{TEPE}^*(\boldsymbol w)
\right|
=
O_p\left(
\frac{M}{\sqrt{n_s}}
+
\kappa_n
+
\frac{1}{n_s}
\right)
=
o_p(\xi_n),
\end{equation}
where the last equality follows from Assumption~5. Since
$\xi_n=\inf_{\boldsymbol w\in\mathcal W}\mathrm{TEPE}^*(\boldsymbol w)$,
\eqref{eq:proof-thm2-tepe-approx} implies
\begin{equation}
\label{eq:proof-thm2-inf-tepe}
\inf_{\boldsymbol w\in\mathcal W}
\mathrm{TEPE}(\boldsymbol w)
=
\xi_n\{1+o_p(1)\}.
\end{equation}
Combining \eqref{eq:proof-thm2-rem-relative} and
\eqref{eq:proof-thm2-inf-tepe}, we have
\begin{equation}
\label{eq:proof-thm2-rem-over-tepe}
\sup_{\boldsymbol w\in\mathcal W}
\left|
\frac{
d_n(\boldsymbol w)+e_n(\boldsymbol w)
}{
\mathrm{TEPE}(\boldsymbol w)
}
\right|
=
o_p(1).
\end{equation}
Finally, \eqref{eq:proof-thm2-exact-expansion} and
\eqref{eq:proof-thm2-rem-over-tepe} imply
\begin{equation*}
\sup_{\boldsymbol w\in\mathcal W}
\left|
\frac{
\E\left[
C(\boldsymbol w)-r_n\mid\mathcal F
\right]
-
\mathrm{TEPE}(\boldsymbol w)
}{
\mathrm{TEPE}(\boldsymbol w)
}
\right|
=
o_p(1).
\end{equation*}
This completes the proof.

\qed

\subsection{Proof of Corollary 1}

If the working model is correctly specified, then there exists
$\boldsymbol\theta_0$ such that
$f(\boldsymbol x)=\boldsymbol h(\boldsymbol x)^\top\boldsymbol\theta_0$ for
$\rho_{\mathcal X}^{S}$- and $\rho_{\mathcal X}^{T}$-almost every
$\boldsymbol x$. Hence $\boldsymbol\theta^*=\boldsymbol\theta_0$ and
$\boldsymbol b=\boldsymbol0$ with probability one. Together with
$\widehat{\boldsymbol\Sigma}^{T}=\boldsymbol\Sigma^{T}$, this implies
$d_n(\boldsymbol w)=e_n(\boldsymbol w)=0$ for every
$\boldsymbol w\in\mathcal W_n$. The result follows directly from (23).

\qed

\subsection{Proof of Theorem~3}

Throughout the proof, all stochastic orders are under
Assumptions~1--4 and 6. Since Assumption~6 implies Assumption~5, the auxiliary
bounds in Lemmas~\ref{lem:proof-thm2-moment-implications} and
\ref{lem:proof-thm2-uniform-bounds} remain valid. In addition, the argument
leading to \eqref{eq:proof-thm2-tepe-approx} gives
\begin{equation}
\label{eq:proof-thm3-tepe-approx}
\sup_{\boldsymbol w\in\mathcal W}
\left|
\mathrm{TEPE}(\boldsymbol w)
-
\mathrm{TEPE}^*(\boldsymbol w)
\right|
=
O_p\left(
\frac{M}{\sqrt{n_s}}
+
\kappa_n
+
\frac{1}{n_s}
\right)
=
o_p(\xi_n),
\end{equation}
where the last equality follows from Assumption~6. Because
$\xi_n=\inf_{\boldsymbol w\in\mathcal W}\mathrm{TEPE}^*(\boldsymbol w)$,
\eqref{eq:proof-thm3-tepe-approx} implies
\begin{equation}
\label{eq:proof-thm3-inf-tepe}
\inf_{\boldsymbol w\in\mathcal W}
\mathrm{TEPE}(\boldsymbol w)
=
\xi_n\{1+o_p(1)\}.
\end{equation}

We next approximate the criterion uniformly over $\mathcal W$. Define
\begin{equation*}
K_n
=
\left(
\widehat{\boldsymbol\theta}_{M}
-
\boldsymbol\theta^*
\right)^\top
\widehat{\boldsymbol\Sigma}^{T}
\left(
\widehat{\boldsymbol\theta}_{M}
-
\boldsymbol\theta^*
\right)
-
\E_{\boldsymbol x\sim\rho_{\mathcal X}^{T}}
\{b(\boldsymbol x)^2\}.
\end{equation*}
Using
\begin{equation*}
\widehat{\boldsymbol\theta}(\boldsymbol w)-\widehat{\boldsymbol\theta}_{M}
=
\{\widehat{\boldsymbol\theta}(\boldsymbol w)-\boldsymbol\theta^*\}
-
\{\widehat{\boldsymbol\theta}_{M}-\boldsymbol\theta^*\},
\end{equation*}
the definition of $C(\boldsymbol w)$, and the expansion of
$\mathrm{TEPE}(\boldsymbol w)$ in \eqref{eq:proof-thm2-tepe-expansion}, we
obtain
\begin{equation}
\label{eq:proof-thm3-C-minus-TEPE-decomp}
C(\boldsymbol w)-K_n-\mathrm{TEPE}(\boldsymbol w)
=
T_{1n}(\boldsymbol w)
-
2T_{2n}(\boldsymbol w)
+
T_{3n}(\boldsymbol w),
\end{equation}
where
\begin{align*}
T_{1n}(\boldsymbol w)
={}&
\left\{
\widehat{\boldsymbol\theta}(\boldsymbol w)
-
\boldsymbol\theta^*
\right\}^\top
\widehat{\boldsymbol\Sigma}^{T}
\left\{
\widehat{\boldsymbol\theta}(\boldsymbol w)
-
\boldsymbol\theta^*
\right\}
\\
&-
\E\left[
\left\{
\widehat{\boldsymbol\theta}(\boldsymbol w)
-
\boldsymbol\theta^*
\right\}^\top
\boldsymbol\Sigma^{T}
\left\{
\widehat{\boldsymbol\theta}(\boldsymbol w)
-
\boldsymbol\theta^*
\right\}
\mid\mathcal F
\right],
\\
T_{2n}(\boldsymbol w)
={}&
\left\{
\widehat{\boldsymbol\theta}(\boldsymbol w)
-
\boldsymbol\theta^*
\right\}^\top
\widehat{\boldsymbol\Sigma}^{T}
\left(
\widehat{\boldsymbol\theta}_{M}
-
\boldsymbol\theta^*
\right),
\\
T_{3n}(\boldsymbol w)
={}&
\phi_n\widehat\sigma^2
\operatorname{tr}
\left\{
\widehat{\boldsymbol\Sigma}^{T}
\boldsymbol L(\boldsymbol w)
\boldsymbol L_{M}^{\top}
\right\}.
\end{align*}
It is enough to show that, for $j=1,2,3$,
\begin{equation}
\label{eq:proof-thm3-Tj-target}
\sup_{\boldsymbol w\in\mathcal W}
|T_{jn}(\boldsymbol w)|
=
o_p(\xi_n).
\end{equation}

We first consider $T_{1n}(\boldsymbol w)$. Let
$\bar{\boldsymbol\theta}_m=\boldsymbol L_m\boldsymbol f$, where
$\boldsymbol f=\{f(\boldsymbol x_1),\ldots,f(\boldsymbol x_{n_s})\}^\top$.
From the proof of Theorem~2, uniformly over $\boldsymbol w\in\mathcal W$,
\begin{equation*}
\left\|
\sum_{m=1}^{M}w_m\bar{\boldsymbol\theta}_m
-
\boldsymbol\theta^*(\boldsymbol w)
\right\|_2
=
O_p\left(\frac{M}{\sqrt{n_s}}+\kappa_n\right).
\end{equation*}
Together with Lemma~\ref{lem:proof-thm2-moment-implications}, this implies
\begin{equation}
\label{eq:proof-thm3-deterministic-bound}
\sup_{\boldsymbol w\in\mathcal W}
\left\|
\sum_{m=1}^{M}w_m\bar{\boldsymbol\theta}_m
-
\boldsymbol\theta^*
\right\|_2
=
O_p(1).
\end{equation}
Furthermore,
\begin{equation}
\label{eq:proof-thm3-Leps-bound}
\sup_{\boldsymbol w\in\mathcal W}
\|\boldsymbol L(\boldsymbol w)\boldsymbol\varepsilon\|_2
\le
\max_{1\le m\le M}
\|\boldsymbol L_m\boldsymbol\varepsilon\|_2
=
O_p\left(\frac{M}{\sqrt{n_s}}\right).
\end{equation}
Indeed, conditional on $\mathcal F$,
\begin{equation*}
\E\{\|\boldsymbol L_m\boldsymbol\varepsilon\|_2^2\mid\mathcal F\}
=
\sigma^2\|\boldsymbol L_m\|_F^2 .
\end{equation*}
Since Lemma~\ref{lem:proof-thm2-uniform-bounds} gives
$\max_m\|\boldsymbol L_m\|_F=O_p(n_s^{-1/2})$, Markov's inequality and the
union bound over $m=1,\ldots,M$ yield
\eqref{eq:proof-thm3-Leps-bound}.

Since
\begin{equation*}
\widehat{\boldsymbol\theta}(\boldsymbol w)-\boldsymbol\theta^*
=
\left\{
\sum_{m=1}^{M}w_m\bar{\boldsymbol\theta}_m
-
\boldsymbol\theta^*
\right\}
+
\boldsymbol L(\boldsymbol w)\boldsymbol\varepsilon,
\end{equation*}
we decompose $T_{1n}(\boldsymbol w)$ as
\begin{align*}
T_{1n}(\boldsymbol w)
={}&
\left\{
\sum_{m=1}^{M}w_m\bar{\boldsymbol\theta}_m-\boldsymbol\theta^*
\right\}^\top
(\widehat{\boldsymbol\Sigma}^{T}-\boldsymbol\Sigma^{T})
\left\{
\sum_{m=1}^{M}w_m\bar{\boldsymbol\theta}_m-\boldsymbol\theta^*
\right\}
\\
&+
2
\left\{
\sum_{m=1}^{M}w_m\bar{\boldsymbol\theta}_m-\boldsymbol\theta^*
\right\}^\top
\widehat{\boldsymbol\Sigma}^{T}
\boldsymbol L(\boldsymbol w)\boldsymbol\varepsilon
\\
&+
\boldsymbol\varepsilon^\top
\boldsymbol L(\boldsymbol w)^\top
\boldsymbol\Sigma^{T}
\boldsymbol L(\boldsymbol w)
\boldsymbol\varepsilon
-
\sigma^2
\operatorname{tr}
\left\{
\boldsymbol\Sigma^{T}
\boldsymbol L(\boldsymbol w)
\boldsymbol L(\boldsymbol w)^\top
\right\}
\\
&+
\boldsymbol\varepsilon^\top
\boldsymbol L(\boldsymbol w)^\top
(\widehat{\boldsymbol\Sigma}^{T}-\boldsymbol\Sigma^{T})
\boldsymbol L(\boldsymbol w)\boldsymbol\varepsilon .
\end{align*}
The first component is uniformly $O_p(n_t^{-1/2})=o_p(\xi_n)$ by
\eqref{eq:proof-thm3-deterministic-bound},
\eqref{eq:proof-thm2-At-rate}, and Assumption~6. The second component is
uniformly $O_p(M/\sqrt{n_s})=o_p(\xi_n)$ by
\eqref{eq:proof-thm3-deterministic-bound},
\eqref{eq:proof-thm3-Leps-bound}, $\|\widehat{\boldsymbol\Sigma}^{T}\|=O_p(1)$,
and Assumption~6. The fourth component is uniformly bounded by
\begin{equation*}
\|\widehat{\boldsymbol\Sigma}^{T}-\boldsymbol\Sigma^{T}\|
\sup_{\boldsymbol w\in\mathcal W}
\|\boldsymbol L(\boldsymbol w)\boldsymbol\varepsilon\|_2^2
=
O_p\left(n_t^{-1/2}\frac{M^2}{n_s}\right)
=
o_p(\xi_n).
\end{equation*}

It remains to bound the centered quadratic component. For $i,j=1,\ldots,M$,
define
\begin{equation*}
Q_{ij}
=
\boldsymbol\varepsilon^\top
\boldsymbol L_i^\top
\boldsymbol\Sigma^{T}
\boldsymbol L_j
\boldsymbol\varepsilon
-
\sigma^2
\operatorname{tr}
\left\{
\boldsymbol\Sigma^{T}
\boldsymbol L_j
\boldsymbol L_i^\top
\right\}.
\end{equation*}
The centered quadratic component equals
$\sum_{i=1}^{M}\sum_{j=1}^{M}w_iw_jQ_{ij}$ and is bounded in absolute value by
$\max_{1\le i,j\le M}|Q_{ij}|$. Conditional on $\mathcal F$, Assumption~2
implies the quadratic-form variance bound
\begin{equation*}
\operatorname{Var}(Q_{ij}\mid\mathcal F)
\le
C
\left\|
\boldsymbol L_i^\top
\boldsymbol\Sigma^{T}
\boldsymbol L_j
\right\|_F^2
=
O_p(n_s^{-2})
\end{equation*}
uniformly over $i,j$. Chebyshev's inequality and the union bound then give
\begin{equation*}
\max_{1\le i,j\le M}
|Q_{ij}|
=
O_p\left(\frac{M}{n_s}\right)
=
o_p(\xi_n).
\end{equation*}
Combining these bounds gives
\begin{equation}
\label{eq:proof-thm3-T1-rate}
\sup_{\boldsymbol w\in\mathcal W}
|T_{1n}(\boldsymbol w)|
=
o_p(\xi_n).
\end{equation}

We next consider $T_{2n}(\boldsymbol w)$. Since $\lambda_M=1$, the
importance-weighting identity gives $\boldsymbol\theta^*(1)=\boldsymbol\theta^*$.
By \eqref{eq:proof-thm2-bartheta-rate} and
\eqref{eq:proof-thm3-Leps-bound},
\begin{equation}
\label{eq:proof-thm3-M-endpoint-rate}
\|\widehat{\boldsymbol\theta}_{M}-\boldsymbol\theta^*\|_2
=
O_p\left(\frac{M}{\sqrt{n_s}}+\kappa_n\right)
=
o_p(\xi_n),
\end{equation}
where the last equality follows from Assumption~6. Also,
\begin{equation*}
\sup_{\boldsymbol w\in\mathcal W}
\|\widehat{\boldsymbol\theta}(\boldsymbol w)-\boldsymbol\theta^*\|_2
=
O_p(1),
\qquad
\|\widehat{\boldsymbol\Sigma}^{T}\|=O_p(1).
\end{equation*}
Therefore,
\begin{equation}
\label{eq:proof-thm3-T2-rate}
\sup_{\boldsymbol w\in\mathcal W}
|T_{2n}(\boldsymbol w)|
=
o_p(\xi_n).
\end{equation}

It remains to consider $T_{3n}(\boldsymbol w)$. Lemma~\ref{lem:proof-thm2-uniform-bounds}
and $\|\widehat{\boldsymbol\Sigma}^{T}\|=O_p(1)$ give
\begin{equation*}
\sup_{\boldsymbol w\in\mathcal W}
\left|
\operatorname{tr}
\left\{
\widehat{\boldsymbol\Sigma}^{T}
\boldsymbol L(\boldsymbol w)
\boldsymbol L_{M}^\top
\right\}
\right|
=
O_p(n_s^{-1}).
\end{equation*}
Moreover, $\widehat\sigma^2=O_p(1)$. Indeed,
$\widehat\sigma^2=\|\boldsymbol G\boldsymbol y\|_2^2/\operatorname{tr}(\boldsymbol G)$
and
$\boldsymbol G\boldsymbol y=\boldsymbol G\boldsymbol b+
\boldsymbol G\boldsymbol\varepsilon$. Since $\boldsymbol G$ is idempotent,
$\|\boldsymbol G\boldsymbol b\|_2^2\le\|\boldsymbol b\|_2^2$, and hence
\begin{equation*}
\frac{\|\boldsymbol G\boldsymbol b\|_2^2}{\operatorname{tr}(\boldsymbol G)}
=
O_p(1)
\end{equation*}
by \eqref{eq:proof-thm2-b-moment}. Also,
\begin{equation*}
\E\left[
\frac{\|\boldsymbol G\boldsymbol\varepsilon\|_2^2}{\operatorname{tr}(\boldsymbol G)}
\mid \mathcal F
\right]
=
\sigma^2,
\end{equation*}
so Markov's inequality gives
$\|\boldsymbol G\boldsymbol\varepsilon\|_2^2/\operatorname{tr}(\boldsymbol G)
=O_p(1)$. Thus
\begin{equation}
\label{eq:proof-thm3-T3-rate}
\sup_{\boldsymbol w\in\mathcal W}
|T_{3n}(\boldsymbol w)|
=
O_p\left(\frac{\phi_n}{n_s}\right)
=
o_p(\xi_n),
\end{equation}
where the last equality follows from Assumption~6. Combining
\eqref{eq:proof-thm3-C-minus-TEPE-decomp},
\eqref{eq:proof-thm3-T1-rate}, \eqref{eq:proof-thm3-T2-rate}, and
\eqref{eq:proof-thm3-T3-rate} gives
\begin{equation}
\label{eq:proof-thm3-criterion-approx}
\sup_{\boldsymbol w\in\mathcal W}
\left|
C(\boldsymbol w)
-
K_n
-
\mathrm{TEPE}(\boldsymbol w)
\right|
=
o_p(\xi_n).
\end{equation}

Let
\begin{equation*}
\boldsymbol w_n^o
\in
\argmin_{\boldsymbol w\in\mathcal W}
\mathrm{TEPE}(\boldsymbol w).
\end{equation*}
Such a minimizer exists because $\mathcal W$ is compact and
$\mathrm{TEPE}(\boldsymbol w)$ is continuous in $\boldsymbol w$. Since
$\widehat{\boldsymbol w}$ minimizes $C(\boldsymbol w)$ over $\mathcal W$, we
have
\begin{align*}
\mathrm{TEPE}(\widehat{\boldsymbol w})
&\le
C(\widehat{\boldsymbol w})-K_n
+
\sup_{\boldsymbol w\in\mathcal W}
|C(\boldsymbol w)-K_n-\mathrm{TEPE}(\boldsymbol w)|
\\
&\le
C(\boldsymbol w_n^o)-K_n
+
\sup_{\boldsymbol w\in\mathcal W}
|C(\boldsymbol w)-K_n-\mathrm{TEPE}(\boldsymbol w)|
\\
&\le
\mathrm{TEPE}(\boldsymbol w_n^o)
+
2
\sup_{\boldsymbol w\in\mathcal W}
|C(\boldsymbol w)-K_n-\mathrm{TEPE}(\boldsymbol w)|
\\
&=
\inf_{\boldsymbol w\in\mathcal W}\mathrm{TEPE}(\boldsymbol w)
+
o_p(\xi_n).
\end{align*}
By \eqref{eq:proof-thm3-inf-tepe},
\begin{equation*}
\inf_{\boldsymbol w\in\mathcal W}\mathrm{TEPE}(\boldsymbol w)
=
\xi_n\{1+o_p(1)\}.
\end{equation*}
Hence
\begin{equation*}
\mathrm{TEPE}(\widehat{\boldsymbol w})
\le
\inf_{\boldsymbol w\in\mathcal W}
\mathrm{TEPE}(\boldsymbol w)\{1+o_p(1)\}.
\end{equation*}
The reverse inequality follows from the definition of the infimum. Therefore,
\begin{equation*}
\frac{
\mathrm{TEPE}(\widehat{\boldsymbol w})
}{
\inf_{\boldsymbol w\in\mathcal W}
\mathrm{TEPE}(\boldsymbol w)
}
=
1+o_p(1).
\end{equation*}
This proves (26) and completes the proof.

\qed

\subsection{Proof of Proposition~1}

Under correct specification, there exists
$\boldsymbol\theta_0\in\mathbb R^p$ such that
$f(\boldsymbol x)=\boldsymbol h(\boldsymbol x)^\top\boldsymbol\theta_0$
for $\rho_{\mathcal X}^{S}$- and $\rho_{\mathcal X}^{T}$-almost every
$\boldsymbol x$. Hence $\boldsymbol\theta^*=\boldsymbol\theta_0$ and
$b(\boldsymbol x)=0$ for both $\rho_{\mathcal X}^{S}$- and
$\rho_{\mathcal X}^{T}$-almost every $\boldsymbol x$. Therefore, with
probability one,
\begin{equation*}
\boldsymbol y
=
\boldsymbol H\boldsymbol\theta^*
+
\boldsymbol\varepsilon,
\qquad
\boldsymbol\varepsilon=(\varepsilon_1,\ldots,\varepsilon_{n_s})^\top .
\end{equation*}

For each $m=1,\ldots,M$, the matrix $\boldsymbol L_m$ satisfies
$\boldsymbol L_m\boldsymbol H=\boldsymbol I_p$. Hence, for any
$\boldsymbol w\in\mathcal W$,
\begin{equation*}
\boldsymbol L(\boldsymbol w)\boldsymbol H
=
\sum_{m=1}^{M}w_m\boldsymbol L_m\boldsymbol H
=
\boldsymbol I_p .
\end{equation*}
It follows that
\begin{equation*}
\widehat{\boldsymbol\theta}(\boldsymbol w)
=
\boldsymbol L(\boldsymbol w)\boldsymbol y
=
\boldsymbol\theta^*
+
\boldsymbol L(\boldsymbol w)\boldsymbol\varepsilon .
\end{equation*}
Since
$f(\boldsymbol z)=\boldsymbol h(\boldsymbol z)^\top\boldsymbol\theta^*$
for $\rho_{\mathcal X}^{T}$-almost every $\boldsymbol z$, the definition of
$\mathrm{TEPE}$ gives
\begin{align}
\label{eq:proof-prop1-risk-w}
\mathrm{TEPE}(\boldsymbol w)
&=
\E\left[
\left\{
\boldsymbol h(\boldsymbol z)^\top
\boldsymbol L(\boldsymbol w)\boldsymbol\varepsilon
\right\}^{2}
\mid
\mathcal F
\right]
\nonumber\\
&=
\E\left[
\boldsymbol\varepsilon^\top
\boldsymbol L(\boldsymbol w)^\top
\boldsymbol\Sigma^{T}
\boldsymbol L(\boldsymbol w)
\boldsymbol\varepsilon
\mid
\mathcal F
\right]
\nonumber\\
&=
\sigma^2
\operatorname{tr}
\left\{
\boldsymbol\Sigma^{T}
\boldsymbol L(\boldsymbol w)
\boldsymbol L(\boldsymbol w)^\top
\right\},
\end{align}
where the second equality uses the independence of the target covariate
$\boldsymbol z$ from the source errors and $\mathcal F$, and the last equality
uses
$\E(\boldsymbol\varepsilon\mid\mathcal F)=\boldsymbol0$ and
$\E(\boldsymbol\varepsilon\boldsymbol\varepsilon^\top\mid\mathcal F)
=\sigma^2\boldsymbol I_{n_s}$.

Since $\lambda_1=0$,
\begin{equation*}
\boldsymbol L_1
=
(\boldsymbol H^\top\boldsymbol H)^{-1}\boldsymbol H^\top,
\qquad
\boldsymbol L_1^\top
=
\boldsymbol H(\boldsymbol H^\top\boldsymbol H)^{-1}.
\end{equation*}
Using $\boldsymbol L(\boldsymbol w)\boldsymbol H=\boldsymbol I_p$, we have
\begin{equation*}
\boldsymbol L(\boldsymbol w)\boldsymbol L_1^\top
=
(\boldsymbol H^\top\boldsymbol H)^{-1}
=
\boldsymbol L_1\boldsymbol L_1^\top .
\end{equation*}
Taking transposes also gives
\begin{equation*}
\boldsymbol L_1\boldsymbol L(\boldsymbol w)^\top
=
\boldsymbol L_1\boldsymbol L_1^\top .
\end{equation*}
Consequently,
\begin{equation}
\label{eq:proof-prop1-LwL1-identity}
\left\{
\boldsymbol L(\boldsymbol w)-\boldsymbol L_1
\right\}
\left\{
\boldsymbol L(\boldsymbol w)-\boldsymbol L_1
\right\}^{\top}
=
\boldsymbol L(\boldsymbol w)\boldsymbol L(\boldsymbol w)^\top
-
\boldsymbol L_1\boldsymbol L_1^\top .
\end{equation}
Combining \eqref{eq:proof-prop1-risk-w} and
\eqref{eq:proof-prop1-LwL1-identity}, we obtain
\begin{align*}
&
\mathrm{TEPE}(\boldsymbol w)
-
\mathrm{TEPE}\{\boldsymbol w^{(1)}\}
\\
&\quad =
\sigma^2
\operatorname{tr}
\left[
\boldsymbol\Sigma^{T}
\left\{
\boldsymbol L(\boldsymbol w)-\boldsymbol L_1
\right\}
\left\{
\boldsymbol L(\boldsymbol w)-\boldsymbol L_1
\right\}^{\top}
\right].
\end{align*}
Since $\boldsymbol\Sigma^{T}$ is positive semidefinite and
\[
\left\{
\boldsymbol L(\boldsymbol w)-\boldsymbol L_1
\right\}
\left\{
\boldsymbol L(\boldsymbol w)-\boldsymbol L_1
\right\}^{\top}
\]
is positive semidefinite, the trace on the right-hand side is nonnegative.
This proves (27).

\qed

\subsection{Proof of Theorem~4}

Write $\mathcal N_n=\mathcal N_n( a_{n})$ and
$\mathcal N_n^c=\{1,\ldots,M\}\setminus\mathcal N_n$. The argument is carried
out on the event where the relevant sample Gram matrices are nonsingular; the
bounds below imply that this event has probability tending to one. Under
correct specification, there exists $\boldsymbol\theta_0\in\mathbb R^p$ such
that
$f(\boldsymbol x)=\boldsymbol h(\boldsymbol x)^\top\boldsymbol\theta_0$
for $\rho_{\mathcal X}^{S}$- and
$\rho_{\mathcal X}^{T}$-almost every $\boldsymbol x$. By uniqueness of the
target projection parameter, $\boldsymbol\theta^*=\boldsymbol\theta_0$. Hence
$b(\boldsymbol x)=0$ for both $\rho_{\mathcal X}^{S}$- and
$\rho_{\mathcal X}^{T}$-almost every $\boldsymbol x$. Therefore, with
probability one,
\begin{equation*}
\boldsymbol y
=
\boldsymbol H\boldsymbol\theta^*
+
\boldsymbol\varepsilon,
\qquad
\boldsymbol\varepsilon=(\varepsilon_1,\ldots,\varepsilon_{n_s})^\top .
\end{equation*}
Since $\boldsymbol L_m\boldsymbol H=\boldsymbol I_p$, we have
\begin{equation}
\label{eq:proof-thm4-theta-noise}
\widehat{\boldsymbol\theta}_m
=
\boldsymbol\theta^*
+
\boldsymbol L_m\boldsymbol\varepsilon,
\qquad
\widehat{\boldsymbol\theta}(\boldsymbol w)
=
\boldsymbol\theta^*
+
\boldsymbol L(\boldsymbol w)\boldsymbol\varepsilon .
\end{equation}

Define
\begin{equation*}
Q_n(\boldsymbol w)
=
\left\{
\widehat{\boldsymbol\theta}(\boldsymbol w)
-
\widehat{\boldsymbol\theta}_M
\right\}^{\top}
\widehat{\boldsymbol\Sigma}^{T}
\left\{
\widehat{\boldsymbol\theta}(\boldsymbol w)
-
\widehat{\boldsymbol\theta}_M
\right\}
\end{equation*}
and
\begin{equation*}
\tau_{m,n}
=
\operatorname{tr}
\left\{
\widehat{\boldsymbol\Sigma}^{T}
\boldsymbol L_m
\boldsymbol L_M^{\top}
\right\},
\qquad m=1,\ldots,M .
\end{equation*}
Then the criterion can be written as
\begin{equation}
\label{eq:proof-thm4-C-decomp}
C(\boldsymbol w)
=
Q_n(\boldsymbol w)
+
\phi_n\widehat\sigma^2
\sum_{m=1}^{M}w_m\tau_{m,n}.
\end{equation}
Since $\widehat{\boldsymbol\Sigma}^{T}$ is nonnegative definite,
\begin{equation}
\label{eq:proof-thm4-Q-nonnegative}
Q_n(\boldsymbol w)\ge0
\end{equation}
for every $\boldsymbol w\in\mathcal W$.

We first establish a uniform expansion for the trace term. Let
\begin{equation*}
\widehat{\boldsymbol S}_m
=
\frac{1}{n_s}
\boldsymbol H^\top
\widehat{\boldsymbol\Delta}_m
\boldsymbol H,
\qquad
\widehat{\boldsymbol R}_m
=
\frac{1}{n_s}
\boldsymbol H^\top
\widehat{\boldsymbol\Delta}_m
\widehat{\boldsymbol\Delta}_M
\boldsymbol H .
\end{equation*}
Since
\begin{equation*}
\boldsymbol L_m
=
\widehat{\boldsymbol S}_m^{-1}
\left(
\frac{1}{n_s}
\boldsymbol H^\top\widehat{\boldsymbol\Delta}_m
\right),
\qquad
\boldsymbol L_M^\top
=
\left(
\frac{1}{n_s}
\widehat{\boldsymbol\Delta}_M\boldsymbol H
\right)
\widehat{\boldsymbol S}_M^{-1},
\end{equation*}
we have
\begin{equation}
\label{eq:proof-thm4-trace-representation}
n_s\tau_{m,n}
=
\operatorname{tr}
\left\{
\widehat{\boldsymbol\Sigma}^{T}
\widehat{\boldsymbol S}_m^{-1}
\widehat{\boldsymbol R}_m
\widehat{\boldsymbol S}_M^{-1}
\right\}.
\end{equation}

By the same entrywise Chebyshev and union-bound argument used in
Lemma~\ref{lem:proof-thm2-uniform-bounds}, together with Assumptions~4 and
7,
\begin{equation*}
\max_{1\le m\le M}
\left\|
\widehat{\boldsymbol S}_m
-
\boldsymbol\Sigma(\lambda_m)
\right\|
=
O_p\left(\frac{M}{\sqrt{n_s}}+\kappa_n\right).
\end{equation*}
Assumption~9 implies $M/\sqrt{n_s}=o(\eta_n)$ and
$\kappa_n=o_p(\eta_n)$. Hence
\begin{equation}
\label{eq:proof-thm4-S-opeta}
\max_{1\le m\le M}
\left\|
\widehat{\boldsymbol S}_m
-
\boldsymbol\Sigma(\lambda_m)
\right\|
=
o_p(\eta_n).
\end{equation}
Combining \eqref{eq:proof-thm4-S-opeta} with Assumption~1 and Weyl's
inequality gives
$\max_m\|\widehat{\boldsymbol S}_m^{-1}\|=O_p(1)$ and
\begin{equation}
\label{eq:proof-thm4-Sinv-expansion}
\max_{1\le m\le M}
\left\|
\widehat{\boldsymbol S}_m^{-1}
-
\boldsymbol\Sigma(\lambda_m)^{-1}
\right\|
=
o_p(\eta_n).
\end{equation}
Indeed,
\begin{equation*}
\widehat{\boldsymbol S}_m^{-1}
-
\boldsymbol\Sigma(\lambda_m)^{-1}
=
\widehat{\boldsymbol S}_m^{-1}
\{\boldsymbol\Sigma(\lambda_m)-\widehat{\boldsymbol S}_m\}
\boldsymbol\Sigma(\lambda_m)^{-1},
\end{equation*}
and both inverse matrices are uniformly bounded with probability tending to
one.

Next consider $\widehat{\boldsymbol R}_m$. Its population counterpart is
\begin{equation*}
\boldsymbol\Gamma(\lambda_m)
=
\E_{\boldsymbol x\sim\rho_{\mathcal X}^{S}}
\left[
\delta(\boldsymbol x)^{1+\lambda_m}
\boldsymbol h(\boldsymbol x)\boldsymbol h(\boldsymbol x)^\top
\right].
\end{equation*}
The oracle empirical part satisfies
\begin{equation*}
\max_{1\le m\le M}
\left\|
\frac{1}{n_s}
\sum_{i=1}^{n_s}
\delta(\boldsymbol x_i)^{1+\lambda_m}
\boldsymbol h(\boldsymbol x_i)\boldsymbol h(\boldsymbol x_i)^\top
-
\boldsymbol\Gamma(\lambda_m)
\right\|
=
O_p\left(\frac{M}{\sqrt{n_s}}\right)
=
o_p(\eta_n),
\end{equation*}
by the same entrywise Chebyshev and union-bound argument. Here Assumption~4
bounds $\delta(\boldsymbol x)^{1+\lambda_m}$ uniformly because
$1+\lambda_m\in[1,2]$, and Assumption~7 gives the fourth moment of
$\boldsymbol h(\boldsymbol x)$.

For the feasible replacement part, since $\lambda_M=1$ belongs to the
candidate exponent set, the candidate-set definition of $\kappa_n$ gives
\begin{equation*}
\max_{1\le i\le n_s}
|\widehat\delta(\boldsymbol x_i)-\delta(\boldsymbol x_i)|
\le
\kappa_n .
\end{equation*}
Together with Assumption~4 and $\kappa_n=o_p(1)$, this gives
$\max_i\widehat\delta(\boldsymbol x_i)=O_p(1)$. For each candidate exponent
$\lambda_m$,
\begin{align*}
&
\max_{1\le i\le n_s}
\left|
\widehat\delta(\boldsymbol x_i)^{1+\lambda_m}
-
\delta(\boldsymbol x_i)^{1+\lambda_m}
\right|
\\
&\quad
\le
\max_i\widehat\delta(\boldsymbol x_i)^{\lambda_m}
\max_i|\widehat\delta(\boldsymbol x_i)-\delta(\boldsymbol x_i)|
+
\overline\delta
\max_i
\left|
\widehat\delta(\boldsymbol x_i)^{\lambda_m}
-
\delta(\boldsymbol x_i)^{\lambda_m}
\right|.
\end{align*}
Taking the maximum over $m=1,\ldots,M$ on both sides gives
\begin{equation*}
\max_{1\le m\le M}
\max_{1\le i\le n_s}
\left|
\widehat\delta(\boldsymbol x_i)^{1+\lambda_m}
-
\delta(\boldsymbol x_i)^{1+\lambda_m}
\right|
=
O_p(\kappa_n)
=
o_p(\eta_n).
\end{equation*}
Consequently,
\begin{align*}
&
\max_{1\le m\le M}
\left\|
\frac{1}{n_s}
\sum_{i=1}^{n_s}
\left\{
\widehat\delta(\boldsymbol x_i)^{1+\lambda_m}
-
\delta(\boldsymbol x_i)^{1+\lambda_m}
\right\}
\boldsymbol h(\boldsymbol x_i)\boldsymbol h(\boldsymbol x_i)^\top
\right\|
\\
&\quad
\le
O_p(\kappa_n)
\frac{1}{n_s}
\sum_{i=1}^{n_s}
\|\boldsymbol h(\boldsymbol x_i)\|_2^2
=
O_p(\kappa_n)
=
o_p(\eta_n).
\end{align*}
Combining the oracle and feasible bounds yields
\begin{equation}
\label{eq:proof-thm4-R-opeta}
\max_{1\le m\le M}
\left\|
\widehat{\boldsymbol R}_m
-
\boldsymbol\Gamma(\lambda_m)
\right\|
=
o_p(\eta_n).
\end{equation}

Similarly, using the unlabeled target covariates and the identity
\begin{equation*}
\E_{\boldsymbol x\sim\rho_{\mathcal X}^{T}}
\|\boldsymbol h(\boldsymbol x)\|_2^4
=
\E_{\boldsymbol x\sim\rho_{\mathcal X}^{S}}
\left[
\delta(\boldsymbol x)\|\boldsymbol h(\boldsymbol x)\|_2^4
\right]
<\infty,
\end{equation*}
Chebyshev's inequality applied entrywise gives
\begin{equation}
\label{eq:proof-thm4-Sigmat-opeta}
\left\|
\widehat{\boldsymbol\Sigma}^{T}
-
\boldsymbol\Sigma^{T}
\right\|
=
O_p(n_t^{-1/2})
=
o_p(\eta_n),
\end{equation}
where the last equality follows from Assumption~9 and
$n=\min\{n_s,n_t\}$.

Combining \eqref{eq:proof-thm4-trace-representation},
\eqref{eq:proof-thm4-Sinv-expansion}, \eqref{eq:proof-thm4-R-opeta}, and
\eqref{eq:proof-thm4-Sigmat-opeta}, we obtain
\begin{equation}
\label{eq:proof-trace-nu-uniform}
\max_{1\le m\le M}
\left|
n_s\tau_{m,n}
-
\nu(\lambda_m)
\right|
=
o_p(\eta_n).
\end{equation}
Indeed, the population counterpart of $n_s\tau_{m,n}$ is
\begin{equation*}
\operatorname{tr}
\left[
\boldsymbol\Sigma^{T}
\boldsymbol\Sigma(\lambda_m)^{-1}
\boldsymbol\Gamma(\lambda_m)
\boldsymbol\Sigma(1)^{-1}
\right].
\end{equation*}
By the importance-weighting identity, $\boldsymbol\Sigma^{T}=\boldsymbol\Sigma(1)$.
Hence, by cyclic invariance of the trace,
\begin{align*}
\operatorname{tr}
\left[
\boldsymbol\Sigma^{T}
\boldsymbol\Sigma(\lambda_m)^{-1}
\boldsymbol\Gamma(\lambda_m)
\boldsymbol\Sigma(1)^{-1}
\right]
&=
\operatorname{tr}
\left[
\boldsymbol\Sigma(1)
\boldsymbol\Sigma(\lambda_m)^{-1}
\boldsymbol\Gamma(\lambda_m)
\boldsymbol\Sigma(1)^{-1}
\right]
\\
&=
\operatorname{tr}
\left\{
\boldsymbol\Sigma(\lambda_m)^{-1}
\boldsymbol\Gamma(\lambda_m)
\right\}
=
\nu(\lambda_m).
\end{align*}

We next control the residual variance estimator. Under correct specification,
\begin{equation*}
\widehat\sigma^2
=
\frac{\|\boldsymbol G\boldsymbol\varepsilon\|_2^2}{\operatorname{tr}(\boldsymbol G)}.
\end{equation*}
Since $\boldsymbol G$ is symmetric and idempotent,
$\operatorname{tr}(\boldsymbol G)=n_s-p$ and
$\operatorname{tr}(\boldsymbol G^2)=n_s-p$. Conditional on $\mathcal F$,
$\E(\widehat\sigma^2\mid\mathcal F)=\sigma^2$. Assumption~2 and the standard
variance bound for quadratic forms in conditionally independent errors give
$\operatorname{Var}(\widehat\sigma^2\mid\mathcal F)=O(n_s^{-1})$. Thus
\begin{equation}
\label{eq:proof-sigma-consistency}
\widehat\sigma^2
=
\sigma^2+o_p(1).
\end{equation}
In particular, because $\sigma^2>0$, $\widehat\sigma^2$ is bounded away from
zero with probability tending to one.

We now compare the criterion at the selected weights with the criterion at the
OLS endpoint. Let $\boldsymbol w^{(1)}$ denote the vertex that places unit
mass on the first candidate. Since $\widehat{\boldsymbol w}$ minimizes
$C(\boldsymbol w)$ over $\mathcal W$,
\begin{equation}
\label{eq:proof-min-ineq}
C(\widehat{\boldsymbol w})-C(\boldsymbol w^{(1)})\le0.
\end{equation}
Using \eqref{eq:proof-thm4-C-decomp}, we have
\begin{align*}
C(\widehat{\boldsymbol w})-C(\boldsymbol w^{(1)})
={}&
Q_n(\widehat{\boldsymbol w})-Q_n(\boldsymbol w^{(1)})
+
\phi_n\widehat\sigma^2
\sum_{m=1}^{M}\widehat w_m(\tau_{m,n}-\tau_{1,n}).
\end{align*}
Since $Q_n(\widehat{\boldsymbol w})\ge0$,
\begin{equation}
\label{eq:proof-min-lower}
C(\widehat{\boldsymbol w})-C(\boldsymbol w^{(1)})
\ge
-
Q_n(\boldsymbol w^{(1)})
+
\phi_n\widehat\sigma^2
\sum_{m=1}^{M}\widehat w_m(\tau_{m,n}-\tau_{1,n}).
\end{equation}

We next bound $Q_n(\boldsymbol w^{(1)})$. By
\eqref{eq:proof-thm4-theta-noise},
\begin{equation*}
Q_n(\boldsymbol w^{(1)})
=
\boldsymbol\varepsilon^\top
(\boldsymbol L_1-\boldsymbol L_M)^\top
\widehat{\boldsymbol\Sigma}^{T}
(\boldsymbol L_1-\boldsymbol L_M)
\boldsymbol\varepsilon .
\end{equation*}
The bound $\max_m\|\boldsymbol L_m\|_F=O_p(n_s^{-1/2})$ follows from the same
argument as in Lemma~\ref{lem:proof-thm2-uniform-bounds}. Since
$\|\widehat{\boldsymbol\Sigma}^{T}\|=O_p(1)$,
\begin{equation*}
\E\{Q_n(\boldsymbol w^{(1)})\mid\mathcal F\}=O_p(n_s^{-1}).
\end{equation*}
By Markov's inequality,
\begin{equation}
\label{eq:proof-Q1-rate}
Q_n(\boldsymbol w^{(1)})
=
O_p(n_s^{-1}).
\end{equation}

By \eqref{eq:proof-trace-nu-uniform}, uniformly in $m$,
$n_s\tau_{m,n}=\nu(\lambda_m)+o_p(\eta_n)$. Since
$\widehat{\boldsymbol w}\in\mathcal W$,
\begin{equation}
\label{eq:proof-weighted-trace-gap}
n_s
\sum_{m=1}^{M}
\widehat w_m(\tau_{m,n}-\tau_{1,n})
=
\sum_{m=1}^{M}
\widehat w_m
\{\nu(\lambda_m)-\nu(0)\}
+
o_p(\eta_n).
\end{equation}
By Assumption~8,
$\nu(\lambda_m)\ge\nu(0)$ for $m=1,\ldots,M$. Moreover, by the definition of
$\eta_n$,
\begin{equation*}
\nu(\lambda_m)-\nu(0)\ge\eta_n,
\qquad m\notin\mathcal N_n .
\end{equation*}
Therefore,
\begin{equation}
\label{eq:proof-penalty-gap}
n_s
\sum_{m=1}^{M}
\widehat w_m(\tau_{m,n}-\tau_{1,n})
\ge
\eta_n
\sum_{m\notin\mathcal N_n}
\widehat w_m
+
o_p(\eta_n).
\end{equation}

Combining \eqref{eq:proof-min-ineq}, \eqref{eq:proof-min-lower},
\eqref{eq:proof-Q1-rate}, \eqref{eq:proof-sigma-consistency}, and
\eqref{eq:proof-penalty-gap}, we obtain
\begin{equation*}
0
\ge
-
O_p(n_s^{-1})
+
\frac{\phi_n}{n_s}
\{\sigma^2+o_p(1)\}
\left[
\eta_n
\sum_{m\notin\mathcal N_n}
\widehat w_m
+
o_p(\eta_n)
\right].
\end{equation*}
Since $\sigma^2>0$, the preceding inequality implies
\begin{equation*}
\sum_{m\notin\mathcal N_n}
\widehat w_m
\le
O_p\left(\frac{1}{\phi_n\eta_n}\right)
+
o_p(1).
\end{equation*}
By Assumption~9, $\phi_n\eta_n\to\infty$. Hence
\begin{equation*}
\sum_{m\notin\mathcal N_n( a_{n})}
\widehat w_m
\xrightarrow{p}
0,
\end{equation*}
which proves (28).

It remains to prove the second assertion. If there exists a constant $c>0$
such that $\lambda_2\ge c$ for all sufficiently large $n$, then, because
$ a_{n}\to0$, $\mathcal N_n( a_{n})=\{1\}$ eventually.
Therefore,
\begin{equation*}
1-\widehat w_1
=
\sum_{m\notin\mathcal N_n( a_{n})}
\widehat w_m
\xrightarrow{p}
0,
\end{equation*}
or equivalently,
$\widehat w_1 \overset{p}{\longrightarrow} 1$.
This proves (29) and completes the proof. 

\qed
\section{Implementation details for the simulation studies}
\label{sec:supp-simulation-implementation}

This section provides implementation details for the density-ratio estimators used in the simulation studies. 

\subsection{One-dimensional misspecification design}
\label{subsec:supp-dr-one-dimensional}

We consider KLIEP, kernel density estimation (KDE), and the oracle density ratio. KLIEP is implemented with Gaussian kernels, and its bandwidth is selected by five-fold cross-validation following \citet{Sugiyamaetal2008}. For KDE, the source and target covariate densities are estimated separately using Gaussian kernels, with bandwidths selected by Silverman's rule of thumb \citep{Silverman1986}; their ratio is then evaluated at the source covariates. The oracle benchmark replaces the estimated density ratio by its true value.

\subsection{Six-dimensional misspecification design}
\label{subsec:supp-dr-six-dimensional}

We consider KLIEP, a Gaussian plug-in estimator, and the oracle density ratio. KLIEP is implemented as above. For the Gaussian plug-in estimator, multivariate Gaussian distributions are fitted separately to the source and target covariates. Let $\widehat{\boldsymbol\mu}^{S}$ and $\widehat V^{S}$ be the source sample mean and covariance matrix, and let $\widehat{\boldsymbol\mu}^{T}$ and $\widehat V^{T}$ be their target-sample counterparts. The density-ratio estimate at a source covariate $\boldsymbol{x}_i$ is
\[
\widehat{\delta}(\boldsymbol{x}_i)
=
\frac{
\varphi_d(\boldsymbol{x}_i;\widehat{\boldsymbol\mu}^{T},\widehat V^{T})
}{
\varphi_d(\boldsymbol{x}_i;\widehat{\boldsymbol\mu}^{S},\widehat V^{S})
},
\]
where $\varphi_d(\cdot;\boldsymbol\mu,V)$ denotes the $d$-dimensional Gaussian density.  The oracle benchmark uses the true density ratio.

\subsection{Correctly specified design}
\label{subsec:supp-dr-correct-specification}

The main paper reports results based on KLIEP under both candidate exponent grids. The supplementary results use KDE and the oracle density ratio. The KDE implementation estimates the source and target covariate densities separately and evaluates their ratio at the source covariates.
The oracle benchmark uses the true density ratio.

\section{Additional simulations under model misspecification}
\label{sec:supp-misspecification-simulations}

Section~5.1 of the main paper reports the six-dimensional results obtained with KLIEP under the fixed candidate grid. This section presents a one-dimensional misspecification experiment and the remaining six-dimensional results, including the diverging-grid KLIEP results and results based on alternative and oracle density ratios.

We consider a one-dimensional design adapted from \citet{SugiyamaKrauledatMuller2007}. The response is generated from
\[
y=f(x)+\varepsilon=\operatorname{sinc}(x)+\varepsilon,
\]
where $\operatorname{sinc}(x)=\sin(x)/x$ for $x\neq0$, $\operatorname{sinc}(0)=1$, and $\varepsilon\sim N(0,\sigma^2)$. The source and target covariate distributions are $N(1,0.5^2)$ and $N(\zeta,0.25^2)$, respectively, with $\zeta\in\{1.8,2.0\}$. For each value of $\zeta$, $\sigma^2$ is chosen so that the target-distribution coefficient of determination $R^2=\operatorname{Var}_{x\sim\rho_{\mathcal X}^{T}}\{f(x)\}/[\operatorname{Var}_{x\sim\rho_{\mathcal X}^{T}}\{f(x)\}+\sigma^2]$ takes values in $\{0.1,0.2,\ldots,0.9\}$. The working basis is $\boldsymbol h(x)=(1,x)^\top$ and is therefore misspecified.

The candidate exponent grids and competing methods are the same as in the main paper. We take $n_s\in\{50,100,150,200\}$, $n_t=1000$, and $n_{\mathrm{eval}}=1000$. Results are averaged over $B=1000$ Monte Carlo replications. For each method, the empirical target excess prediction error is normalized by that of AIWMA-$\log n$.

Across the settings considered, AIWMA-$\log n$ generally delivers favorable predictive performance by balancing insufficient correction against the variance inflation caused by large density-ratio values. OLS remains competitive when both $R^2$ and $n_s$ are small, whereas stronger correction becomes increasingly beneficial as the noise level decreases and the bias induced by covariate shift becomes more pronounced. By contrast, IWLS performs poorly under the stronger shift because of its substantial variance inflation. These qualitative conclusions are consistent across both candidate grids and all three density-ratio specifications.

\subsection{One-dimensional design using KLIEP}
\label{subsec:supp-one-dimensional-kliep}

Figures~\ref{fig:supp-sinc-kliep-fixed-zeta18}--\ref{fig:supp-sinc-kliep-diverging-zeta20} report the normalized TEPE results obtained with KLIEP density-ratio estimation under the fixed and diverging candidate grids.

\begin{figure}[H]
\centering
\includegraphics[
width=0.90\linewidth,
height=0.90\textheight,
keepaspectratio
]{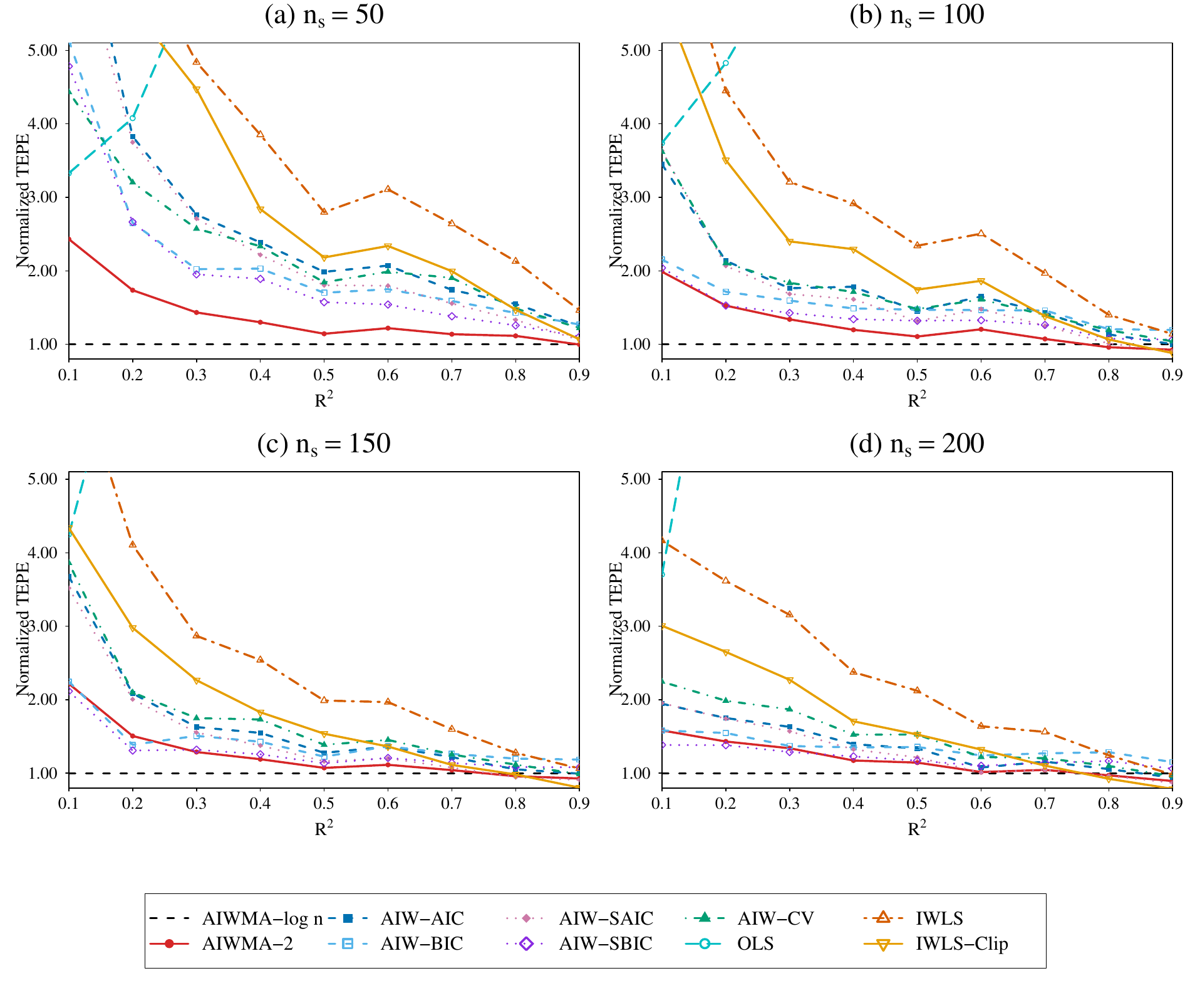}
\caption{Normalized target excess prediction error relative to AIWMA-$\log n$ in the one-dimensional misspecification design, using KLIEP density-ratio estimation and the fixed candidate grid, with $\zeta=1.8$.}
\label{fig:supp-sinc-kliep-fixed-zeta18}
\end{figure}

\begin{figure}[H]
\centering
\includegraphics[
width=0.90\linewidth,
height=0.90\textheight,
keepaspectratio
]{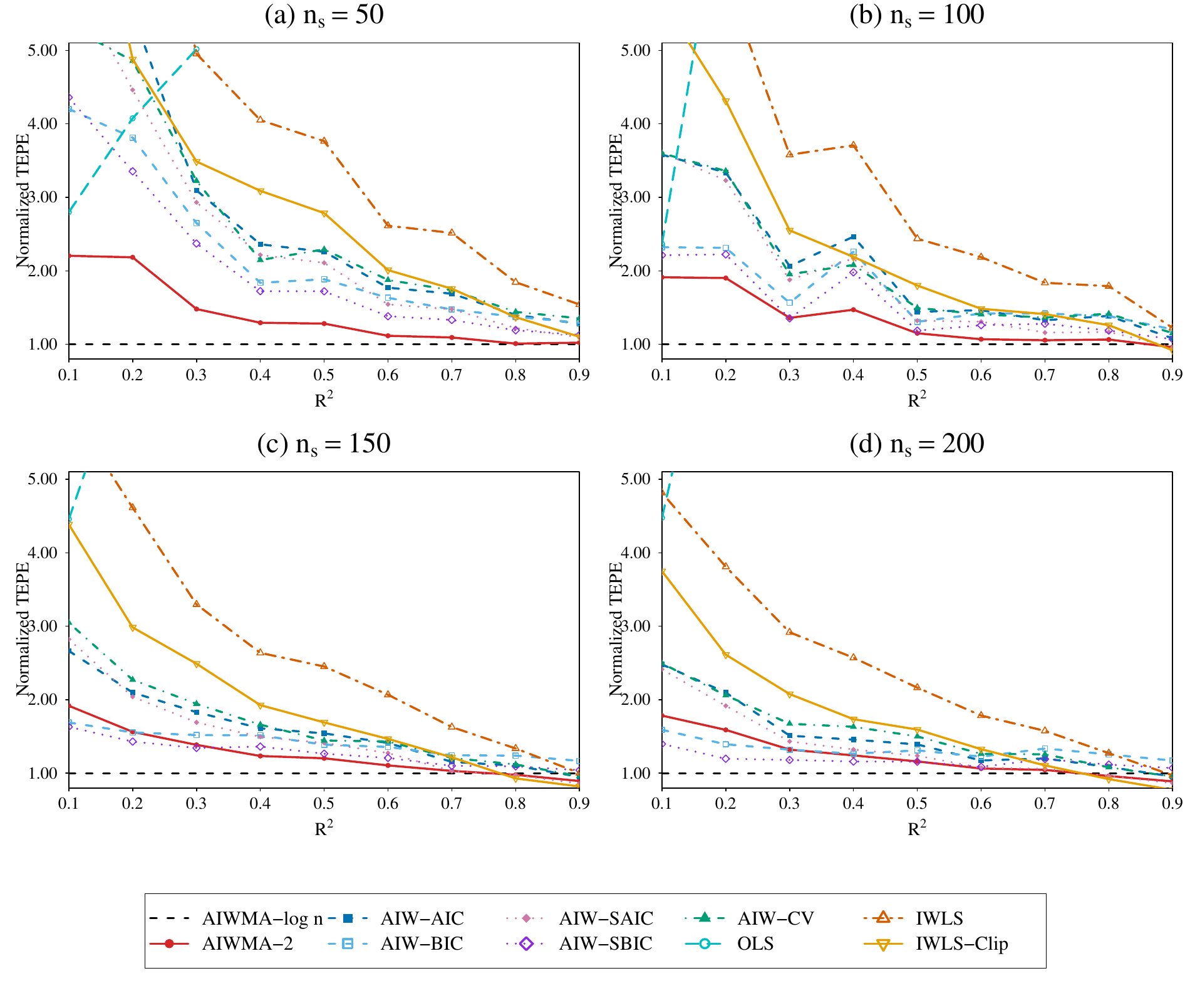}
\caption{Normalized target excess prediction error relative to AIWMA-$\log n$ in the one-dimensional misspecification design, using KLIEP density-ratio estimation and the diverging candidate grid, with $\zeta=1.8$.}
\label{fig:supp-sinc-kliep-diverging-zeta18}
\end{figure}

\begin{figure}[H]
\centering
\includegraphics[
width=0.90\linewidth,
height=0.90\textheight,
keepaspectratio
]{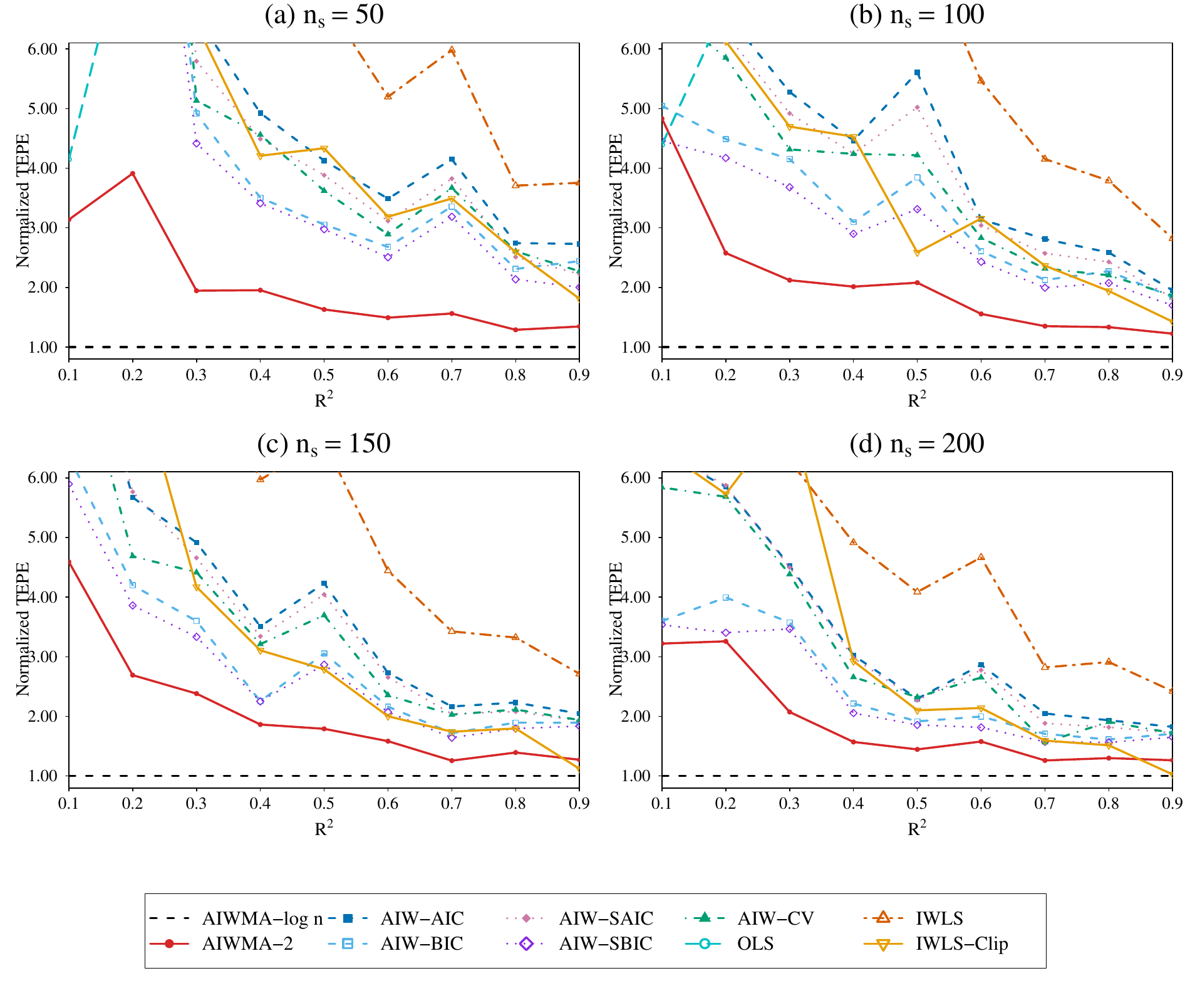}
\caption{Normalized target excess prediction error relative to AIWMA-$\log n$ in the one-dimensional misspecification design, using KLIEP density-ratio estimation and the fixed candidate grid, with $\zeta=2.0$.}
\label{fig:supp-sinc-kliep-fixed-zeta20}
\end{figure}

\begin{figure}[H]
\centering
\includegraphics[
width=0.90\linewidth,
height=0.90\textheight,
keepaspectratio
]{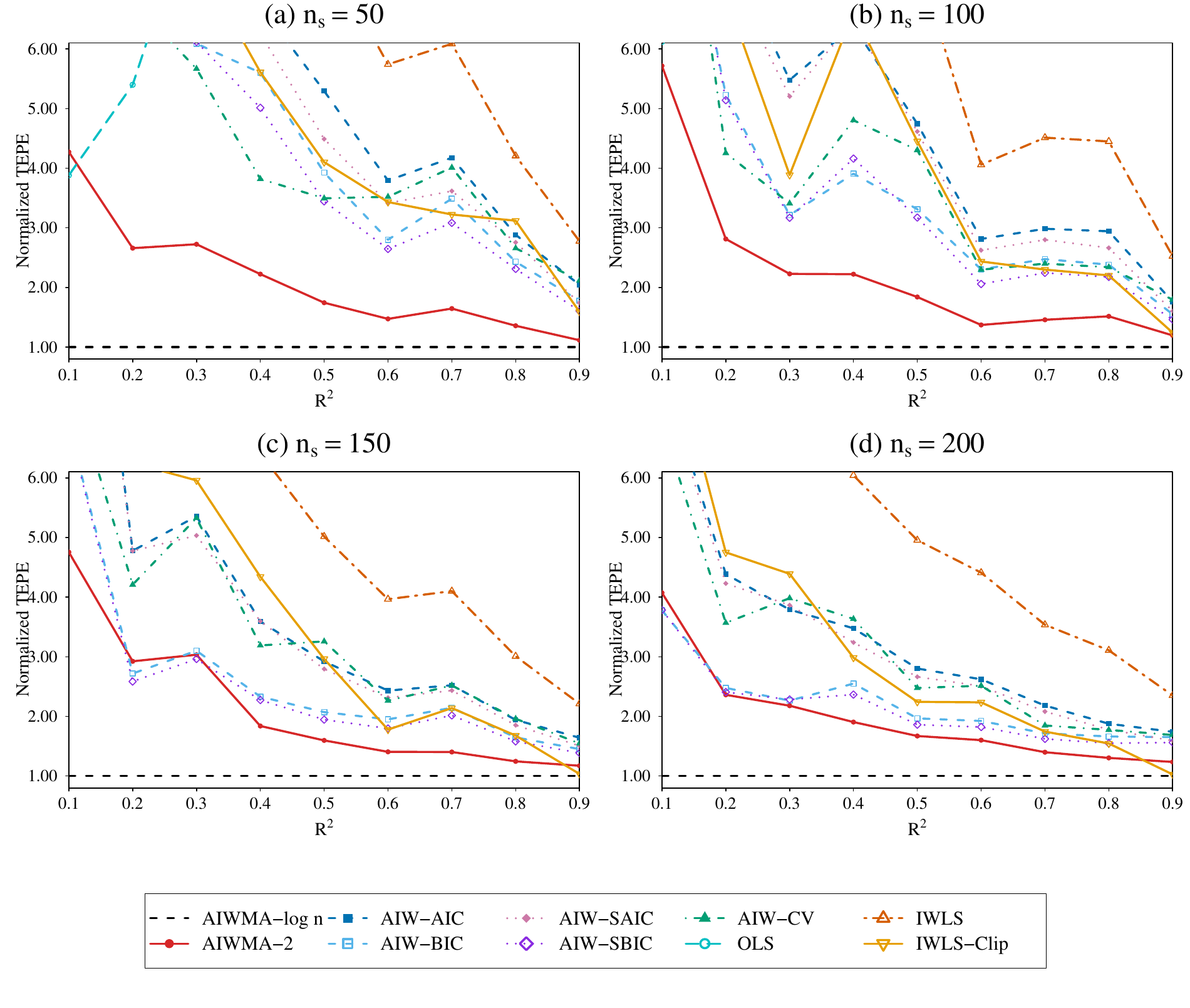}
\caption{Normalized target excess prediction error relative to AIWMA-$\log n$ in the one-dimensional misspecification design, using KLIEP density-ratio estimation and the diverging candidate grid, with $\zeta=2.0$.}
\label{fig:supp-sinc-kliep-diverging-zeta20}
\end{figure}

\subsection{One-dimensional design using KDE}
\label{subsec:supp-one-dimensional-kde}

Figures~\ref{fig:supp-sinc-kde-fixed-zeta18}--\ref{fig:supp-sinc-kde-diverging-zeta20} report the normalized TEPE results obtained with KDE density-ratio estimation under the fixed and diverging candidate grids.

\begin{figure}[H]
\centering
\includegraphics[
width=0.90\linewidth,
height=0.90\textheight,
keepaspectratio
]{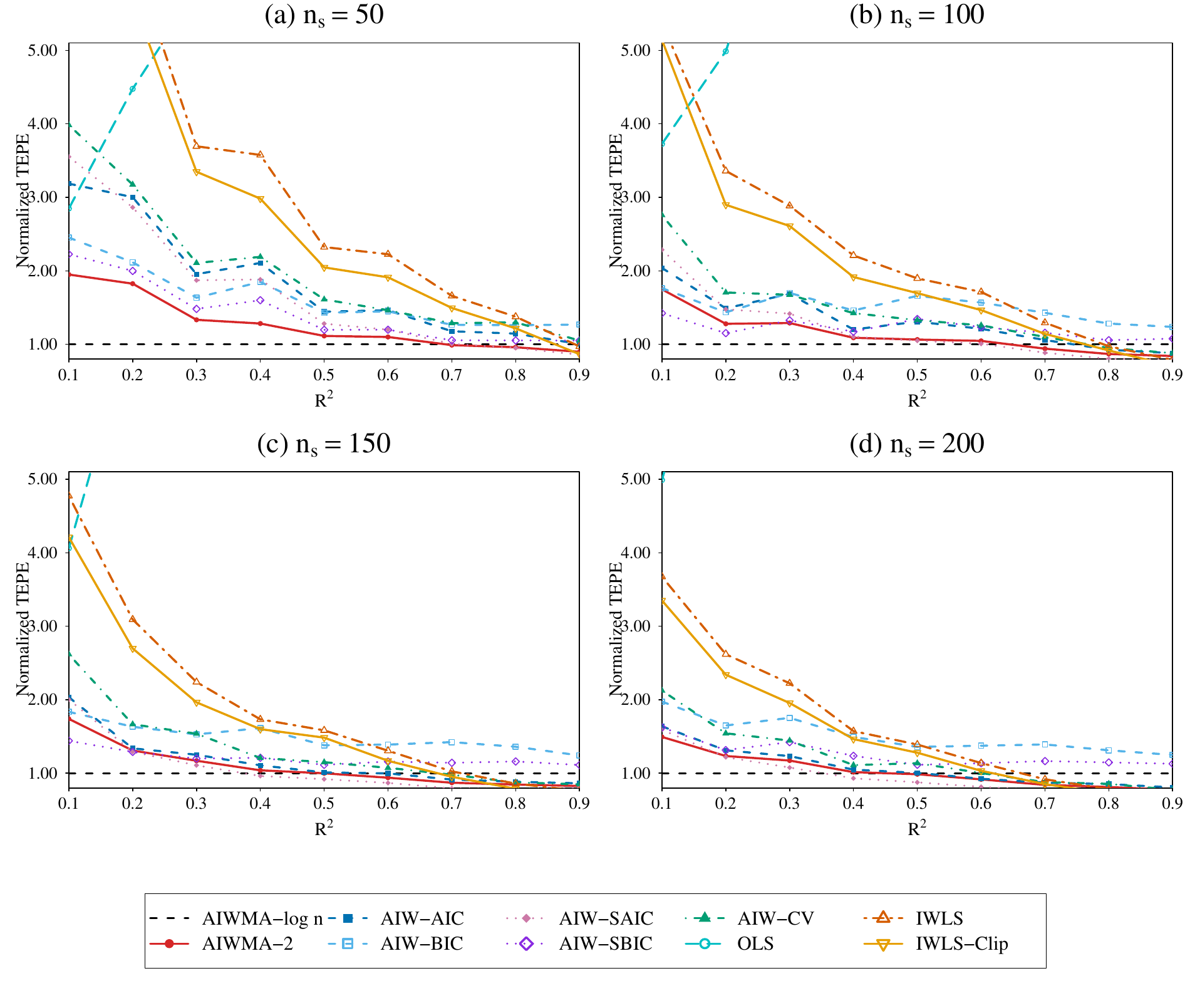}
\caption{Normalized target excess prediction error relative to AIWMA-$\log n$ in the one-dimensional misspecification design, using KDE density-ratio estimation and the fixed candidate grid, with $\zeta=1.8$.}
\label{fig:supp-sinc-kde-fixed-zeta18}
\end{figure}

\begin{figure}[H]
\centering
\includegraphics[
width=0.90\linewidth,
height=0.90\textheight,
keepaspectratio
]{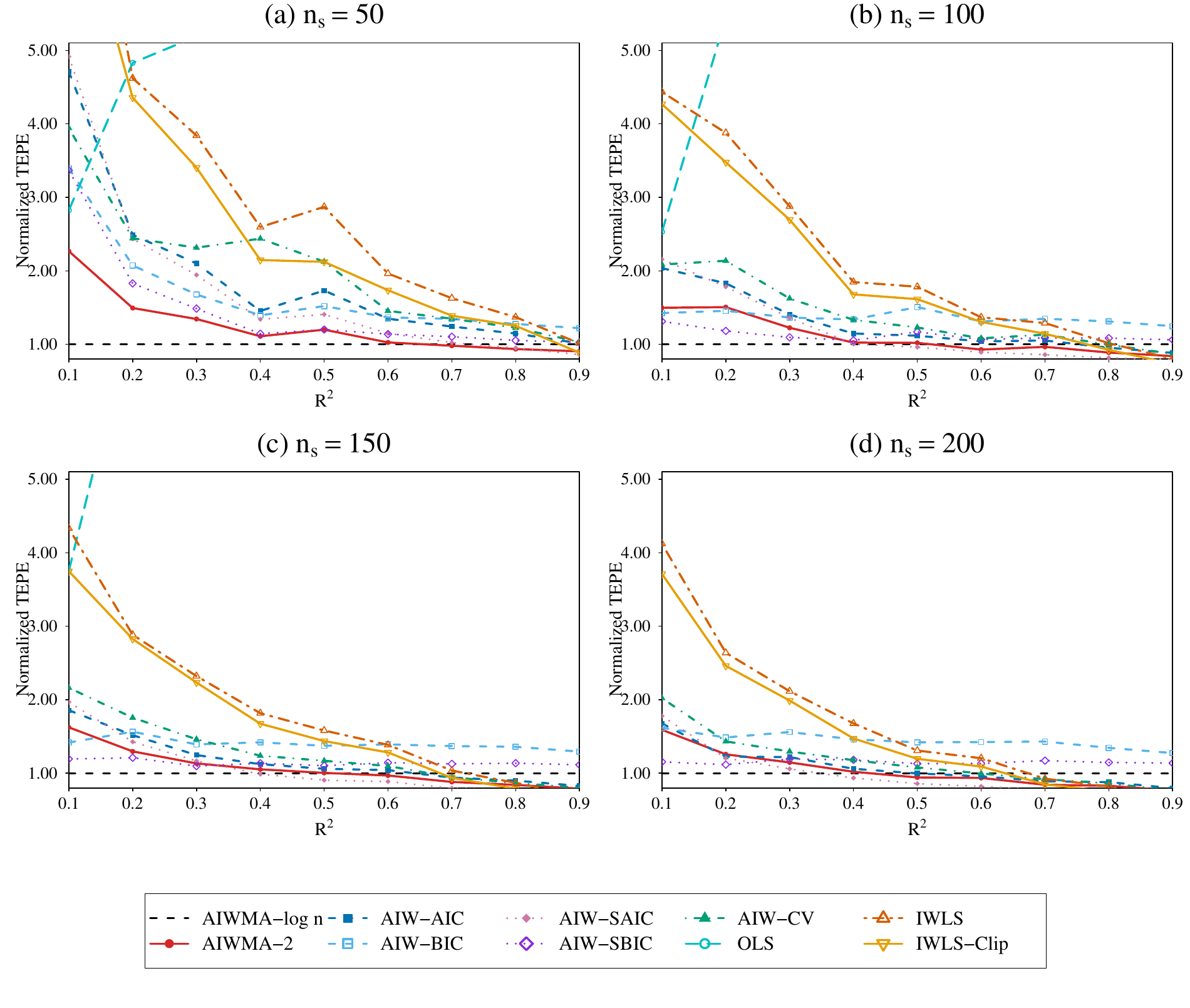}
\caption{Normalized target excess prediction error relative to AIWMA-$\log n$ in the one-dimensional misspecification design, using KDE density-ratio estimation and the diverging candidate grid, with $\zeta=1.8$.}
\label{fig:supp-sinc-kde-diverging-zeta18}
\end{figure}

\begin{figure}[H]
\centering
\includegraphics[
width=0.90\linewidth,
height=0.90\textheight,
keepaspectratio
]{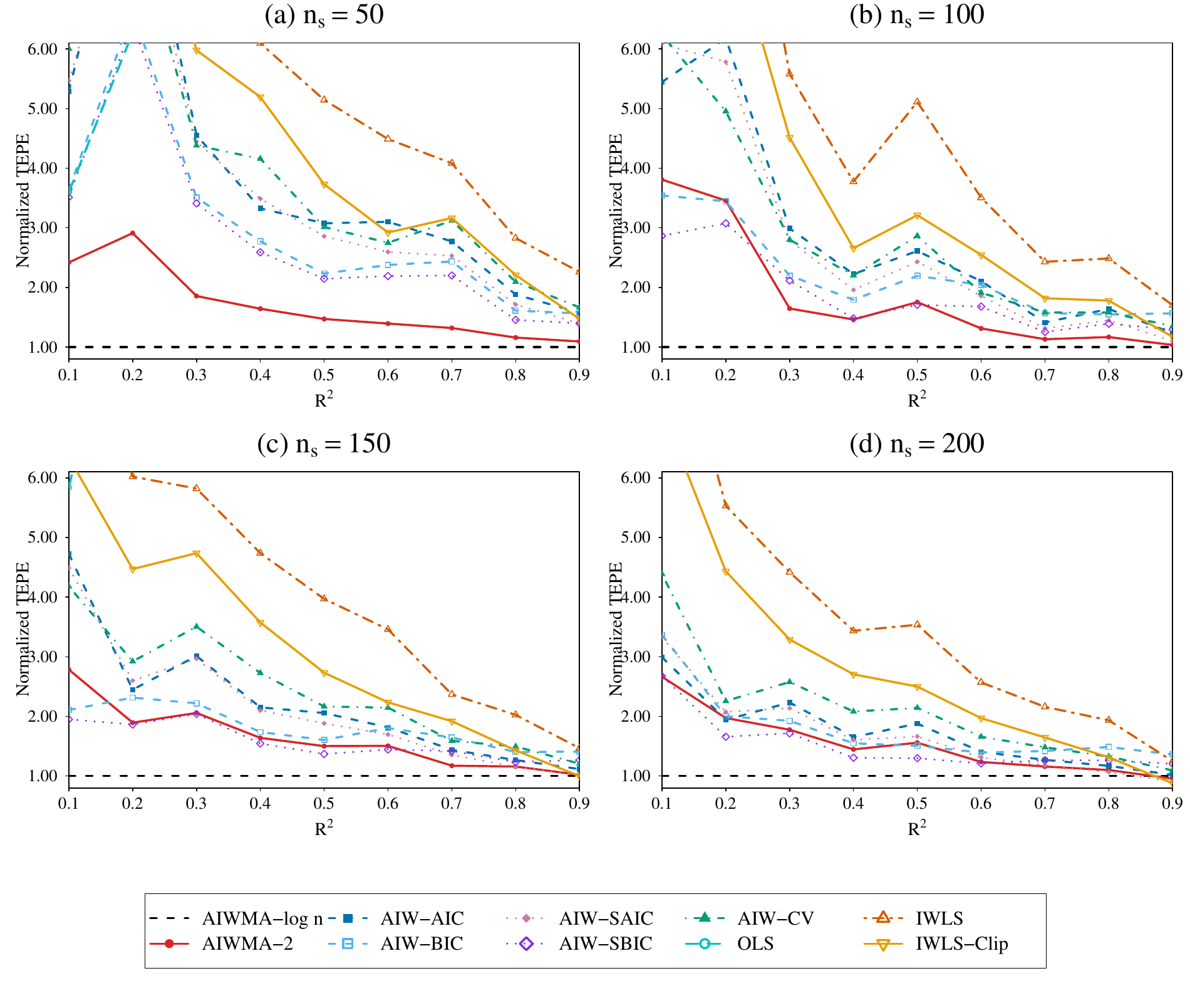}
\caption{Normalized target excess prediction error relative to AIWMA-$\log n$ in the one-dimensional misspecification design, using KDE density-ratio estimation and the fixed candidate grid, with $\zeta=2.0$.}
\label{fig:supp-sinc-kde-fixed-zeta20}
\end{figure}

\begin{figure}[H]
\centering
\includegraphics[
width=0.90\linewidth,
height=0.90\textheight,
keepaspectratio
]{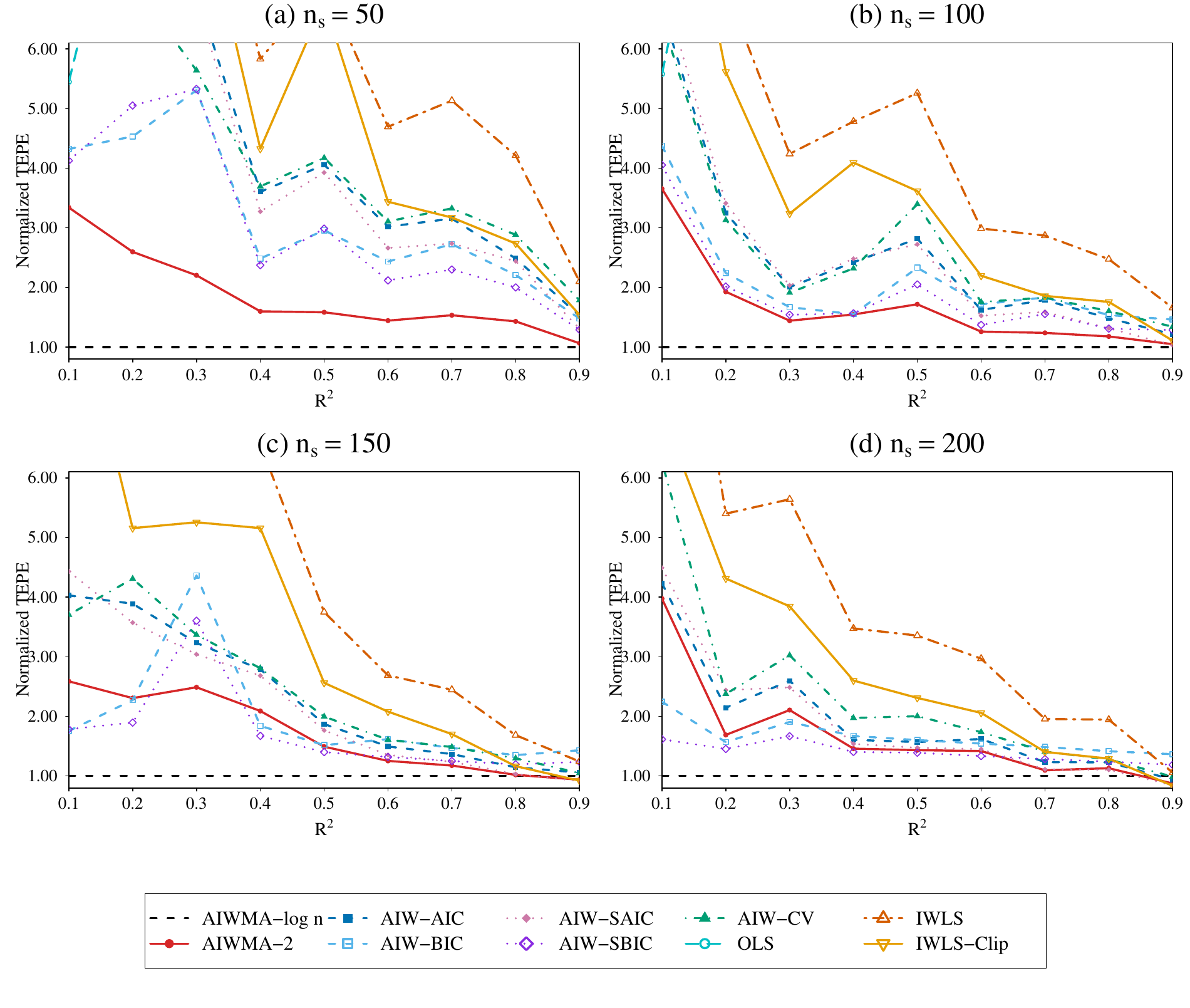}
\caption{Normalized target excess prediction error relative to AIWMA-$\log n$ in the one-dimensional misspecification design, using KDE density-ratio estimation and the diverging candidate grid, with $\zeta=2.0$.}
\label{fig:supp-sinc-kde-diverging-zeta20}
\end{figure}

\subsection{One-dimensional design using the oracle density ratio}
\label{subsec:supp-one-dimensional-oracle}

Figures~\ref{fig:supp-sinc-oracle-fixed-zeta18}--\ref{fig:supp-sinc-oracle-diverging-zeta20} report the normalized TEPE results obtained with the oracle density ratio under the fixed and diverging candidate grids.

\begin{figure}[H]
\centering
\includegraphics[
width=0.90\linewidth,
height=0.90\textheight,
keepaspectratio
]{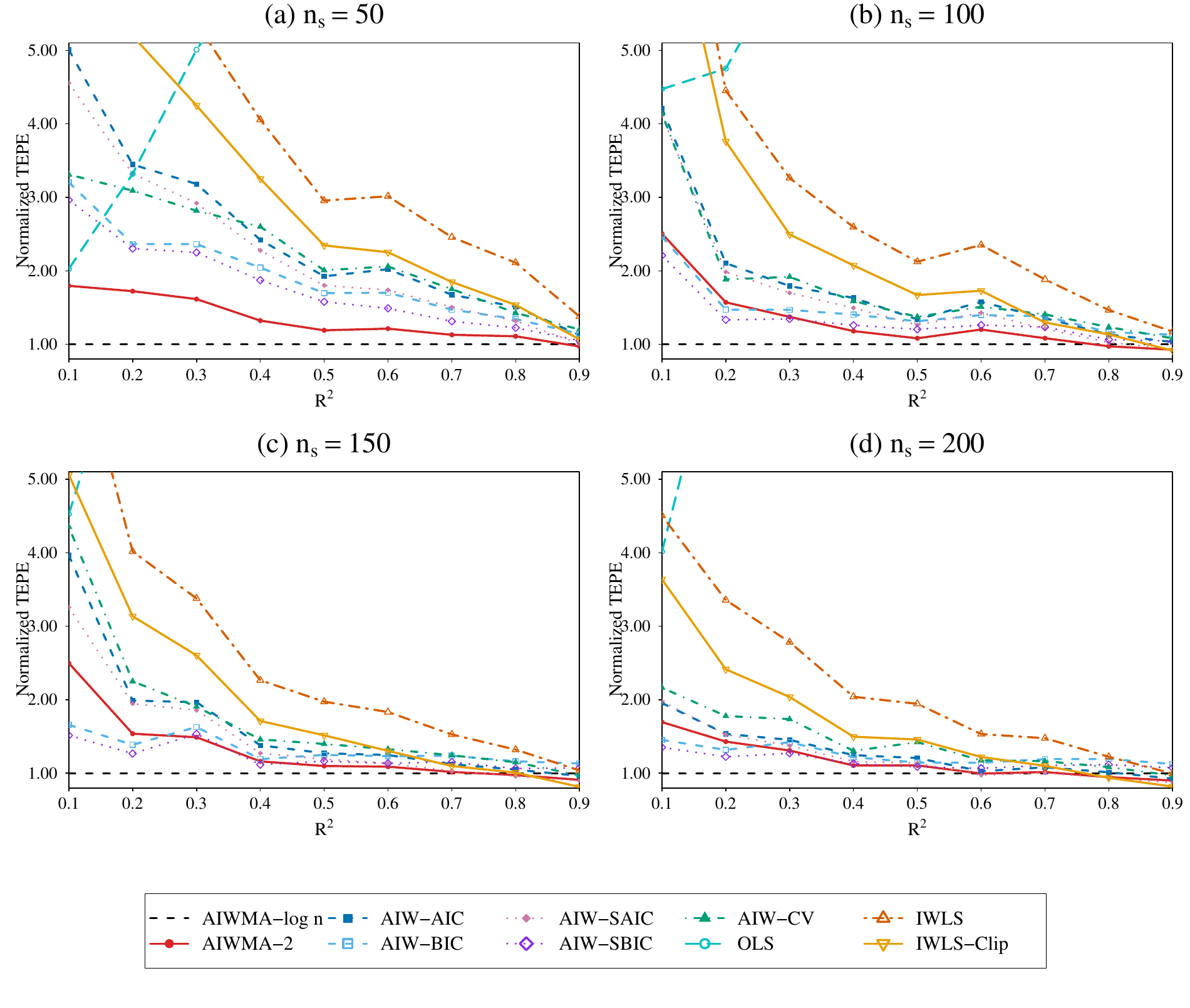}
\caption{Normalized target excess prediction error relative to AIWMA-$\log n$ in the one-dimensional misspecification design, using the oracle density ratio and the fixed candidate grid, with $\zeta=1.8$.}
\label{fig:supp-sinc-oracle-fixed-zeta18}
\end{figure}

\begin{figure}[H]
\centering
\includegraphics[
width=0.90\linewidth,
height=0.90\textheight,
keepaspectratio
]{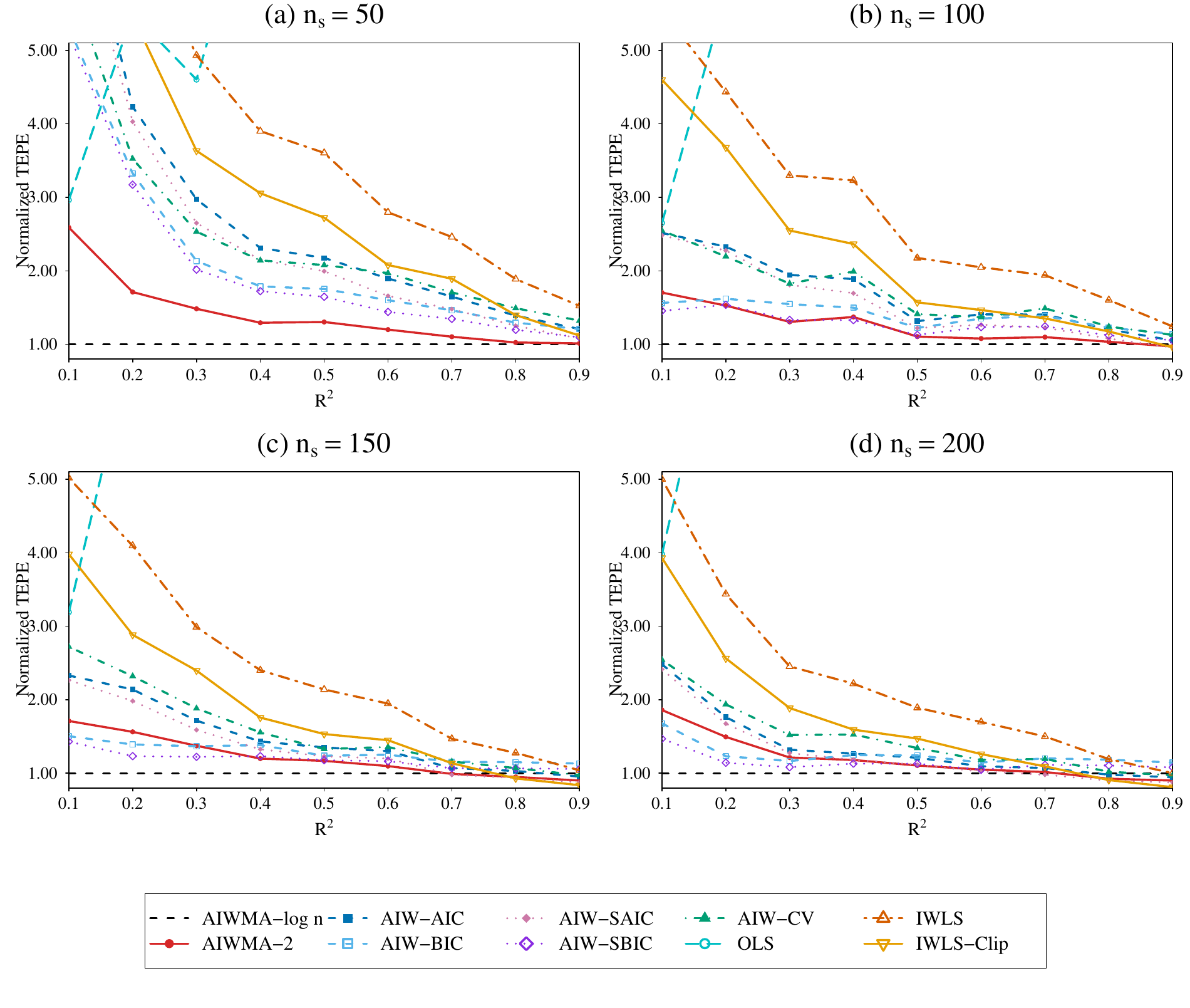}
\caption{Normalized target excess prediction error relative to AIWMA-$\log n$ in the one-dimensional misspecification design, using the oracle density ratio and the diverging candidate grid, with $\zeta=1.8$.}
\label{fig:supp-sinc-oracle-diverging-zeta18}
\end{figure}

\begin{figure}[H]
\centering
\includegraphics[
width=0.90\linewidth,
height=0.90\textheight,
keepaspectratio
]{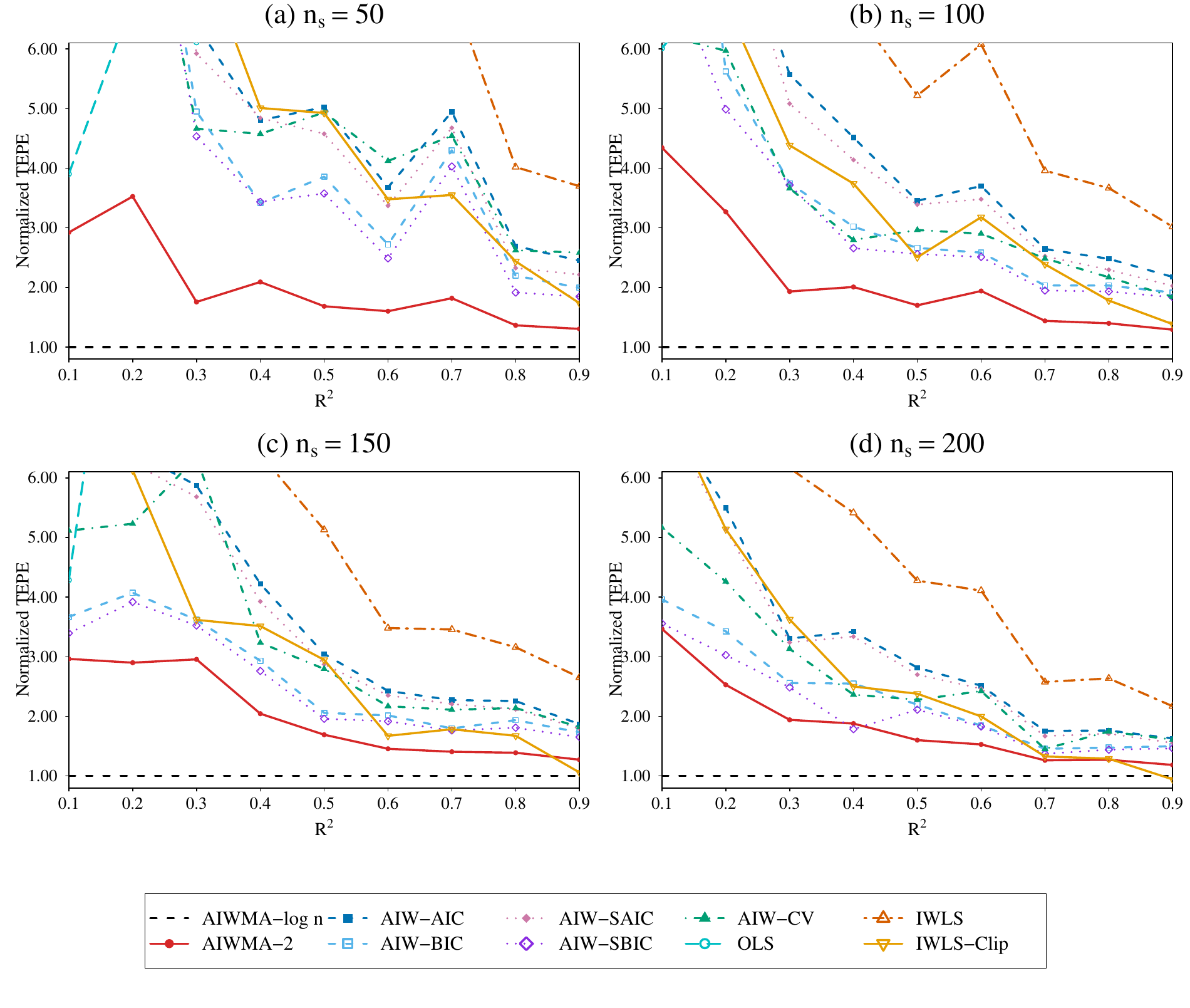}
\caption{Normalized target excess prediction error relative to AIWMA-$\log n$ in the one-dimensional misspecification design, using the oracle density ratio and the fixed candidate grid, with $\zeta=2.0$.}
\label{fig:supp-sinc-oracle-fixed-zeta20}
\end{figure}

\begin{figure}[H]
\centering
\includegraphics[
width=0.90\linewidth,
height=0.90\textheight,
keepaspectratio
]{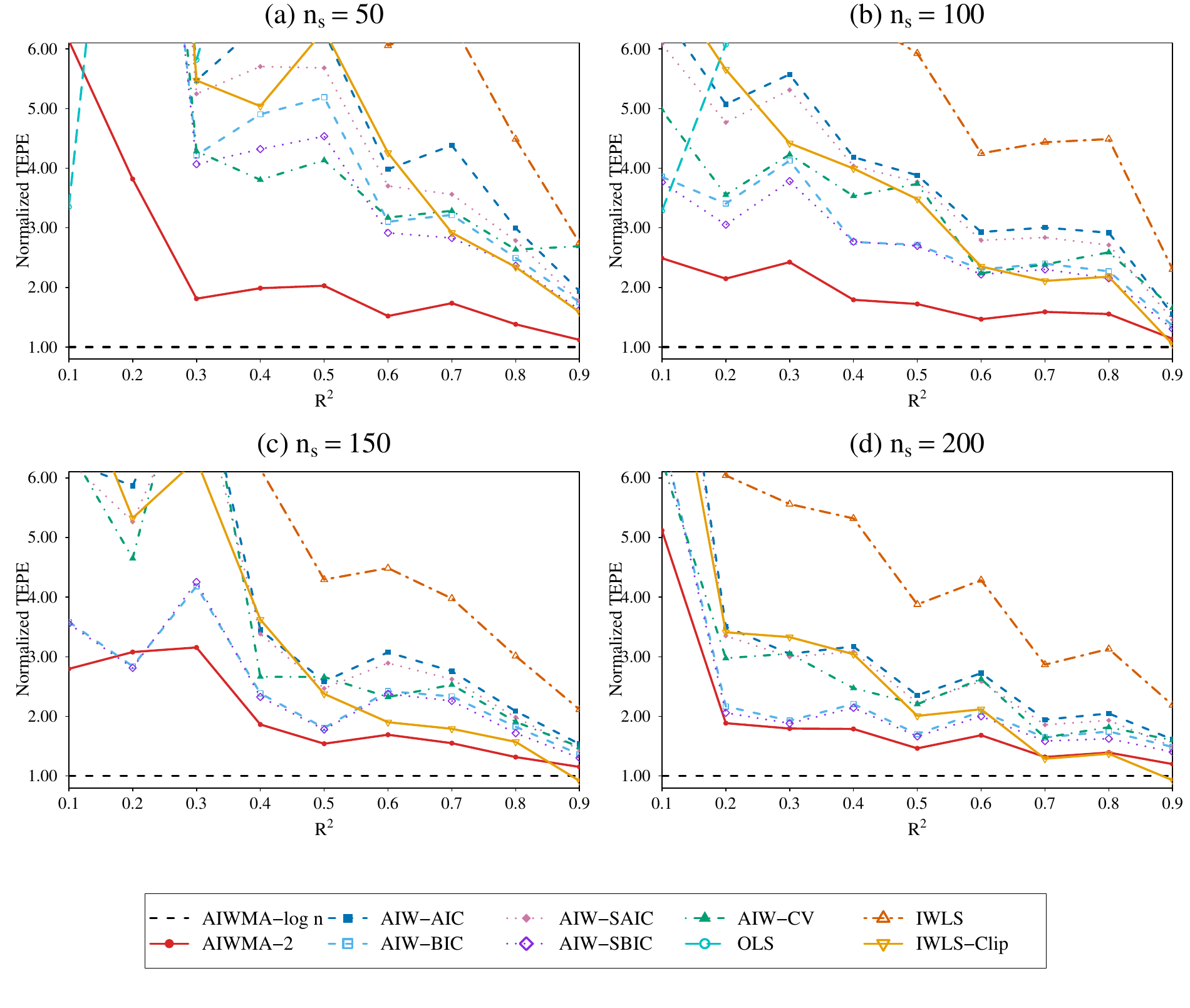}
\caption{Normalized target excess prediction error relative to AIWMA-$\log n$ in the one-dimensional misspecification design, using the oracle density ratio and the diverging candidate grid, with $\zeta=2.0$.}
\label{fig:supp-sinc-oracle-diverging-zeta20}
\end{figure}

We retain the data-generating process, sample sizes, evaluation criterion, and competing methods considered in Section~5.1 of the main paper. The main paper reports the TEPE results obtained using KLIEP density-ratio estimation under the fixed candidate grid. Here, we report the corresponding TEPE results using KLIEP under the diverging grid, together with those obtained using the Gaussian plug-in and oracle density ratios under both candidate grids. Overall, the findings are qualitatively consistent with those in the main paper: AIWMA-$\log n$ continues to deliver favorable predictive performance across different shift strengths, signal levels, candidate grids, and density-ratio specifications.

\subsection{Six-dimensional design using KLIEP}
\label{subsec:supp-six-dimensional-kliep}
Figures~\ref{fig:supp-six-kliep-diverging-zeta12} and~\ref{fig:supp-six-kliep-diverging-zeta15} report the normalized TEPE results obtained with KLIEP density-ratio estimation under the  diverging candidate grids.

\begin{figure}[H]
\centering
\includegraphics[
width=0.90\linewidth,
height=0.90\textheight,
keepaspectratio
]{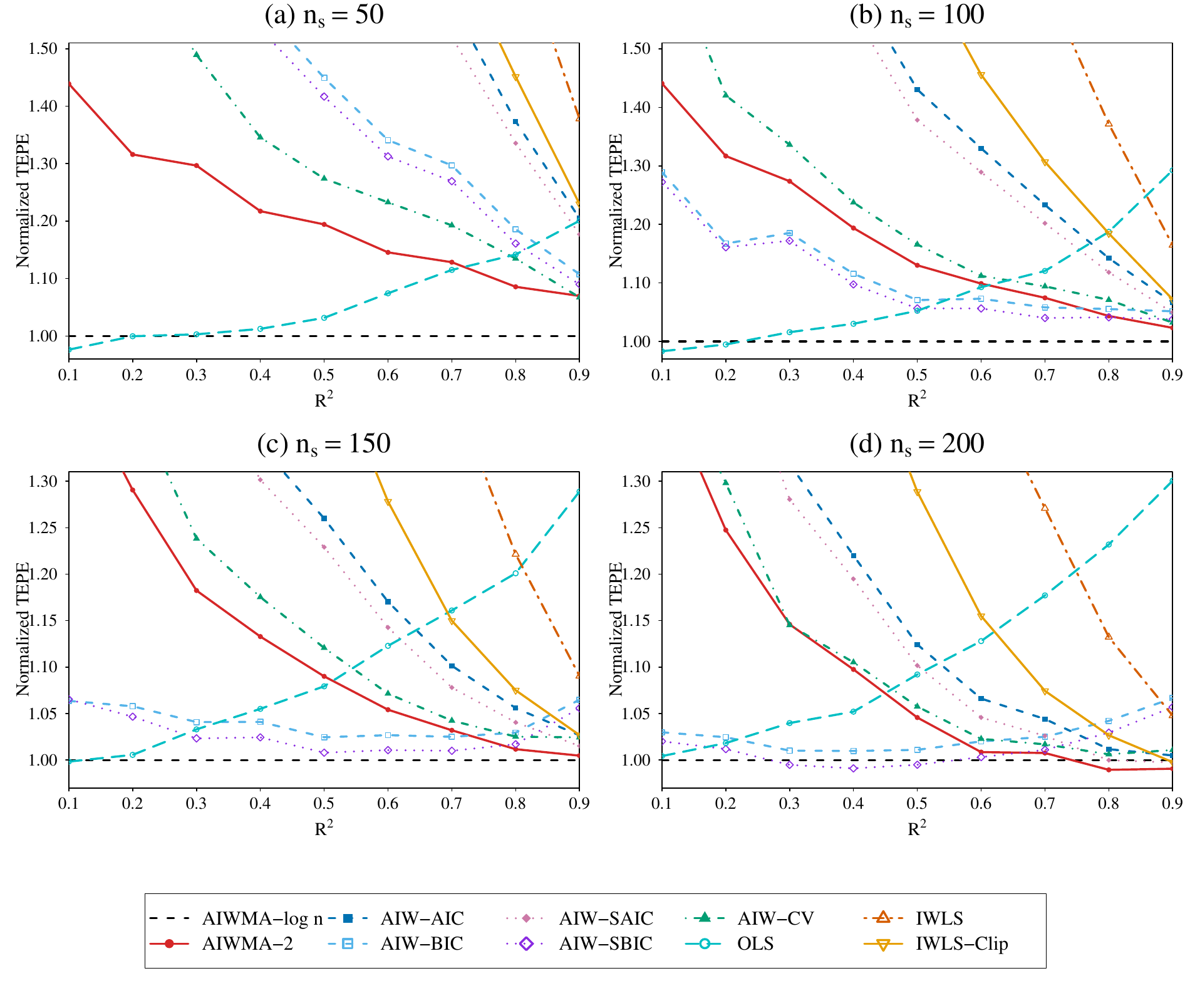}
\caption{Normalized target excess prediction error relative to AIWMA-$\log n$ in the six-dimensional misspecification design, using KLIEP density-ratio estimation and the diverging candidate grid, with $\zeta=1.2$.}
\label{fig:supp-six-kliep-diverging-zeta12}
\end{figure}

\begin{figure}[H]
\centering
\includegraphics[
width=0.90\linewidth,
height=0.90\textheight,
keepaspectratio
]{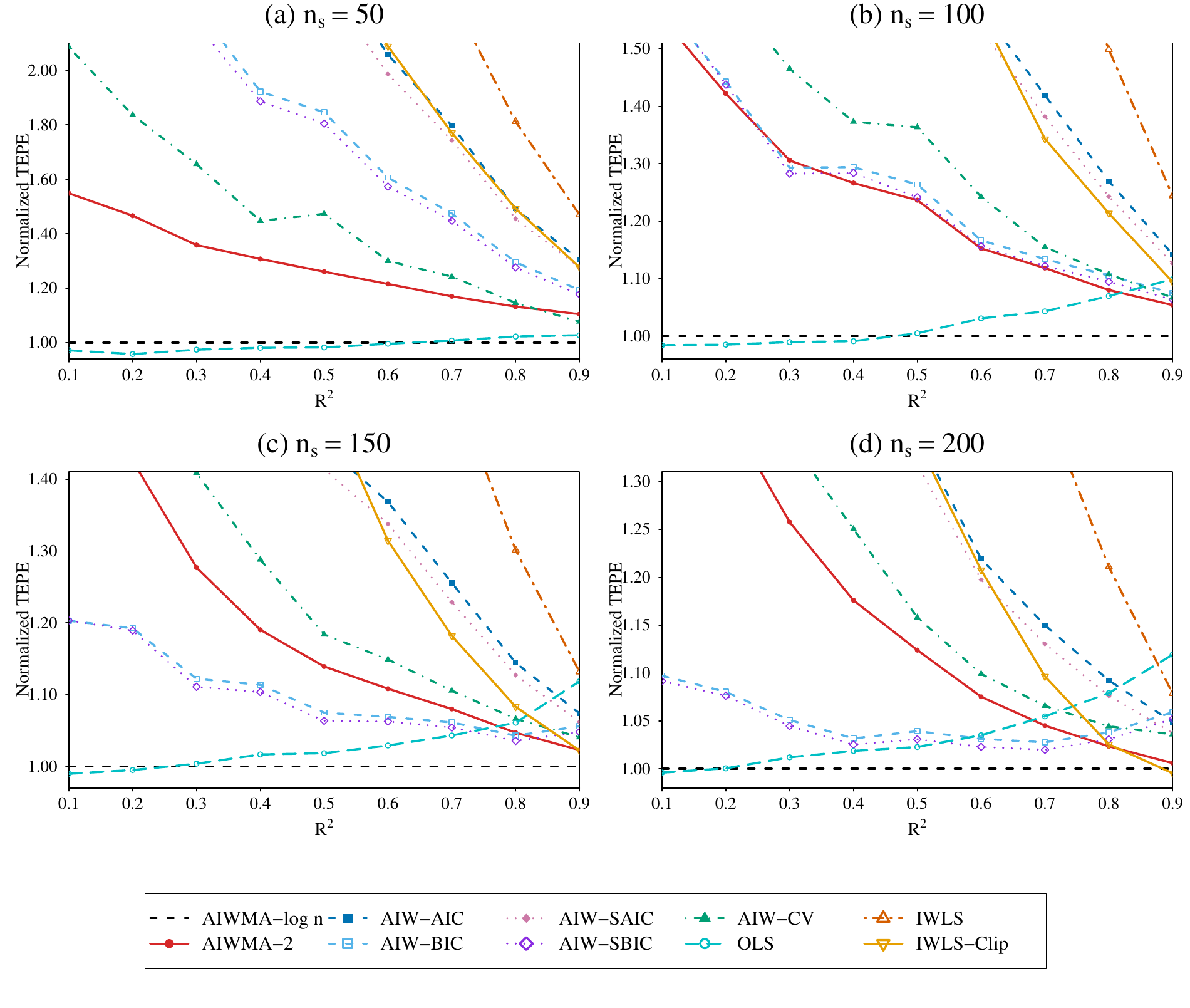}
\caption{Normalized target excess prediction error relative to AIWMA-$\log n$ in the six-dimensional misspecification design, using KLIEP density-ratio estimation and the diverging candidate grid, with $\zeta=1.5$.}
\label{fig:supp-six-kliep-diverging-zeta15}
\end{figure}

\subsection{Six-dimensional design using the Gaussian plug-in estimator}
\label{subsec:supp-six-dimensional-plugin}
Figures~\ref{fig:supp-six-plugin-fixed-zeta12}--\ref{fig:supp-six-plugin-diverging-zeta15} report the normalized TEPE results obtained with Gaussian plug-in density-ratio estimation under the fixed and diverging candidate grids.

\begin{figure}[H]
\centering
\includegraphics[
width=0.90\linewidth,
height=0.90\textheight,
keepaspectratio
]{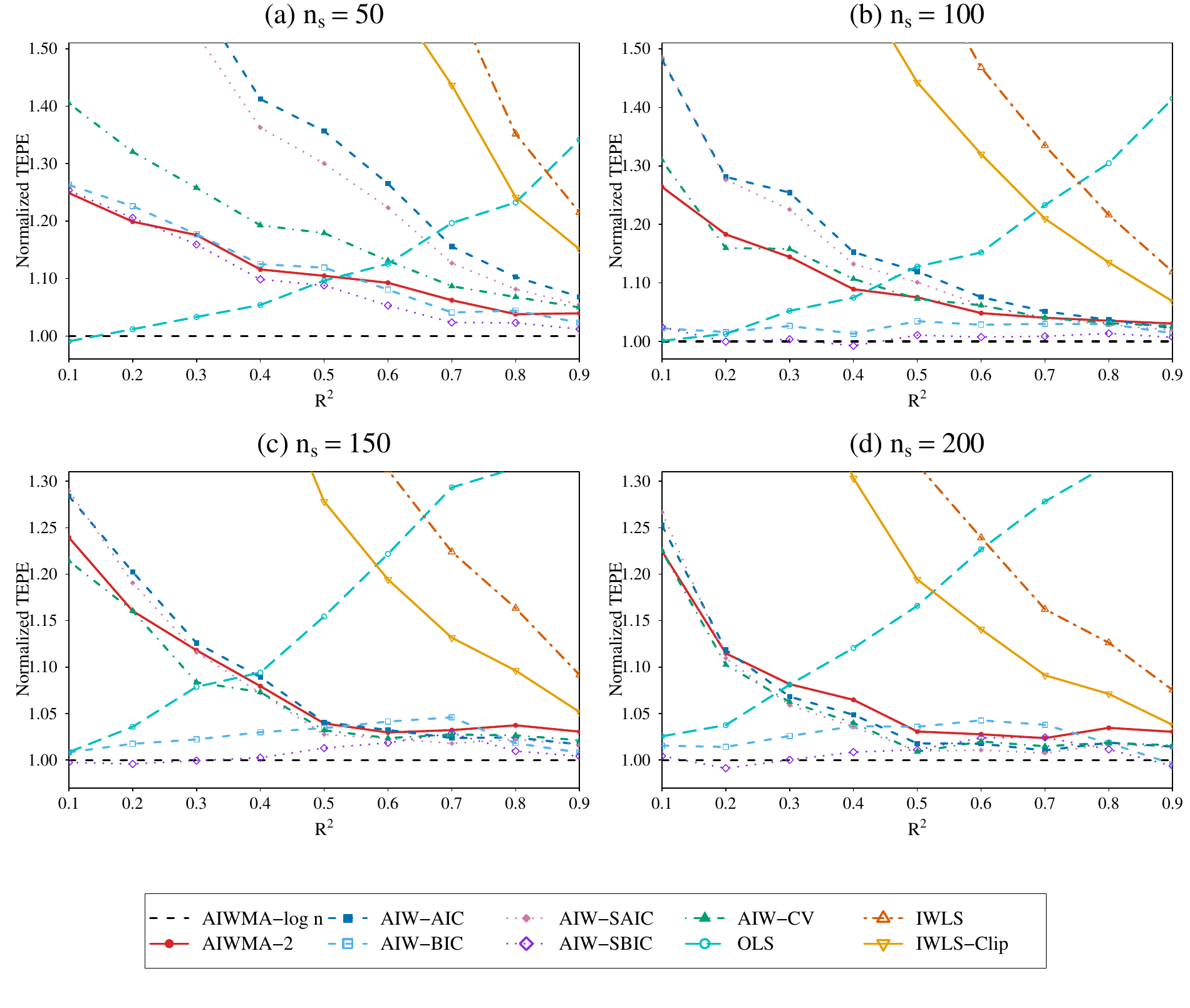}
\caption{Normalized target excess prediction error relative to AIWMA-$\log n$ in the six-dimensional misspecification design, using Gaussian plug-in density-ratio estimation and the fixed candidate grid, with $\zeta=1.2$.}
\label{fig:supp-six-plugin-fixed-zeta12}
\end{figure}

\begin{figure}[H]
\centering
\includegraphics[
width=0.90\linewidth,
height=0.90\textheight,
keepaspectratio
]{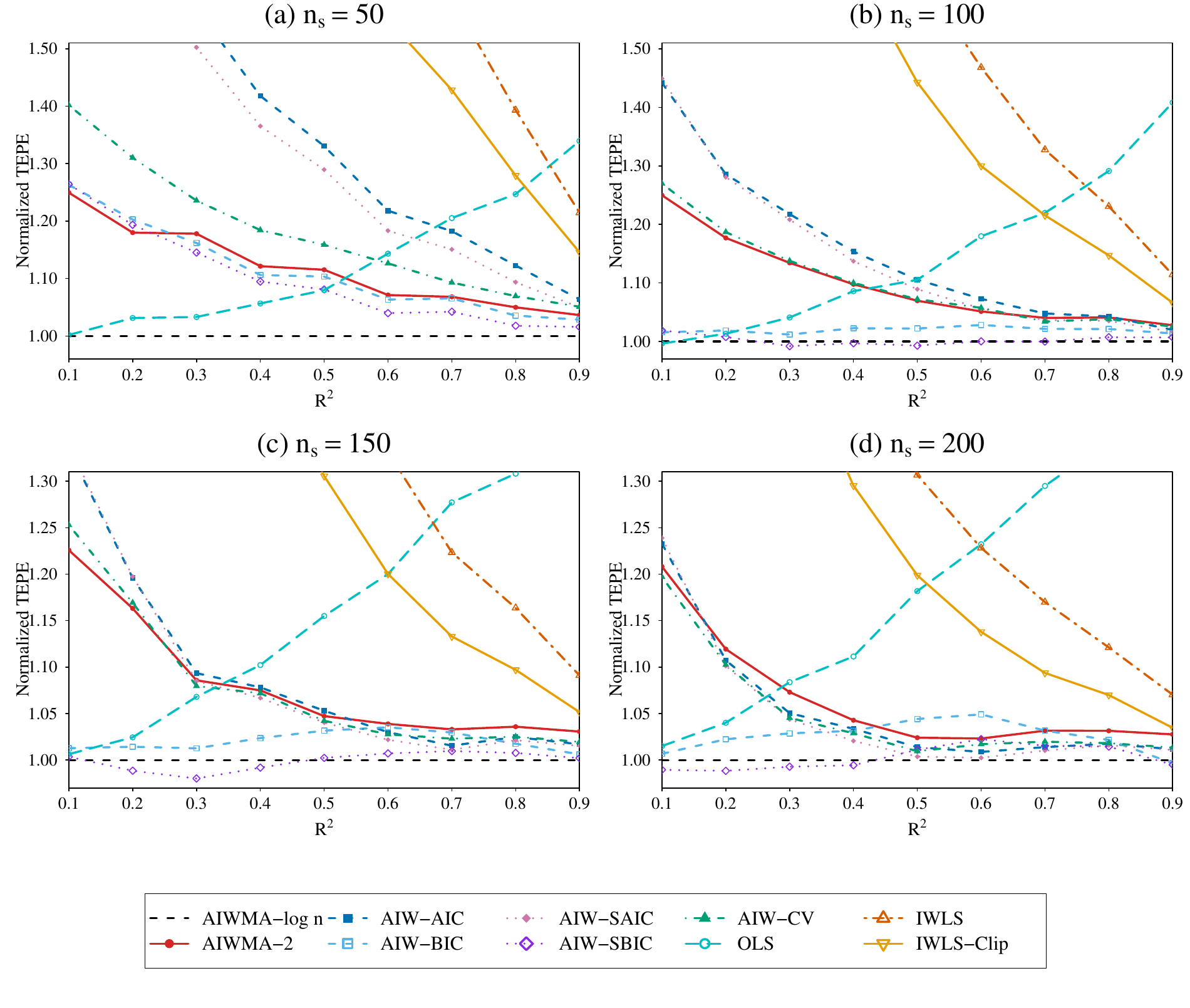}
\caption{Normalized target excess prediction error relative to AIWMA-$\log n$ in the six-dimensional misspecification design, using Gaussian plug-in density-ratio estimation and the diverging candidate grid, with $\zeta=1.2$.}
\label{fig:supp-six-plugin-diverging-zeta12}
\end{figure}

\begin{figure}[H]
\centering
\includegraphics[
width=0.90\linewidth,
height=0.90\textheight,
keepaspectratio
]{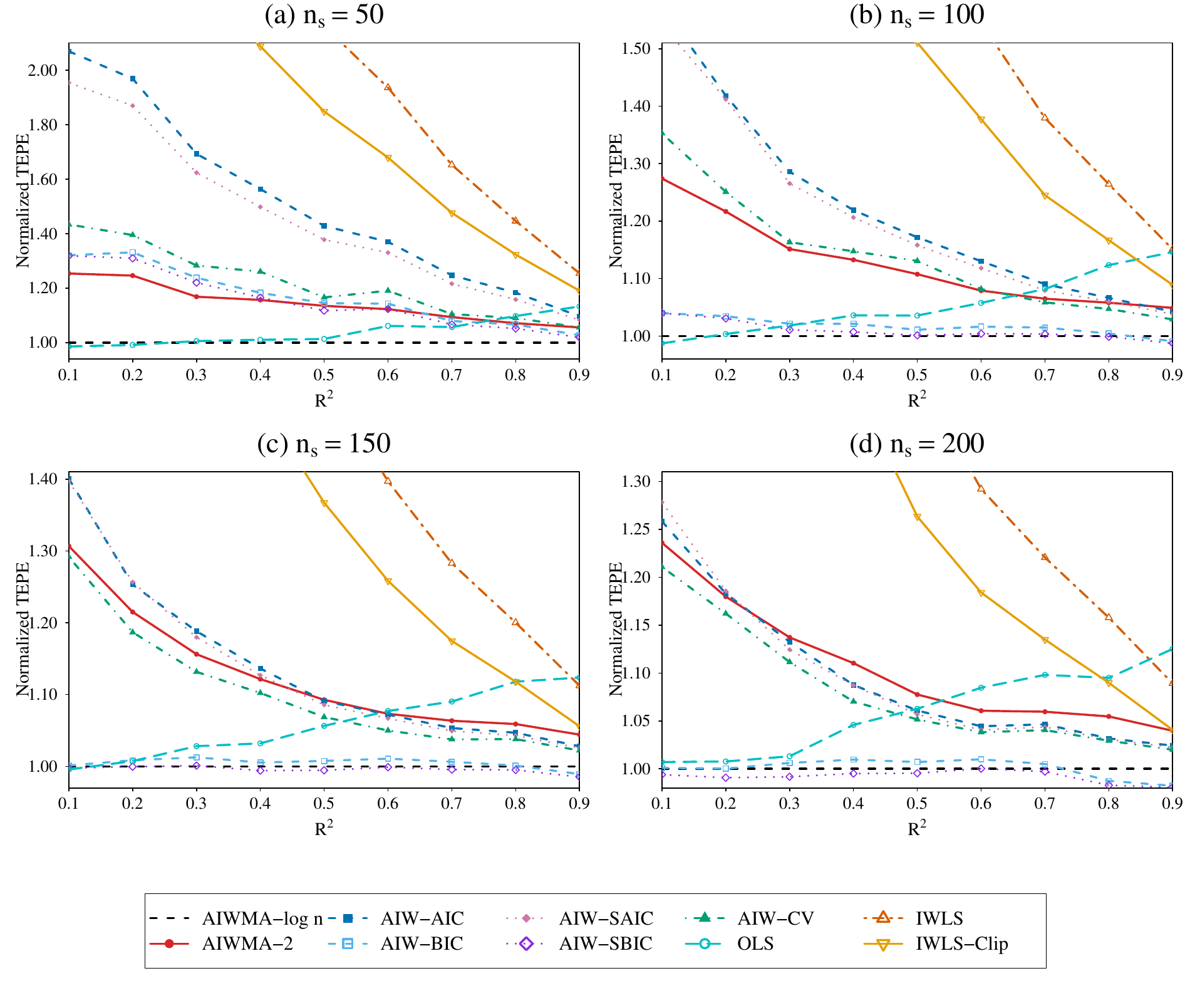}
\caption{Normalized target excess prediction error relative to AIWMA-$\log n$ in the six-dimensional misspecification design, using Gaussian plug-in density-ratio estimation and the fixed candidate grid, with $\zeta=1.5$.}
\label{fig:supp-six-plugin-fixed-zeta15}
\end{figure}

\begin{figure}[H]
\centering
\includegraphics[
width=0.90\linewidth,
height=0.90\textheight,
keepaspectratio
]{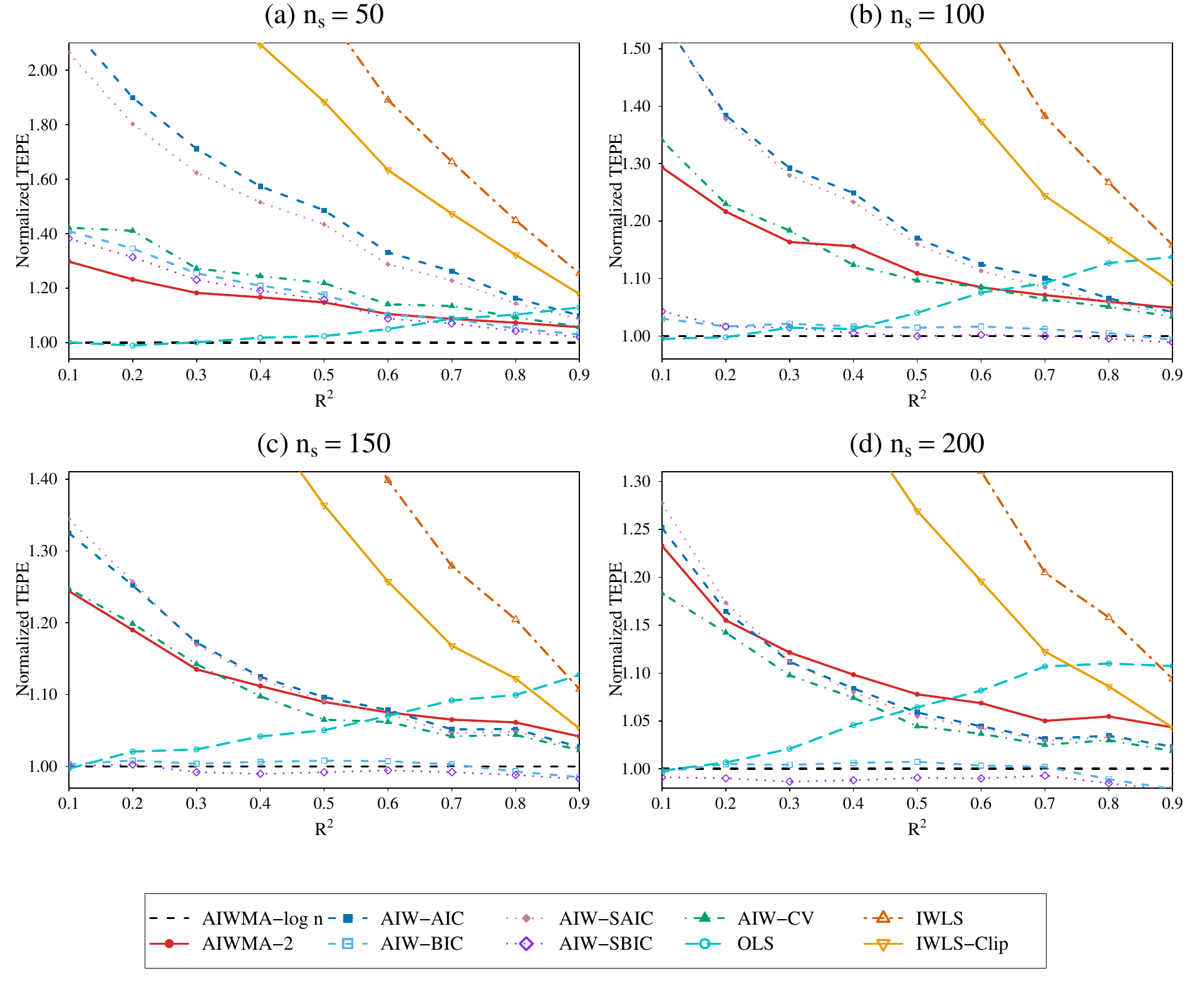}
\caption{Normalized target excess prediction error relative to AIWMA-$\log n$ in the six-dimensional misspecification design, using Gaussian plug-in density-ratio estimation and the diverging candidate grid, with $\zeta=1.5$.}
\label{fig:supp-six-plugin-diverging-zeta15}
\end{figure}

\subsection{Six-dimensional design using the oracle density ratio}
\label{subsec:supp-six-dimensional-oracle}
Figures~\ref{fig:supp-six-oracle-fixed-zeta12}--\ref{fig:supp-six-oracle-diverging-zeta15} report the normalized TEPE results obtained with the oracle density ratio under the fixed and diverging candidate grids.

\begin{figure}[H]
\centering
\includegraphics[
width=0.90\linewidth,
height=0.90\textheight,
keepaspectratio
]{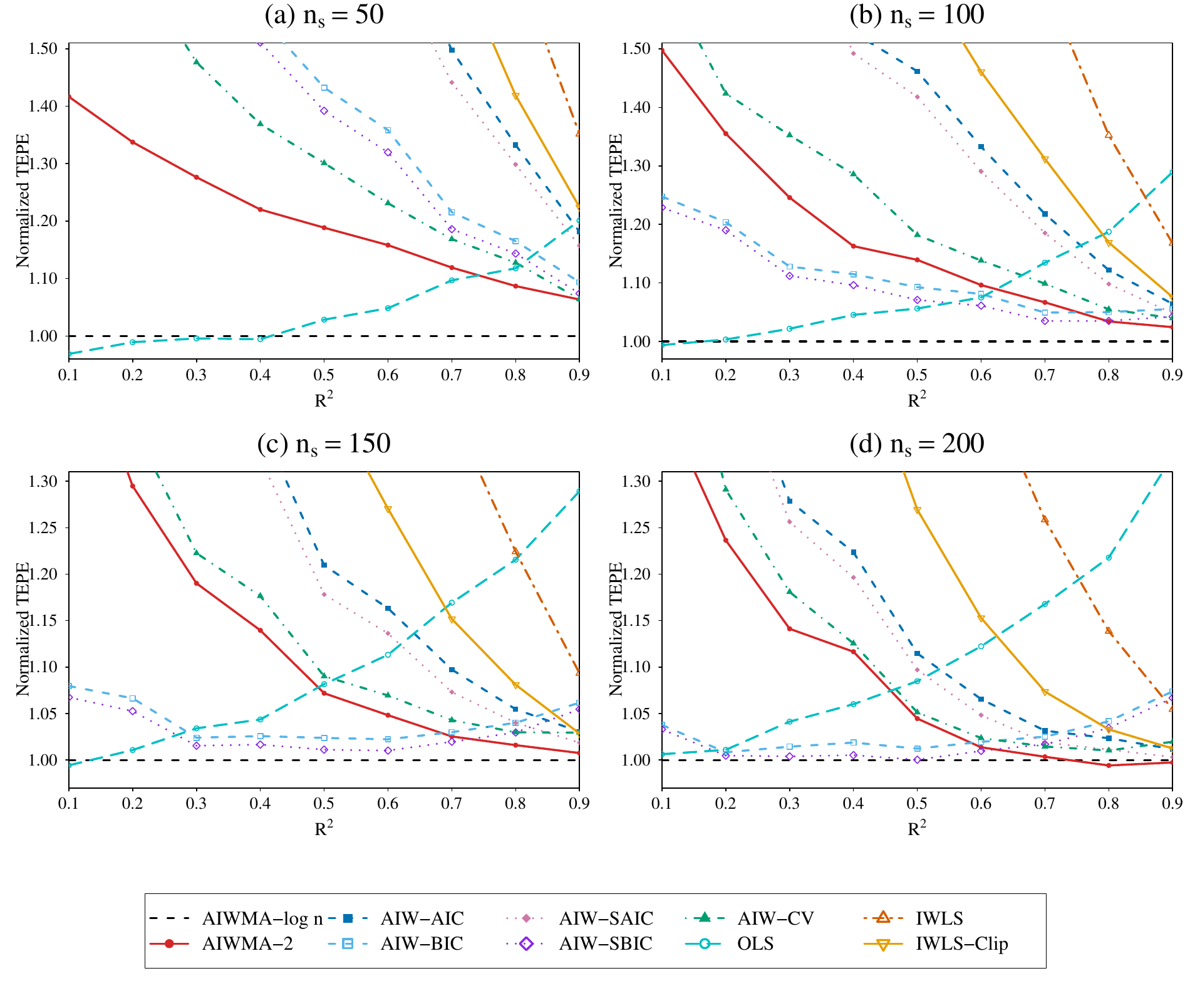}
\caption{Normalized target excess prediction error relative to AIWMA-$\log n$ in the six-dimensional misspecification design, using the oracle density ratio and the fixed candidate grid, with $\zeta=1.2$.}
\label{fig:supp-six-oracle-fixed-zeta12}
\end{figure}

\begin{figure}[H]
\centering
\includegraphics[
width=0.90\linewidth,
height=0.90\textheight,
keepaspectratio
]{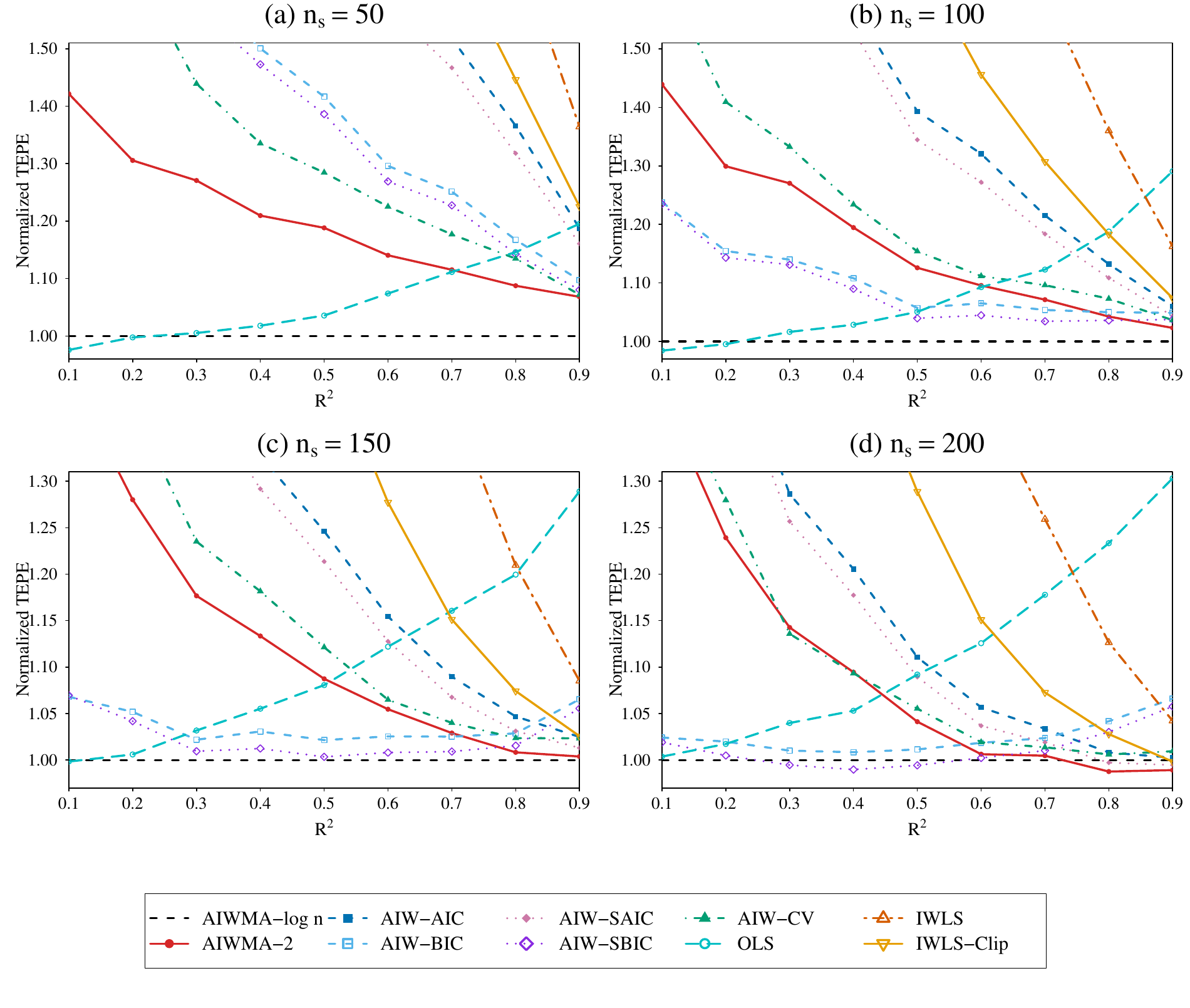}
\caption{Normalized target excess prediction error relative to AIWMA-$\log n$ in the six-dimensional misspecification design, using the oracle density ratio and the diverging candidate grid, with $\zeta=1.2$.}
\label{fig:supp-six-oracle-diverging-zeta12}
\end{figure}

\begin{figure}[H]
\centering
\includegraphics[
width=0.90\linewidth,
height=0.90\textheight,
keepaspectratio
]{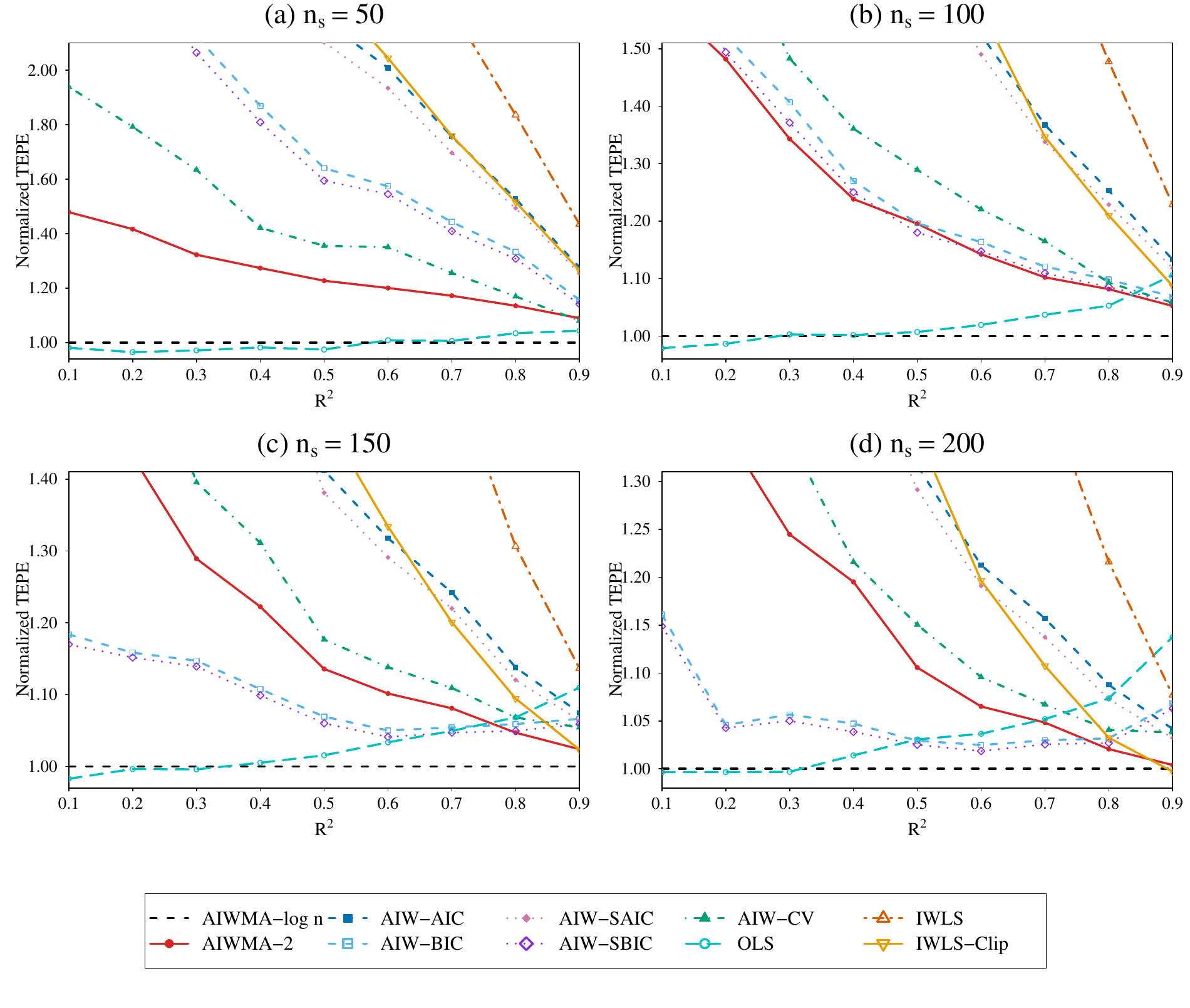}
\caption{Normalized target excess prediction error relative to AIWMA-$\log n$ in the six-dimensional misspecification design, using the oracle density ratio and the fixed candidate grid, with $\zeta=1.5$.}
\label{fig:supp-six-oracle-fixed-zeta15}
\end{figure}

\begin{figure}[H]
\centering
\includegraphics[
width=0.90\linewidth,
height=0.90\textheight,
keepaspectratio
]{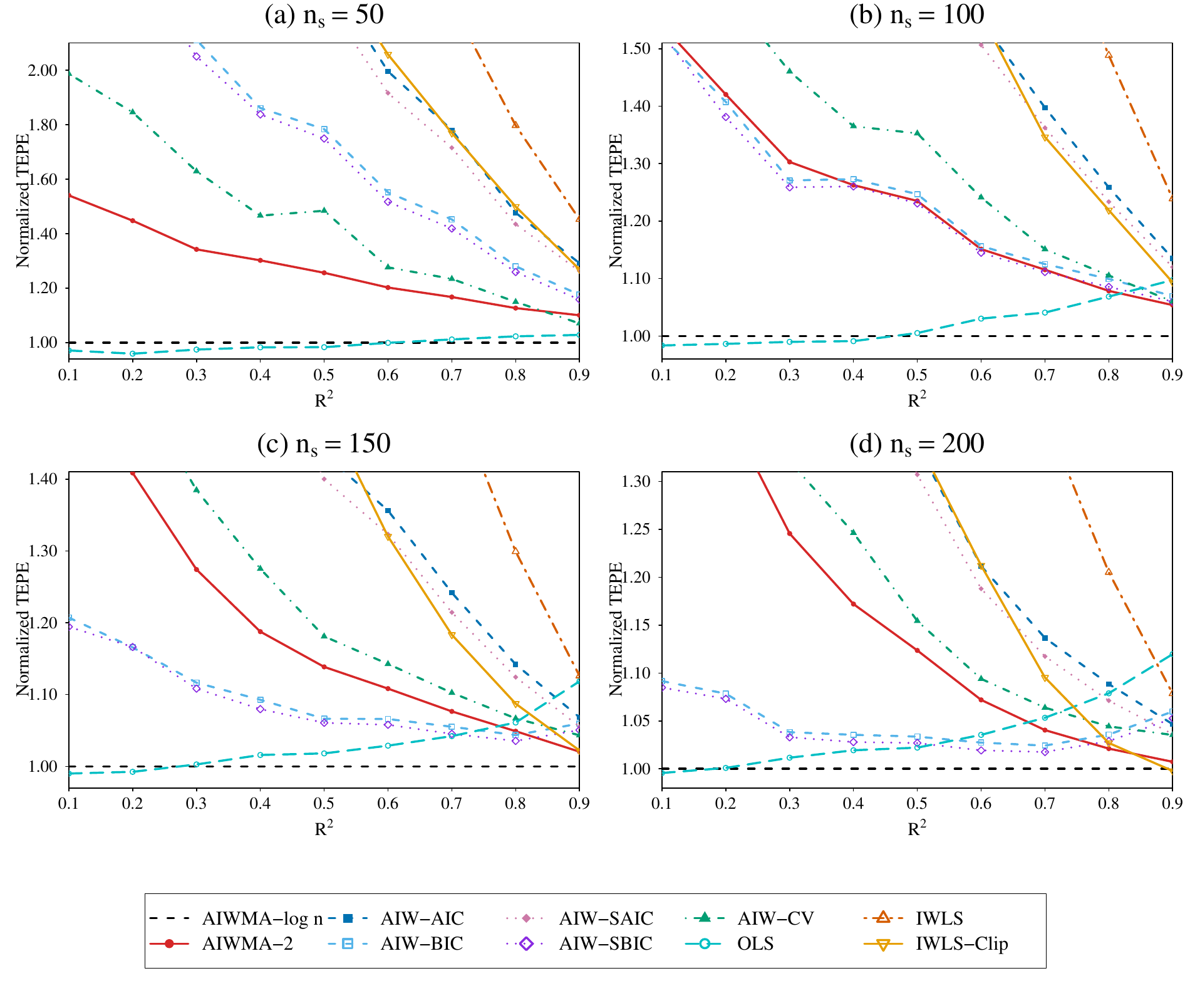}
\caption{Normalized target excess prediction error relative to AIWMA-$\log n$ in the six-dimensional misspecification design, using the oracle density ratio and the diverging candidate grid, with $\zeta=1.5$.}
\label{fig:supp-six-oracle-diverging-zeta15}
\end{figure}

\section{Additional simulations under correct specification}
\label{sec:supp-correct-specification-simulations}

Section~5.2 of the main paper reports the weight-concentration results obtained with KLIEP under both the fixed and diverging candidate grids. Figures~\ref{fig:supp-weight-concentration-kde} and~
\ref{fig:supp-weight-concentration-oracle} report the corresponding
results based on KDE and the oracle density ratio, respectively. The data-generating process, candidate grids, shift levels, target $R^2$ values, sample sizes, and number of Monte Carlo replications are unchanged.

The diagnostic is the average weight assigned to the OLS endpoint, denoted by $w_{\lambda=0}$. Under both density-ratio specifications, this average generally increases with $n_s$ across all three shift configurations and the three target $R^2$ levels. Concentration is most transparent under the oracle density ratio. These results indicate that the qualitative weight-concentration pattern is not specific to KLIEP.

\begin{figure}[H]
\centering
\includegraphics[
width=0.90\linewidth,
height=0.90\textheight,
keepaspectratio
]{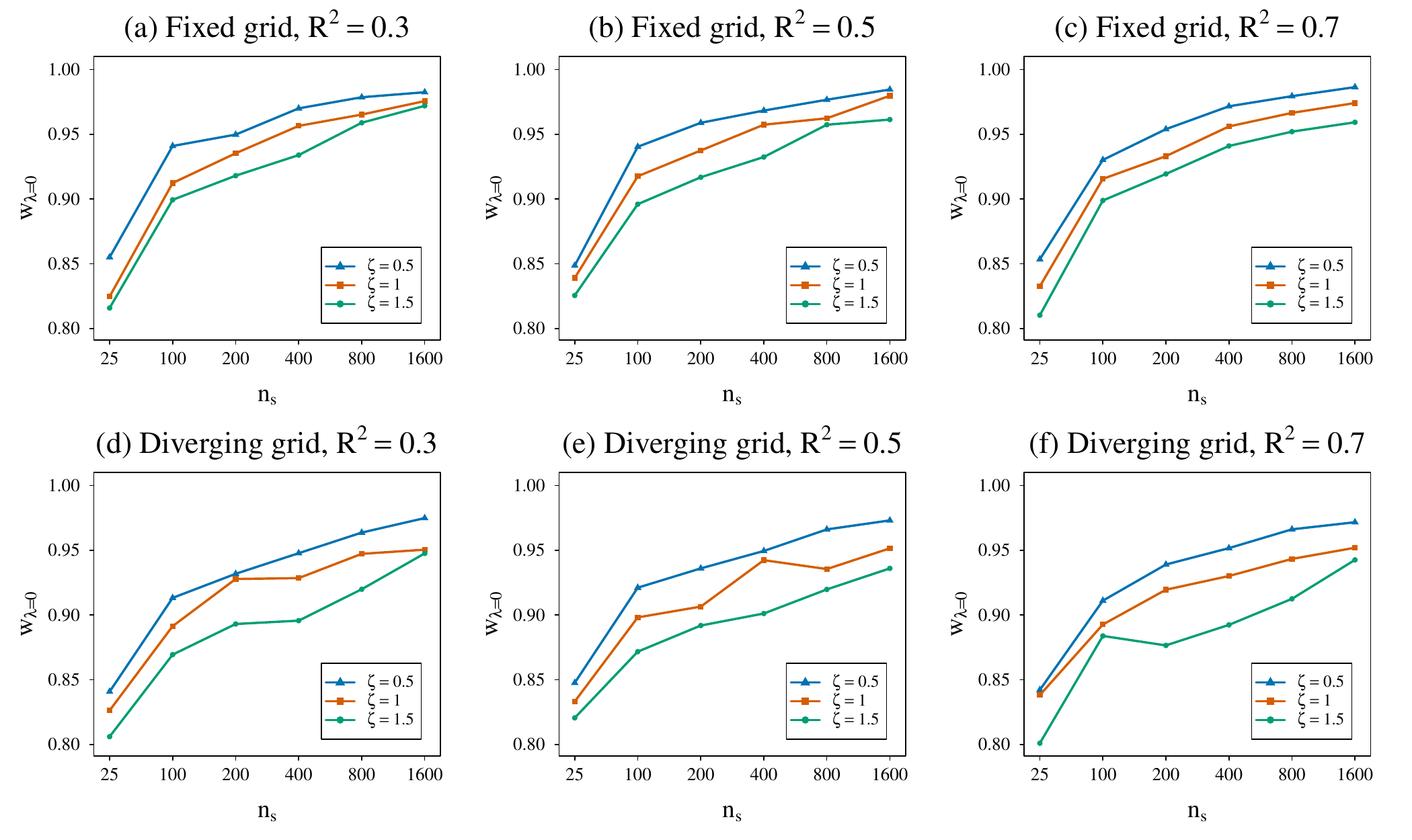}
\caption{Average weight assigned to the OLS endpoint under correct specification, using KDE density-ratio estimation. The upper and lower rows correspond to the fixed and diverging candidate grids, respectively.}
\label{fig:supp-weight-concentration-kde}
\end{figure}

\begin{figure}[H]
\centering
\includegraphics[
width=0.90\linewidth,
height=0.90\textheight,
keepaspectratio
]{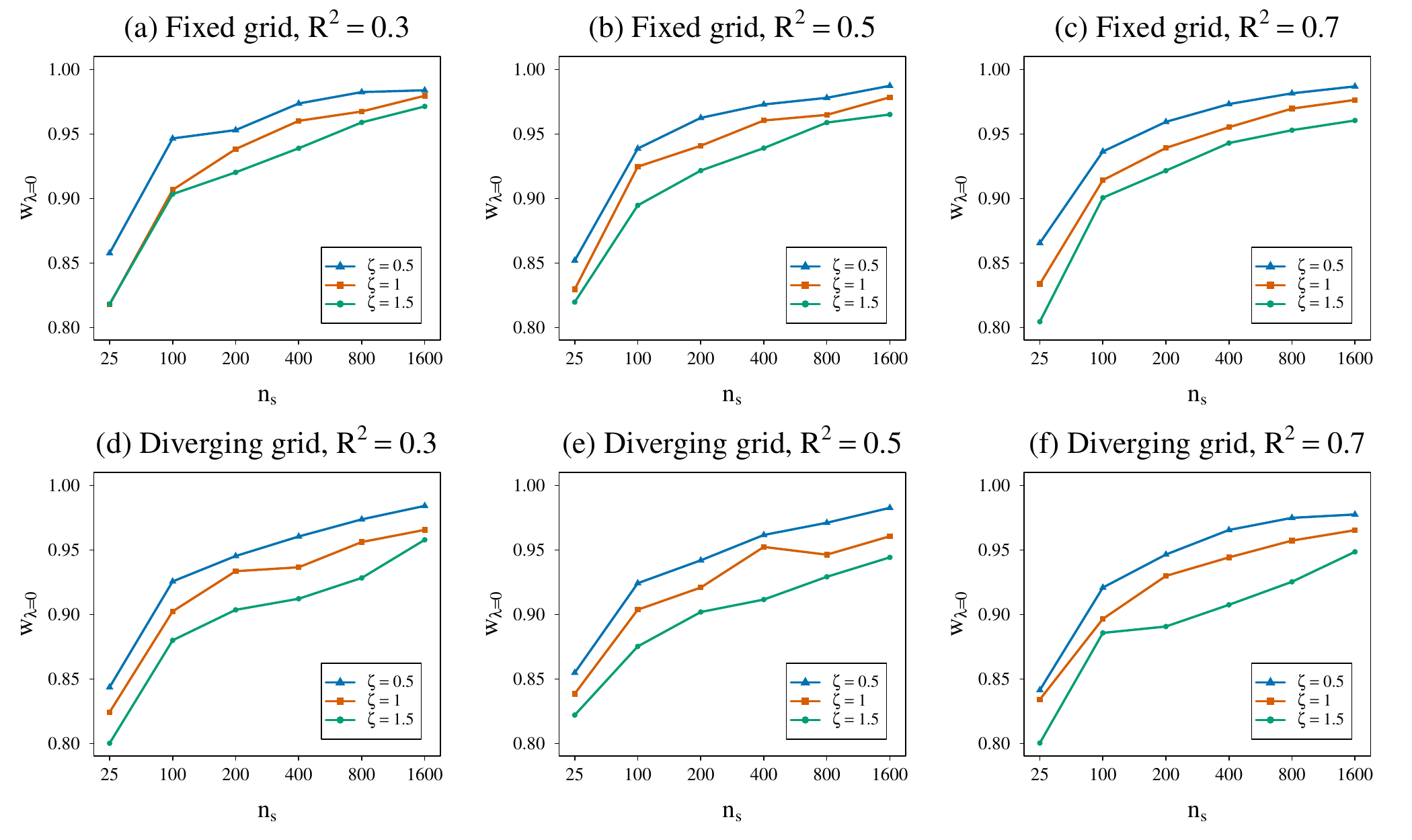}
\caption{Average weight assigned to the OLS endpoint under correct specification, using the oracle density ratio. The upper and lower rows correspond to the fixed and diverging candidate grids, respectively.}
\label{fig:supp-weight-concentration-oracle}
\end{figure}

\section{Data processing for the real-data analysis}
\label{sec:real_data_processing}

The response is infant birth weight, recorded in grams as \texttt{DBWT} and
converted to kilograms for analysis \citep{NCHS2016Natality,NCHS2024Natality}. We restrict the analysis to singleton
hospital births to U.S. residents, with maternal age between 15 and 49 years,
gestational age between 24 and 42 weeks, and birth weight between 227 and
8165 grams. Records with a missing or unknown value for any working covariate
are excluded. 
The working model contains seven continuous covariates: maternal age
(\texttt{MAGER}), prepregnancy body mass index (\texttt{BMI}), number of
prenatal visits (\texttt{PREVIS}), number of cigarettes smoked before
pregnancy (\texttt{CIG\_0}), number of cigarettes smoked during the first
trimester (\texttt{CIG\_1}), gestational weight gain (\texttt{WTGAIN}), and
the obstetric estimate of gestational age (\texttt{OEGest\_Comb}). The two
cigarette-count variables are transformed as $\log(1+c)$, where $c$ is the
reported number of cigarettes. 
The model also contains 13 categorical covariates: maternal education
(\texttt{MEDUC}), maternal race or ethnicity (\texttt{MRACEHISP}), maternal
nativity (\texttt{MBSTATE\_REC}), marital status (\texttt{DMAR}), parity
(\texttt{LBO\_REC}), timing of prenatal-care initiation
(\texttt{PRECARE5}), participation in the Special Supplemental Nutrition
Program for Women, Infants, and Children (\texttt{WIC}), prepregnancy
diabetes (\texttt{RF\_PDIAB}), gestational diabetes
(\texttt{RF\_GDIAB}), prepregnancy hypertension (\texttt{RF\_PHYPE}),
gestational hypertension (\texttt{RF\_GHYPE}), previous preterm birth
(\texttt{RF\_PPTERM}), and infant sex (\texttt{SEX}). These variables are
represented by treatment-coded indicators. After categorical expansion, the
working model contains 27 nonintercept covariates and an intercept, giving
parameter dimension $p=28$. 
The annual public-use files contain several million records. After applying
the eligibility and complete-case restrictions separately to the 2016 and
2024 files, a reproducible uniform reservoir sample of at most 300,000
records is retained from each year to form the source and target analysis
pools. Within each replication, the labeled source sample is drawn without
replacement from the 2016 pool. The unlabeled target sample and the labeled
target evaluation sample are drawn without replacement from the 2024 pool
and are mutually disjoint.

\end{document}